%% file: main.tex
\documentclass[11pt]{article}

\usepackage{amssymb,amsthm,amsmath}

\usepackage[ruled,linesnumbered]{algorithm2e}
 \DontPrintSemicolon 
\usepackage{multirow}

\SetCommentSty{mycommfont}


\usepackage{subfig}
\usepackage[usenames,dvipsnames]{xcolor}
\usepackage[colorlinks,citecolor=blue,linkcolor=BrickRed]{hyperref}\usepackage[colorlinks,citecolor=blue,linkcolor=BrickRed]{hyperref}
\usepackage{fullpage}
\usepackage{graphicx}
\usepackage{enumitem}
\usepackage{wrapfig}
\usepackage{comment}
\usepackage{cleveref}
\usepackage{float}
\usepackage{soul}
\usepackage{thmtools} 
\usepackage{thm-restate}

\graphicspath{{figures/}}

\newtheorem{theorem}{Theorem}[section]
\newtheorem{corollary}[theorem]{Corollary}

\newtheorem{lemma}[theorem]{Lemma}

\newtheorem{definition}[theorem]{Definition}

\newtheorem{claim}[theorem]{Claim}

\newif\ifFULL
\FULLtrue

\DeclareMathOperator*{\argmin}{arg\,min}

\input{macros}

\allowdisplaybreaks

\newenvironment{proofof}[1]{\smallskip\noindent{ {\bf Proof of #1.}}}%
        {\hspace*{\fill}$\Box$\par}

\title{Sparse Submodular Function Minimization}

\ifdefined\isSTOCSubmission
\author{}
\else
\author{
Andrei Graur\\
Stanford University\\
\texttt{agraur@stanford.edu}
\and 
Haotian Jiang\\
Microsoft Research, Redmond\\
\texttt{jhtdavid96@gmail.com}
\and 
Aaron Sidford\\
Stanford University\\
\texttt{sidford@stanford.edu}
}
\fi

\date{\today}

\begin{document}
\begin{titlepage}
  \maketitle

\begin{abstract}
In this paper we study the problem of minimizing a submodular function $f : 2^V \rightarrow \R$ that is guaranteed to have a $k$-sparse minimizer. We give a deterministic algorithm that computes an additive $\epsilon$-approximate minimizer of such $f$ in $\widetilde{O}(\mathsf{poly}(k) \log(|f|/\epsilon))$ parallel depth using a polynomial number of queries to an evaluation oracle of $f$, where $|f| = \max_{S \subseteq V} |f(S)|$. Further, we give a randomized algorithm that computes an exact minimizer of $f$ with high probability using $\widetilde{O}(|V| \cdot \mathsf{poly}(k))$ queries and polynomial time. When $k = \widetilde{O}(1)$, our algorithms use either nearly-constant parallel depth or a nearly-linear number of evaluation oracle queries. All previous algorithms for this problem either use $\Omega(|V|)$ parallel depth or $\Omega(|V|^2)$ queries. 

In contrast to state-of-the-art weakly-polynomial and strongly-polynomial time algorithms for SFM, our algorithms use first-order optimization methods, e.g., mirror descent and follow the regularized leader. We introduce what we call {\em sparse dual certificates}, which encode information on the structure of sparse minimizers, and both our parallel and sequential algorithms provide new algorithmic tools for allowing first-order optimization methods to efficiently compute them. Correspondingly, our algorithm does not invoke fast matrix multiplication or general linear system solvers and in this sense is more combinatorial than previous state-of-the-art methods.

\ifdefined\isSTOCSubmission
\else

\fi

\end{abstract}

\thispagestyle{empty}
\end{titlepage}

{
    \hypersetup{linkcolor=black}
    \tableofcontents
}

\input{intro_before_result}

\input{intro_part_2}

\input{approach}

\input{prelim}

\input{framework}

\input{parallel_algo}

\input{sequential_algo_dual_certificate}

\input{sequential_algo_other}

\input{extension}

\section*{Acknowledgements}

We thank Deeparnab Chakrabarty for helpful discussions and we thank the anonymous reviewers for helpful feedback. Andrei Graur was supported in part by the Nakagawa departmental fellowship award from the Management Science and Engineering Department at Stanford University, NSF CAREER Award CCF-1844855, and NSF Grant CCF-1955039. Part of work was done while Haotian Jiang was a Ph.D. student at the University of Washington and supported by a Packard fellowship. Aaron Sidford was supported in part by a Microsoft Research Faculty Fellowship, NSF CAREER Award CCF-1844855, NSF Grant CCF-1955039, a PayPal research award, and a Sloan Research Fellowship. 

\bibliographystyle{alpha}
\bibliography{bib.bib}

\appendix
\input{appendix.tex}

\end{document}

%% file: macros.tex
\newcommand{\E}{\mathbb{E}}
\newcommand{\R}{\mathbb{R}}
\newcommand{\p}{\mathbb{P}}

\newcommand{\otilde}{\widetilde{O}}

\newcommand{\poly}{\mathsf{poly}}

\newcommand{\ind}{\mathbf{1}}

\newcommand{\norm}[1]{\left\Vert #1\right\Vert}

\newcommand{\EO}{\mathsf{EO}}

\newcommand{\rfam}{\mathcal{R}}

\newcommand{\defeq}{\stackrel{\mathrm{{\scriptscriptstyle def}}}{=}}%
\newcommand{\vvar}[1]{\vec{#1}}%
\newcommand{\vones}{\vvar{1}}%
\newcommand{\gradient}{\bigtriangledown}%
\newcommand{\grad}{\gradient}%
\newcommand{\vzero}{\vvar 0}%
\newcommand{\innerProduct}[2]{\big\langle#1 , #2 \big\rangle}%
\newcommand{\onesVec}{\vvar{1}}%
\newcommand{\indicVec}[1]{\onesVec_{#1}}%

\newcommand{\sparsity}{k}
\renewcommand{\Pr}{\mathbb{P}}

\newcommand{\lovasz}{{Lov\'{a}sz} }
\newcommand{\ext}{\hat{f}}

\newcommand{\optlb}{\phi}
\newcommand{\vsampling}{$\mathsf{vSampling}$}

%% file: intro_before_result.tex
\newpage
\section{Introduction}
\label{sec:intro_pre1.1}

Submodular function minimization (SFM) is a foundational problem in combinatorial optimization.
Submodular functions encompass a wide range of functions that appear naturally in practical applications, including graph cut functions, matroid rank functions, set coverage functions, and utility functions from economics. 
Since seminal work of Edmonds in 1970 \cite{E70}, SFM has served as a central tool in many areas such as theoretical computer science, operations research, game theory, and recently, machine learning. 
We refer interested readers to surveys \cite{M05,I08} for a more comprehensive account of the rich history of SFM. 

Throughout this paper we consider a standard setting for SFM. We are given a set function $f: 2^V \rightarrow \R$, where $V$ is an $n$-element finite set, known as the \emph{ground set}, and $f$ is submodular, i.e., 
\[
f(S \cup \{i\}) - f(S) \geq f(T \cup \{i\}) - f(T)\text{ for all }S \subseteq T \subseteq V\text{ with }i \notin T .
\]
Furthermore, we assume that $f$ is accessed only through an \emph{evaluation oracle} which when queried at any $S \subseteq V$ outputs $f(S)$ in time $\EO$.
We let $|f| \defeq \max_{S \subseteq V} |f(S)|$ and $f^* \defeq \min_{S \subseteq V} f(S)$ and consider the problem of computing an \emph{$\epsilon$-approximate minimizer}, i.e., $S \subseteq V$ with $f(S) \leq f^* + \epsilon$.

Since seminal work of Gr\"otschel, Lov\'asz, and Schrijver \cite{GLS81} showed that SFM can be solved in polynomial time, there have been multiple advances in SFM over the last few decades \cite{S00,IFF01,FI03,I03,V03,O09,IO09,CJK14,LJ15,CLSW17}. In this paper, we focus on algorithms that solve SFM to \emph{high accuracy} with a \emph{polynomial query complexity}, meaning that they solve the problem with a number of queries to an evaluation oracle that scale \emph{weakly-polynomially} ($\poly(n, \log(|f|/\epsilon)$)) \cite{GLS81} or \emph{strongly-polynomially} ($\poly(n)$) \cite{GLS84,GLS88}.\footnote{When $f$ is integer valued, any $\epsilon < 1$ approximate solution is optimal; a variety of the prior work consider only this setting. Throughout the paper we do not distinguish between prior work which consider exactly solving SFM integer valued $f$ (with a dependence on $|f|$) and those that work in the more general setting we consider in this paper.} 
Current state-of-the-art SFM algorithms in these regimes are weakly-polynomial  $\otilde(n^2 \log (n |f|/\epsilon))$-query, polynomial-time algorithms \cite{KTE88,NN89,V89,BV04,LSW15,JLSW20}, strongly-polynomial $\otilde(n^3)$-query, polynomial-time algorithms \cite{LSW15,DVZ21,J22}, and a strongly-polynomial $\otilde(n^2)$-query, exponential-time algorithm \cite{J22} (see \Cref{sec:intro:related_work}  
for more details)\footnote{Throughout the paper we use $\otilde(\cdot)$ to hide $O(\poly(\log n))$ factors.}.

On the hardness side, however, the current state-of-the-art lower bounds 
exclude algorithms making fewer than $\Omega(n \log n)$ queries in the strongly-polynomial regimes \cite{CGJS22} and fewer than $\Omega(n)$ queries in the weakly-polynomial regime \cite{H08,CLSW17}.
Consequently, there are large, $\Omega(n)$ gaps, between these lower bounds and the best known upper bounds. 
Unfortunately, obtaining nearly-linear (or provably near-optimal) query complexity algorithms for SFM has been elusive.  

In light of these developments, it is natural to ask, what additional structural assumptions may be needed to enable faster algorithms? One recent line of work has explored the complexity of \emph{decomposable SFM} \cite{JBS13, NJJ14, EN15, ENV17, KBP19, AKMSV21,DJLPY22}, that is the special case where $f(S) = \sum_{i} f_i(S \cap T_i)$ for submodular $f_i$ and sparse $T_i$ given an oracle for evaluating the individual $f_i$ over $T_i$. A different line of work \cite{CLSW17,ALS20} considers the complexity of \emph{approximate} SFM when the minimizer is $k$-sparse, which we refer to as \emph{$k$-sparse SFM} for brevity.\footnote{This problem is distinct from that of computing the minimum value $\sparsity$-sparse set for a submodular function. 
}  We refer to an SFM algorithm as approximate, if its query complexity is \emph{pseudo-polynomial}, i.e., $O(\poly(n,|f|/\epsilon))$. The state-of-the-art approximate $\sparsity$-sparse SFM algorithm has a query complexity of $\widetilde{O}(k (|f|/\epsilon)^2)$, when $f$ is integer valued and $\epsilon < 1$.

In both of these cases, sparsity plays a prominent role. 
In the specific context of SFM, while various polyhedral and geometric properties of submodular functions have been extensively studied and heavily exploited since the 1970s \cite{E70}, these properties are mostly global, involving the entire set $V$ altogether. On the other hand, assuming $k$-sparsity of the minimizer allows one to take a glimpse into local properties of submodularity, e.g., to understand the role a small number of elements play for the minimization of the function. 

Moreover, sparsity of the minimizer is a natural assumption in convex optimization and submodular function minimization problems. In particular, sparsity arises in signal processing, feature selection, compressed sensing, etc.\ where the solution is often expected to be sparse, i.e., have a small number of non-zero elements \cite{donoho2006compressed, marvasti2012unified, li2020survey}. Sparsity is also common in cases where a regularizer is added to the objective function to encourage sparsity. 
One example of such a setup is the problem of finding an optimal dataset for speech recognition tasks \cite{LB11}. This problem can be written as $f(S) + \lambda |S|$, where $f$ is a submodular objective, and therefore it is expected that the size of the minimizing set is much smaller than the ground set for large values of the regularization coefficient $\lambda$.  
Consequently, understanding how the complexity of algorithms depends on the sparsity leads to better insight into more refined combinatorial and geometric structures of the problems. 
Therefore, the central question we ask in this paper is:
\begin{center}
    \emph{Can we leverage sparsity to improve upon state-of-the-art polynomial query complexities?}
\end{center}

$\sparsity$-sparse SFM is also interesting in light of recent work \cite{BS20} seeking to clarify the {\em parallel depth} of SFM, i.e., the number of parallel rounds of queries to the evaluation oracle required for a query-efficient algorithm. The state-of-the-art parallel depth lower bounds are $\Omega(n/\log n)$ in the strongly-polynomial regime \cite{CGJS22}, which matches the upper bound in \cite{J22} up to a factor of $\log^2 n$, and $\widetilde{\Omega}(n^{1/3})$ in the weakly-polynomial regime \cite{CCK21}. 
These polynomial parallel depth lower bounds crucially rely on the minimizers being dense for the constructed submodular functions,  
and highly parallel algorithms 
might be possible when the submodular function admits a sparse minimizer.  
Therefore, we also ask: \emph{Can we improve the parallel complexities for $k$-sparse SFM?} 
Besides being interesting from an algorithmic perspective, obtaining improved parallel algorithms for $\sparsity$-sparse SFM could aid lower bound development by showing how hard-instances for lower bounds must have dense minimizers.

\subsection{Challenges and Additional Motivations}
\label{subsec:challenges_motivations}

Beyond intrinsic interest in improving the complexity of $\sparsity$-sparse SFM, 
this problem is also an interesting testbed for new techniques and a number of larger open problems on SFM. Here we briefly elaborate on these challenges and motivations for studying $\sparsity$-sparse SFM.

State-of-the-art SFM algorithms typically leverage the \emph{\lovasz extension} \cite{L83} of $f$, a convex  function $\ext : [0,1]^V \rightarrow \R$  
that agrees with $f$ on the hypercube's vertices, i.e., $\ext(\indicVec{S}) = f(S)$ for all $S \subseteq V$. It is known that $\ext$ can be evaluated efficiently and minimizing $\ext$ suffices for SFM (see \Cref{sec:prelim}).
Consquently, SFM algorithms can be readily obtained by applying convex optimization methods to the \lovasz extension. Indeed, state-of-the-art weakly-polynomial SFM algorithms \cite{LSW15,JLSW20} follow this approach by using {\em cutting plane methods}, a class of weakly-polynomial convex optimization methods, to obtain $\epsilon$-approximate minimizers in $\otilde(n \log(1/\epsilon))$ parallel rounds of $\otilde(n)$ queries per round.  
State-of-the-art strongly-polynomial SFM algorithms \cite{LSW15,DVZ21,J22} carefully apply these weakly-polynomial cutting plane methods iteratively.

With the $\sparsity$-sparsity assumption on the solutions, a  natural approach would be to apply these continuous optimization methods to minimize $\ext$ over $S_k^V \defeq \Delta_k^V \cap [0,1]^V$,
where $\Delta_k^V  \defeq \{x \in \R^V_{\geq 0} | \norm{x}_1 \leq k\}$ is the interior of the simplex scaled up by $k$; this suffices for $\sparsity$-sparse SFM since $\indicVec{S^*} \in S_k^V$ for the $\sparsity$-sparse minimizer $S^* \subseteq V$. 
Unfortunately, while changing the domain from $[0,1]^V$ to $S_k^V$ is known to improve the performance of certain pseudo-polynomial convex optimization methods (as in  \cite{CLSW17,ALS20}), it is not known to improve the performance of weakly-polynomial convex optimization algorithms (e.g., state-of-the-art cutting plane method \cite{JLSW20}) by more than logarithmic factors. Furthermore, without using more of the structure $\ext$ it seems unlikely that this change of domain would affect the weakly-polynomial complexity by more than logarithmic factors, since one could scale a hard convex optimization problem to fit inside $S_k^V$ without changing problem parameters by more than a polynomial factor.  

These challenges call for the development of new optimization techniques that better utilize structures of the \lovasz extension and sparsity of the domain, which might lead to applications for a broader range of open problems on SFM. We note several of these additional motivations below.

\paragraph{Strongly-polynomial time $O(n^{3-c})$-query algorithm for SFM.} One of the most important motivations is towards improving strongly-polynomial time SFM algorithms. 
The current best query complexity here is $O(n^3 \log\log n / \log n)$ given in \cite{J22}, but this approach seems unlikely to provide further improvement given the stagnation of progress on obtaining a better approximation factor for the shortest vector problem, on which the algorithm in \cite{J22} crucially relies. 

Other state-of-the-art strongly-polynomial time SFM algorithms with $\widetilde{O}(n^3)$ query complexities in \cite{LSW15,DVZ21} 
learn precedence 
constraints of the form, if $p \in V$ is in a minimizer then so is $q$ (e.g., \cite{IFF01,IO09,LSW15,DVZ21}). In the worst case, these algorithms might make $\otilde(n^2)$ queries to learn only a single coordinate that must be in a minimizer (or not), or for many coordinates $p \in V$ a single $q \in V$ that must be in any minimizer containing $p$. This worst-case behavior is a key barrier towards obtaining strongly-polynomial time algorithms with $O(n^{3-c})$ query complexities for constant $c > 0$.
However, this worst-case behavior is sparse, and $\sparsity$-sparse SFM algorithms which better exploit local properties of submodular functions might be useful to get around the aforementioned barrier in this case and lead to a smaller query complexity. 

\paragraph{SFM versus continuous optimization.} Given the challenges of adapting 
weakly-polynomial convex optimization algorithms to leverage sparsity, obtaining weakly- and strongly-polynomial algorithms for $\sparsity$-sparse SFM could highlight differences between general convex optimization and SFM. Consequently, $\sparsity$-sparse SFM is a natural proving grounds for designing SFM algorithms that go beyond using the boundedness and convexity of the \lovasz extension. 

\paragraph{Combinatorial algorithms and iteration costs.} The use of cutting plane methods in state-of-the-art SFM algorithms comes with certain inherent costs. Key among them is that all known cutting plane methods apply general linear system solvers or matrix multiplication methods, making these methods somewhat intrinsically non-combinatorial. This is inherent as ultimately the problems they solve are more general than that of solving arbitrary linear systems. 

Since, as argued above, obtaining better query complexities for weakly- and strongly-polynomial $\sparsity$-sparse SFM suggests departing from cutting plane methods, the problem could be an interesting one to see where more combinatorial methods or ones with lower iteration costs can shine. State-of-the-art pseudo-polynomial SFM algorithms leverage optimization machinery which does not use linear system solves and correspondingly have runtimes that are within polylogarithmic factors of their query complexity \cite{CLSW17,ALS20}. Though there have been efforts in using alternative optimization methods to solve SFM, e.g., \cite{DVZ21}, the query complexities of such methods are much higher than the state-of-the-art. Correspondingly, $\sparsity$-sparse SFM is an interesting setting to see whether such methods can outperform cutting plane methods.

%% file: intro_part_2.tex
\subsection{Our Results}
\label{sec:intro:results}

Our main results include two algorithms which improve, respectively, the parallel depth and query complexities of polynomial-time $\sparsity$-sparse SFM algorithms.

\paragraph{Parallel depth for $k$-sparse SFM.} In the parallel model for SFM (in the weakly-polynomial regime), the algorithm can submit up to $\poly(n, \log(|f|/\epsilon))$ parallel queries to the evaluation oracle in each round, and its parallel \emph{depth} is defined to be the number of rounds needed to find the minimizer in the worst case. 
Our main result for this model is the following theorem. 

\begin{restatable}[Parallel $k$-sparse SFM]{theorem}{parallel}
\label{thm:main_parallel}
There is a deterministic parallel algorithm for $k$-sparse SFM with parallel depth $\widetilde{O}(k^7 \cdot \log(|f|/\epsilon))$ and runtime $\widetilde{O}(n^2 \cdot k^7\log(|f|/\epsilon) \cdot \EO + \poly(n) \cdot \log(|f|/\epsilon))$. 
\end{restatable}

When the sparsity $k = \widetilde{O}(1)$, the parallel depth in \Cref{thm:main_parallel} is $\widetilde{O}(1)$. To the best of our knowledge, this is the first nearly-constant parallel depth result for SFM, 
beyond the trivial $n^\sparsity$-query algorithm that queries all $\sparsity$-sparse sets in a single round (which does not have polynomial query complexity whenever $\sparsity = \omega(1)$).

Our result is in stark contrast to the best known weakly-polynomial parallel depth of $\widetilde{O}(n)$ for general SFM \cite{LSW15}. 
It is important to emphasize here that $\widetilde{O}(1)$-sparsity is also {\em necessary} for obtaining a nearly-constant parallel depth.
The work of \cite{CCK21} implies that $\widetilde{\Omega}(k^{1/3})$ parallel depth is required for any weakly-polynomial algorithm for $\sparsity$-sparse SFM.

\paragraph{Query complexity for $k$-sparse SFM.} While the algorithm in \Cref{thm:main_parallel} achieves a nearly-constant parallel depth when the sparsity is nearly-constant, even in this setting its query complexity is $\Omega(n^2)$. In light of the question of designing SFM algorithms with nearly-linear query complexity, 
our second main result is a pair of algorithms which improve the weakly- and strongly-polynomial query complexities for $\sparsity$-sparse SFM.  (It remains open as to whether the parallel depth of strongly-polynomial $\sparsity$-sparse SFM can be similarly improved.)

\begin{restatable}[Weakly-polynomial $k$-sparse SFM]{theorem}{weaklysequential}
\label{thm:main_sequential}
There is a randomized algorithm that outputs an $\epsilon$-approximate minimizer
for $k$-sparse SFM whp.\ in $\widetilde{O}((n \cdot \poly(k) \cdot \EO + \poly(n)) \log(|f|/\epsilon))$ time. 
\end{restatable}

\begin{restatable}[Strongly-polynomial $k$-sparse SFM]{theorem}{stronglysequential}
\label{thm:main_strongly_sequential}
There is a randomized algorithm that outputs an exact minimizer for $k$-sparse SFM whp.\ in $\widetilde{O}(n \cdot \poly(k) \cdot \EO + \poly(n))$ time. 
\end{restatable}

We include both theorems above because the $\poly(k)$ in \Cref{thm:main_sequential} is slightly better than that in \Cref{thm:main_strongly_sequential} (see \Cref{subs_stcorrectness}). 
The algorithms in \Cref{thm:main_sequential,thm:main_strongly_sequential} have nearly-linear query complexities when the sparsity $k = \widetilde{O}(1)$. Previously, the only nearly-linear weakly-polynomial query complexity results for SFM were obtained when the submodular function $f$ can be decomposed as 
$f(S) = \sum_{i} f_i(S)$ and each $f_i$ depends only on $\widetilde{O}(1)$ coordinates \cite{AKMSV21,DJLPY22}. However, this is different and the techniques for solving it seem tailored to its structure.

Our algorithms for Theorems \ref{thm:main_parallel}-\ref{thm:main_strongly_sequential} depart from the use of cutting plane methods and do not rely on linear system solves as a sub-procedure. In this sense, they are more combinatorial than state-of-the-art 
weakly-polynomial time \cite{LSW15,JLSW20} and strongly-polynomial time SFM algorithms \cite{LSW15,DVZ21,J22}. 
Somewhat surprisingly, our algorithms combine first-order methods, which have been primarily used for pseudo-polynomial SFM algorithms (e.g., \cite{CLSW17,ALS20}), and arc finding, a technique central to many strongly-polynomial SFM algorithms (e.g., \cite{LSW15,DVZ21}), to obtain very efficient weakly- and strongly-polynomial time algorithms.
Previous combination of these two techniques only appeared in \cite{DVZ21}, but the resulting algorithm has
query complexity and parallel depth at least a factor of $n^2$ larger than the state-of-the-art
algorithms based on cutting plane methods.
The proofs of Theorems \ref{thm:main_sequential} and \ref{thm:main_strongly_sequential} additionally invoke various sampling techniques,
which crucially allows us to save the additional factor of $n$ from querying an entire subgradient of the \lovasz extension in each iteration.  

\subsection{Related Work}
\label{sec:intro:related_work}

SFM is a central combinatorial optimization problem with extensive applications.
The problem of maximizing a submodular function has also been widely studied, but is very different and has seemingly different structure, algorithms, and history (see, e.g., \cite{KG14} for a survey on this topic).

\paragraph{Strongly-, weakly-, and pseudo- polynomial algorithms for SFM.} As discussed in the intro, a fundamental result for SFM is that it can be solved efficiently, in all three regimes of weakly-, strongly-, and pseudo-polynomial. The first weakly- and strongly-polynomial time SFM algorithms were given in the seminal work of Gr\"otschel, Lov\'asz, and Schrijver \cite{GLS81,GLS84,GLS88}.
The first pseudo-polynomial algorithm for SFM was given in a seminal work of Cunningham \cite{C85}. 
Since then, there has been a long line of work on designing better algorithms for SFM in all three regimes \cite{S00,IFF01,FI03,I03,V03,O09,IO09,CJK14,LJ15,LSW15,CLSW17,DVZ21}. The state-of-the-art algorithms for these regimes are shown in  \Cref{table:sfm_algo}.

\begin{table}[htp!]
\centering
\begin{tabular}{|c|c|c|c|}
\hline 
\textbf{Paper} & 
\textbf{Year} & 
\textbf{Running Times} & 
\textbf{Remarks}\tabularnewline
\hline 
\hline 
\multirow{2}{*}{\cite{JLSW20}}  & \multirow{2}{*}{2020} & $O(n^{2}\log nM\cdot \EO+n^{3}\log nM)$ & current best weakly \& \\
 & &
$O(n^{3}\log^{2}n\cdot\EO+n^{4}\log^2 n)$ & strongly runtime
\tabularnewline
\hline

\multirow{2}{*}{\cite{ALS20}} & \multirow{2}{*}{2020} & $\widetilde{O}(n M^2 \cdot \EO + \poly(n))$ & current best pseudo-poly \\
& & $\widetilde{O}(k M^2 \cdot \EO + \poly(n))$  &  current best sparse pseudo-poly
\tabularnewline \hline
\multirow{2}{*}{\cite{J22}} & \multirow{2}{*}{2021} &
$O(n^3 \log\log n/\log n \cdot \EO + \poly(n))$ & current best strongly \\
& & $O(n^2 \log n \cdot \EO + \exp(n))$ & query complexity \\ 
\hline
\end{tabular}
\protect\caption{\label{table:sfm_algo} State-of-the-art weakly-, strongly-, and pseudo-polynomial algorithms for submodular function minimization. $k$ is the sparsity and parameter $M = |f|/\epsilon$.
}
\end{table}

\paragraph{Parallel SFM.}
For the parallel complexity of SFM discussed earlier in the intro, the current best weakly-polynomial algorithm has parallel depth $O(n \log nM)$ \cite{LSW15} and the current best strongly-polynomial algorithms \cite{J22} have parallel depth $O(n \log n)$ (with exponential runtime) or $O(n^2 \log \log n/ \log n)$ (with polynomial runtime). 
In concurrent work \cite{CGJS23}, a superset of the authors give a $\otilde(n^{1/3}/\epsilon^{2/3})$-round $\poly(n)$-time algorithm for obtaining an $\epsilon$-approximate minimizer, and a $2$-round $n^{O(M)}$-time algorithm for computing an exact minimizer. 
As discussed in the intro, lower bounds for parallel SFM have also been studied recently (see \Cref{table:sfm_parallel_lb}). 

\begin{table}[htp!]
\centering
\begin{tabular}{|c|c|c|c|}
\hline 
\textbf{Paper} & \textbf{Year} & \textbf{Parallel Depth} & {\bf Accuracy}
\tabularnewline
\hline 
\cite{BS20} & 2020 & $\Omega(\log n / \log \log n)$ & exact \\ 
\hline
\cite{CCK21} & 2021 & $\widetilde{\Omega}(n^{1/3})$ & $|f|/\poly(n)$ \\ 
\hline
\cite{CGJS22} & 2022 & $\Omega(n/ \log n)$ & exact \\ 
\hline
\end{tabular}
\protect\caption{\label{table:sfm_parallel_lb} Parallel depth lower bounds for query-efficient SFM. In the ``Accuracy'' column, ``exact'' means the algorithm is required to compute an exact minimizer, and ``$|f|/\poly(n)$'' means the algorithm is allowed to output any approximate minimizer with an additive accuracy of $|f|/\poly(n)$. }

\end{table}

\paragraph{Structured SFM.}
Given the aforementioned nearly $n$-factor gap between the state-of-the-art query complexity upper and lower bounds for SFM, there have been 
exciting recent results on improving the query complexity of SFM assuming more fine-grained structures of the submodular functions. In particular, for the problem of decomposable SFM discussed prior to \Cref{subsec:challenges_motivations},
it is known that $f$ can be minimized in weakly-polynomial time using $\widetilde{O}(n)$ total queries to the evaluation oracles of each individual $f_i$ \cite{AKMSV21, DJLPY22}.

\subsection{Paper Organization}
We start by providing an overview of our approach to obtaining our results in \Cref{sec:approach}, followed by preliminaries in \Cref{sec:prelim}. 
Our unified algorithmic framework and key algorithmic tools for obtaining both our parallel and sequential results are introduced in \Cref{sec:framework}. Our parallel results are presented in \Cref{sec:parallel} and our sequential results are obtained in \Cref{sec:sequential}.

%% file: approach.tex
\section{Our Approach}
\label{sec:approach}

Here we provide an overview of our approach towards proving Theorems \ref{thm:main_parallel}-\ref{thm:main_strongly_sequential}. We first give some context and motivation, and then we cover the key components of our approach in Sections \ref{sec:approach:framework}-\ref{sec:approach:sequential}.

To situate our approach, recall that previous state-of-the-art weakly- and strongly-polynomial time SFM algorithms all apply the general continuous optimization tool of \emph{cutting plane methods} \cite{L65,N65,S77,YN76,K80,KTE88,NN89,V89,BV04,LSW15,JLSW20}. Cutting plane methods are known to compute $\epsilon$-approximate minimizers of bounded convex functions on $\mathbb{R}^n$ in $\otilde(n \log(1/\epsilon))$ iterations where each iteration consists of a subgradient computation, which typically takes $\otilde(1)$ depth, $O(n)$ queries to the evaluation oracle of $f$, and $\Omega(n^2)$ additional work \cite{JLSW20} involving linear algebraic operations such as a linear system solve.

In this paper we seek to improve upon these methods for $\sparsity$-sparse SFM both in terms of performance and to avoid general linear algebraic primitives (to obtain, in some sense, a more combinatorial algorithm). 
However, as discussed in \Cref{subsec:challenges_motivations}, it is unclear how to substantially improve cutting plane methods just using the assumption that their is a sparse optimal solution.

Consequently, 
we depart from previous state-of-the-art weakly- and strongly-polynomial SFM algorithms and instead use first-order methods\footnote{Technically speaking, cutting-plane methods also only use first-order information. However, following the conventions of the optimization literature, we do not refer to cutting plane methods as first-order methods.} such as mirror descent (\Cref{MirrorDescent}) and (stochastic) follow-the-regularized-leader (\Cref{st_FTRL}) to minimize the \lovasz extension. 
These methods have performance depending more on problem geometry, e.g., the domains $B_\infty^V$ versus $S_k^V$, than cutting plane methods. 
Also, implementing them often does not require linear system solves and therefore they typically have much smaller iteration costs.

Unfortunately, these desirable features of first-order methods have a cost. 
In contrast to cutting plane methods, when applied to non-smooth convex objectives like the \lovasz extension, their convergence rate depends polynomially on the accuracy rather than polylogarithmically. 
Therefore, it is natural to use such methods for pseudo-polynomial SFM algorithms \cite{CLSW17,ALS20}, but less clear how to leverage them to obtain improved weakly- or strongly- polynomial SFM algorithms.

Fortunately, recent advances in weakly- and strongly-polynomial SFM algorithms provide hope for overcoming this limitation. Work of \cite{LSW15,DVZ21} provide different ways 
to incorporate learned \emph{precedence} 
constraints, i.e., 
if an element is in a minimizer then what other elements must also be in that minimizer, to reduce the scale of the problem. For example, \cite{DVZ21} showed that
it suffices to solve SFM approximately to a relative accuracy of $O(1/n^3)$, in a primal-dual sense, repeatedly to obtain a strongly-polynomial algorithm for SFM.

Despite the above hope for improving $\sparsity$-sparse SFM via first-order methods, there are a number of natural hurdles in the way. For example, the $O(1/n^3)$-error requirement in \cite{DVZ21} is prohibitively expensive for first-order methods to outperform cutting plane methods. 
Additionally, learning and updating precedence constraints need to be made sufficiently efficient. 

Nevertheless, we are able to follow this broad approach by introducing and leveraging a central concept of this paper we call {\em sparse dual certificates} (see \Cref{subsec:cert}). 
In particular, we demonstrate how to carefully apply first-order methods to $1/\poly(\sparsity)$-accuracy to compute sparse dual certificates and, building upon \cite{DVZ21}, how these certificates can be used to efficiently deduce precedence constraints. Our parallel and sequential algorithms differ in their specific implementations of these strategies (see Sections \ref{sec:parallel} and \ref{sec:sequential} respectively). We believe the notion of sparse dual certificates and our algorithmic techniques for computing and using them for $\sparsity$-sparse SFM might have broader applications to improving weakly- or strongly- polynomial time SFM algorithms.

\paragraph{Section organization.} To illustrate our approach and our key insights, we subdivide the remainder of this section. In \Cref{sec:approach:framework}, we provide the general framework we use to iteratively decrease the scale of the $\sparsity$-sparse SFM problem. \Cref{sec:approach:parallel} and \Cref{sec:approach:sequential}, we provide the key ideas in our parallel and sequential algorithms respectively.

\subsection{Framework}
\label{sec:approach:framework}

Building on a long line of work \cite{IFF01,IO09,LSW15} (and in particular, \cite{DVZ21}), our algorithms for minimizing a submodular function $f : 2^V \rightarrow \R$ works  by maintaining a set of precedence constraints indicating elements that must or must not be in any $\sparsity$-sparse minimizer, as well as for each $p \in V$ a set of elements $S_p$ that must be in any $\sparsity$-sparse minimizer $S^*$ containing $p$.
We call these precedence constraints \emph{arc constraints}\footnote{Our definition of arc constraints is only with respect to $\sparsity$-sparse minimizers and is therefore different from the standard one in the literature. See \Cref{sec:framework:extension} for more details.} and their collection a {\em ring family}.

Given these arc constraints, we consider an induced {\em submodular extension} $f^{\sharp \rfam}$ consistent with the ring family. 
$f^{\sharp \rfam}$ is essentially the complement of a submodular extension studied in \cite{DVZ21}; it is crucial that we work with $f^{\sharp \rfam}$ since sparsity is not preserved under complements. 
The extension $f^{\sharp \rfam}$ has many desirable properties. For example, minimizing $f^{\sharp \rfam}$ suffices for minimizing $f$ and any arc constraints learned for $f^{\sharp \rfam}$ apply to $f$.
Beyond consistency and submodularity, the key property we use about $f^{\sharp \rfam}$ is that the marginal vector\footnote{The formal notation we define and use for the marginal vector in \Cref{sec:prelim} and the rest of this paper is $u_{f^{\sharp \rfam}}$. Here we drop the subscript and use $u$ instead for simplicity.} $u \in \mathbb{R}^V$ defined as $u_p \defeq f^{\sharp \rfam}(\{p\}) -f^{\sharp \rfam}(\emptyset)$ for any coordinate $p \in V$ does not increase as we add arc constraints. (See \Cref{sec:framework:extension} for more details.)

By maintaining the ring family and the extension $f^{\sharp \rfam}$, and leveraging their properties, $\sparsity$-sparse SFM reduces to the problem of learning new arc constraints so that we can either 
\begin{enumerate}
\item decrease the scale of $\|u\|_\infty$ by more than a constant factor, or 
\item learn enough arc constraints so that the $\sparsity$-sparse minimizer is clear.
\end{enumerate}
In particular, if $\|u\|_\infty \leq \varepsilon / |V|$, then due to submodularity the largest set consistent with every arc constraint will be an $\varepsilon$-approximate minimizer for the original submodular function (see \Cref{claim:eps_approx_minimizer}).
Note how $\sparsity$-sparsity helps for our purposes: if the set of elements $S_p$ that must be in every $\sparsity$-sparse minimizer containing $p$ has more than $k$ elements, then $p$ cannot be in any $\sparsity$-sparse minimizer and can therefore be discarded.
This allows us to maintain at most $\sparsity$ arc constraints from any element $p \in V$, which significantly decreases the cost of manipulating the arc constraints and the submodular extension $f^{\sharp \rfam}$.

In both our parallel and sequential settings, we use $\norm{u}_\infty$ as a potential function and design efficient parallel and sequential subprocedures to find arc constraints 
to decrease $\norm{u}_\infty$ by a constant factor. Each setting has its distinct challenges and  our techniques differ correspondingly. 
However, there is one common technique underlying these two different implementations, based on the notion of a \emph{sparse dual certificate} (see \Cref{subsec:cert} for details).

\paragraph{Sparse dual certificates.} 
Sparse dual certificates are generalizations of standard dual solutions to SFM \cite{E70} that better capture the sparsity assumptions on the minimizers of the submodular function (see \Cref{def:certificate} for definition). 
In our framework, sparse dual certificates bridge the gap between the task of finding arc constraints and the pseudo-polynomial convergence rate of first-order methods. In particular, we show how to use sparse dual certificates to deduce arc constraints (see \Cref{subsec:cert}). We also develop various algorithmic techniques to compute these certificates by running first-order methods up to only $1/\poly(k)$ accuracy (see Sections \ref{subs_MD} and \ref{subs_stMD} respectively for our parallel and sequential algorithms for computing sparse dual certificates).

\subsection{Parallel Algorithm}
\label{sec:approach:parallel}

To motivate our parallel algorithm, consider minimizing the \lovasz extension $\hat{f}$ of the induced function $f^{\sharp \rfam} : 2^V \rightarrow \R$ 
with a $\sparsity$-sparse minimizer and let $f^* \defeq \min_{S \subseteq V} f^{\sharp \rfam}(S) = \min_{S \subseteq V} f(S)$. 
As discussed above, it suffices to learn arc constraints so that we can decrease $\norm{u}_\infty$ by a constant factor after updating the ring family. We may assume that $\min_{p \in V} u_p \geq 0$ as by submodularity adding any $p$ with $u_p < 0$ to any set decreases its value and therefore $p$ must be in every minimizer.

As a warm-up for this goal, perhaps the first natural question is: under these assumptions 
how efficiently can we compute a $\delta \norm{u}_\infty$-approximate minimizer for a given $\delta = 1/\poly(k)$? The question of deducing arc constraints seems harder than this problem 
since it involves 
proving something about all $\sparsity$-sparse minimizers at that accuracy threshold. 
For this warm-up question, let us even assume for now that $f^* \geq -\Omega(\norm{u}_\infty)$, as the problem is in some sense easier otherwise and will be addressed towards the end of this subsection. 

A natural approach to this warm-up problem, as alluded to earlier, is to apply standard first-order methods such as mirror descent to the \lovasz extension $\hat{f}$ of $f^{\sharp \rfam}$ over the domain $S_k^V$. By submodularity, $u$ entrywise upper bounds the subgradients of $\hat{f}$. If somehow the subgradients were also entrywise lower bounded by $-\norm{u}_\infty$, then 
standard analysis of mirror descent with an entropy regularizer (see Theorem 4.2 of \cite{B15}) applied to $\hat{f}$ over $S_k^V$ would compute an $\delta \norm{u}_\infty$-approximate minimizer in $\otilde(\delta^{-2})$ iterations. Furthermore, since each iteration of this method can be implemented in $O(1)$ depth, this would yield a $\otilde(\delta^{-2})$ depth algorithm as desired.

Unfortunately, 
it is not necessarily the case that every subgradient of $\hat{f}$ is entrywise lower bounded in magnitude by $-\norm{u}_\infty$. 
In fact, its most negative entry can be as negative as $f^* - (n-1) \norm{u}_\infty$, ruling out showing that mirror descent converges in $\otilde(\poly(k,\delta^{-1}))$ iterations. 

To overcome this issue, we show that the structure of $\sparsity$-sparse solutions allows us to \emph{truncate} subgradients. We prove that if we run mirror descent methods with every subgradient coordinate of value $\leq f^* - k \norm u_{\infty}$ set to $f^* - k \norm u_{\infty}$, then this still approximately minimizes the \lovasz extension $\hat{f}$ and computes sparse dual certificates (see \Cref{subs_MD}). Running mirror descent with these truncated subgradients yields a deterministic algorithm which
computes a $\delta \norm u_{\infty}$-approximate minimizer in 
$\otilde(\poly(k) / \delta^2) $ depth and 
$\otilde(n \cdot \poly(k) / \delta^2)$
 evaluation oracle queries. 

The solution to this warm-up problem is the key ingredient in our parallel algorithm. In particular, assuming $f^* \leq -\|u\|_\infty$, we show that the sparse dual certificate obtained by running the warm-up algorithm over $S_V^{k + 1}$ with accuracy $O(\|u\|_\infty/k)$
suffices to conclude an element that must be in every $\sparsity$-sparse minimizer, i.e., a dimensionality reduction.
As dimensionality reduction can occur at most $\sparsity$ times, 
this gives a $\otilde(\poly(k))$-depth $\otilde(n \cdot \poly(k))$-query algorithm.

On the other hand, when $f^* \geq - \norm u_{\infty}$, then we consider each of the induced submodular functions $f_p$ where an element $p$ is always included and run the same algorithm on each such function. Note that each $f_p$, once shifted to evaluate $0$ at the new emptyset (or the singleton $\{p\}$), has minimum value $-\Omega(u_p)$.
Consequently, when this is done for $p$ with $u_p$ near $\norm{u}_\infty$, the procedure finds an element which must be in any $\sparsity$-sparse minimizer containing $p$. 
Importantly, this can be done in parallel for each individual $p$! This conveys the main ideas behind the parallel algorithm. See \Cref{sec:parallel} for details.

\subsection{Sequential Algorithm}
\label{sec:approach:sequential}

In the previous section we outlined the main ideas of our parallel algorithm in \Cref{thm:main_parallel}.
Unfortunately, that algorithm has a rather high query complexity. In every round of decreasing $\norm{u}_\infty$, the algorithm might solve $n$ different induced SFM problems, corresponding to the inclusion of each element $p \in V$, causing the query complexity to scale quadratic in $n$ instead of linear.

In the literature on using ring families for weakly- and strongly-polynomial SFM, there is a standard technique for alleviating the need to apply the algorithm to $n$ different SFM problems to deduce arcs. In \cite{LSW15,DVZ21} and early SFM algorithms (see \cite{LSW15} for a discussion of the history), the algorithm 
obtains a suitable dual certificate for the original function. This dual certificate is then modified by moving individual elements to the start of each permutation, and it is argued that this modification can be used to deduce arcs. In other words, rather than running $n$ optimization methods to deduce $n$ dual certificates, these methods deduce one set of dual certificates and consider $n$ different modifications of it. 

In our sequential algorithm we follow a similar approach, but it brings about a variety of challenges, each of which requires algorithmic and analytic insights to overcome. The first challenge is that truncated subgradients, which are used in our parallel algorithm, do not seem amenable to this technique; it is unclear how to deduce arcs just by moving elements to the front after the truncation, 
which may lose critical information that makes this work. \cite{LSW15,DVZ21} consider the elements of the dual certificate that decrease significantly when moving a certain element to the start of each permutation. However, truncation does not seem to allow for a similar approach, as all components that are negative past a threshold are truncated to the same value.

To overcome this challenge, we provide a first-order method for computing $\epsilon$-approximate minimizers (and their associated sparse dual certificates) using true  (rather than truncated) subgradients. As discussed in \Cref{sec:approach:parallel}, this is difficult as the entries of the subgradient can vary by $\Omega(n \norm u_{\infty})$. Correspondingly, standard analysis of iterative first-order methods, e.g., mirror descent and FTRL (Follow-the-Regularized-Leader), require $\Omega(n^2)$  iterations, which would naively make a prohibitive $\Omega(n^3)$ queries! 
It is therefore imperative that we use a different technique (other than truncation as in the parallel setting) to either reduce the number of iterations or the cost per iteration; we do both. 
In particular, we use stochastic FTRL\footnote{We choose stochastic FTRL rather than stochastic mirror descent to facilitate the attainment of with high probability success guarantees.} (Follow-the-Regularized-Leader) where in each iteration we sample a random $1$-sparse unbiased estimator of the subgradient. We show how this can be implemented using $\otilde(1)$ evaluation queries per iteration and that the total number of iterations is suitably bounded.

Making the above approach work requires a number of insights. First, the number of iterations of stochastic FTRL is straightforwardly boundable in terms of the square of the $\ell_\infty$ norm of the stochastic estimates of the subgradient (analogous to as it was done for mirror descent in the parallel setting). However, unfortunately any sampling scheme in the worst case could have an $\ell_\infty$-norm of $\Omega(n \norm{u}_\infty)$, again leading to $\Omega(n^2)$ iterations.  
To get around this, we instead perform a more fine-grained analysis of the convergence of FTRL in terms of ``local norms'' (see \Cref{subs_stMD} for details). 
This is a known optimization method analysis technique and our analysis is inspired from and perhaps most closely resembles \cite{CJST19}; this technique was not used in  previous work on SFM that uses sampling \cite{CLSW17,HRRS19,ALS20}.

The next challenge is to actually implement sampling using $\otilde(1)$ queries per iteration so that the local norms of the samples are suitably small. Sampling $i \in V$ with probability proportional to $|(g_{x_t})_i|$ and then outputting $\mathsf{sign}(g_{x_t})_i \cdot \norm{g_{x_t}}_1$ would have the desired local norm bound. Additionally, sampling by this is essentially what is done in some of the sampling-based SFM methods \cite{CLSW17,HRRS19,ALS20} (albeit for a different norm analysis particularly relevant for pseudopolynomial SFM algorithms). However, these papers implement this sampling by somewhat complex dynamic data structure which could be challenging to analyze in our setting. Instead, we provide a simple straightforward sampling procedure which we call \vsampling. This sampling scheme picks $i \in V$ proportional to an upper bound $v_i$ for $|(g_{x_t})_i|$ so that $\sum_{i \in I} v_i$ for consecutive coordinates $I$ can be evaluated using  only $O(1)$ queries. This sampling can be implemented using $O(\log n)$ queries by a simple (static) binary tree data structure and we prove it has the desired expected local-norm bounds. 

Another challenge we face is that our stochastic FTRL analysis merely yields a subgradient $y$ that is a suitable dual certificate \textit{in expectation}, whereas we need the guarantee to hold with high probability in order to correctly deduce arc-constraints. To this end, we show that $\norm{y}_{\infty}$ is small with high probability\footnote{Throughout this paper, with high probability means that the probability is $1-n^{-C}$ for some constant $C > 0$.}, and apply the Azuma-Hoeffding concentration inequality for martingales to show that averaging over $\poly(k)$ such subgradients yields a suitable dual certificate with high probability. Showing that no entry of $y$ is too negative carries much of the difficulty in the analysis. For this step, we apply a novel analysis of our optimization method, which uses submodularity structure and couples the iterates of FTRL with iterates of an instantiation of the multiplicative weights algorithm. For details see \Cref{subs_stMD}.

The above method computes an implicit representation of $\otilde(n \cdot \poly(k))$ permutations such that the average of the subgradients they induce is a dual certificate of SFM from which either arcs or coordinates in the minimizer can be deduced. However, naively writing down the gradients that average out to the certificate would require $\Omega(n^2)$-queries. Furthermore, deducing an arc for a single coordinate, through the operation of moving a set of elements to the beginning of each permutation (which we often refer to as move-to-front for simplicity), would also naively require $\Omega(n^2)$-queries, which is prohibitively expensive. To overcome this limitation, we provide efficient methods to sample from these permutations and their associated subgradients. More specifically, we design a method which first draws $\otilde(n \cdot \poly(k))$ samples from the subgradients as a preprocessing step, and then uses the samples to deduce several arcs. The preprocessing step enables an efficient implementation of the move-to-front operations due through an importance sampling technique. Each arc deduced requires $\otilde(\poly(k))$ additional samples. For more details, see \Cref{subsec:arc_finding_sequential}.

This summarizes the main ingredients for obtaining our $\otilde(n \cdot \poly(k) \log(|f|/\epsilon))$-query result. Somewhat surprisingly, a more careful amortized cost analysis reveals that this algorithm is in fact a strongly-polynomial time algorithm that makes $\otilde(n \cdot \poly(k))$ queries. This stems, partially, from a more fine-grained analysis of the size of subgradients and how many arcs are deduced each time we compute an $\epsilon$-approximate minimizer (and its corresponding dual certificate). See \Cref{sec:sequential} for details.

This discussion omits a variety of details which are deferred to \Cref{sec:sequential}. The use of randomness and the loss of parallelism in this sequential algorithm is interesting and we leave it as an open problem to see to what degree a deterministic $\otilde(\poly(k) \log(1/\epsilon))$-depth and $\otilde(n \cdot \poly(k) \log(1/\epsilon))$-work weakly-polynomial time algorithm (and a strongly-polynomial time analog) can be achieved.

%% file: prelim.tex
\section{Preliminaries}
\label{sec:prelim}

\subsection{Notation}
\label{subs:notation}

We often use $V$ to denote a finite set of $n$ elements. 
For any real number $p \geq 1$, we use $\|\cdot\|_p$ to denote the $\ell_p$-norm in $\mathbb{R}^V$. 
We denote the unit cube by $B_\infty^V \defeq [0,1]^V$. For any integer $k > 0$, we denote the interior of the simplex scaled by $k$ as $\Delta_k^V \defeq\{x\in\R_{\geq 0}^V: \|x\|_1 \leq k\}$.
In particular, define the standard simplex $\Delta^V \defeq \Delta^V_1$. 
Further, we use $S_k^V \defeq B_\infty^V \cap \Delta_k^V$ to denote the \emph{truncated (interior of the) $k$-simplex}.
For every $S \subseteq V$, we use $\indicVec{S}$ to denote the indicator vector for the set $S$ (i.e. $(\indicVec{S})_i = 1, \forall i \in S$ and $(\indicVec{S})_i = 0, \forall i \in V \setminus S$). For simplicity, for every $i \in V$, we write $\indicVec{i}$ to denote the vector with a single $1$ in the $i$th coordinate and zero elsewhere. For any vector $v \in \mathbb{R}^V$, we use $|v|$ to denote the vector obtained by taking the coordinate-wise absolute value of $v$.

For any $x\in\R_{\geq 0}^V$ and $y\in\R^V$, we define $\norm y_{x}\defeq\sqrt{\sum_{i\in V}x_{i}y_{i}^{2}}$.
For any $y\in\R^V$, 
$\ell \in \mathbb{Z}_{>0}$, and $P \subseteq V$, we let $y(P) \defeq \sum_{p \in P} y_p$ be the sum of the coordinates of $y$ in $P$ and $y^{\ell}_{-}(P) \defeq \min_{w \in S_\ell^P} y^\top w$ be the sum of the most negative $\ell$ coordinates of $\min\{y, 0\}$ in $P$. We denote $y_-(V) \defeq y_-^n(V)$.

We denote the entropy regularizer (or entropy mirror map) as $r(x) \defeq \sum_{i\in V} x_{i}\log x_{i}$ for any $x\in\R_{\geq0}^{V}$
(where we define $0\log0 \defeq 0$).
We let $V_{x}(y)\defeq r(y)-\left(r(x)+\grad r(x)^{\top}(y-x)\right)$ the Bregman divergence of $r$.
Note that
\begin{align} \label{eq:breg_formula}
V_x(y) 
& = \sum_{i \in V} y_i \log y_i - \sum_{i \in V} x_i \log x_i - \sum_{i \in V} (1 + \log x_i) (y_i - x_i) \nonumber\\
& = \sum_{i \in V} y_i \log(y_i / x_i) + \sum_{i \in V} (x_i - y_i) = \langle y, \log(y/x) \rangle + \langle x - y, \vones \rangle . 
\end{align}

\subsection{Submodular Functions and \lovasz Extension} 
\label{subsec:prelim_submodular}

Let $f: 2^V \rightarrow \mathbb{R}$ be a set function defined on subsets of the $n$-element finite set $V$. 
We use $f^*$ to denote the minimum value of $f$ over $2^V$. 
A set function $f$ is {\em submodular} if it satisfies the following property of {\em diminishing marginal differences}: 

\begin{definition}[Submodularity] A function $f: 2^V \rightarrow \mathbb{R}$ is \emph{submodular} if $f(T \cup \{i\}) - f(T) \leq f(S \cup \{i\}) - f(S)$, for any $S \subseteq T \subseteq V$ and $i \in V \setminus T$. 
\end{definition}

In this section, the set function $f$ we work with is assumed to be submodular even without stated explicitly. 
We may assume without loss of generality that $f(\emptyset) = 0$ by replacing $f(S)$ by $f(S) - f(\emptyset)$. 
We make the following two assumptions about the submodular functions we work with throughout this paper even when it is not explicitly stated: (1) $f(\emptyset) = 0$, and (2) $f$ is accessed through an {\em evaluation oracle}, and use $\EO$ to denote the time to compute $f(S)$ for any $S \subseteq V$.

Throughout, we use $S_{\min}$ to denote the unique minimum minimizer\footnote{Such a minimal minimizer exists because for any two minimizers $S_1^*, S_2^*$ of $f$, their intersection $S_1^* \cap S_2^*$ and union $S_1^* \cup S_2^*$ must also be minimizers since by submodularity, 
$2 f^* = f(S_1^*) + f(S_2^*) \geq  f(S_1^* \cup S_2^*) + f(S_1^* \cap S_2^*)$.} of the given submodular function $f$. 
We also define the vector $u_f \in \mathbb{R}^V$ as $(u_f)_p \defeq f(\{p\})$ for all $p \in V$. 
In particular, by submodularity, $(u_f)_p$ is an upper bound on the marginal of $p$ for any $S \subseteq V \setminus \{p\}$, i.e., $f(S \cup \{p\}) - f(S) \leq (u_f)_p$.  
Given a submodular function $f$ and $P \subseteq V$, we define the contracted function $f_P: 2^{V \setminus P} \rightarrow \mathbb{R}$ as 
\[
f_P(S) \defeq f(S \cup P) - f(P) .
\]
Note that $f_P$ is also a submodular function with $f_P(\emptyset) = 0$.

\smallskip
\noindent \textbf{\lovasz Extension and Subgradients.} Our algorithm for SFM is based on a standard convex relaxation of a submodular function, known as the Lov\'asz extension \cite{L83}.  

\begin{definition}[Lov\'asz Extension]
The \lovasz extension, $\hat{f}: B_\infty^V \rightarrow \mathbb{R}$, of a submodular function $f$ is defined as 
$
\hat{f}(x) \defeq \E_{t \sim [0,1]} [f(\{i: x_i \geq t\})],
$ 
where $t \sim [0,1]$ is drawn uniformly at random. 
\end{definition}

We often overload notation and also use $f:B_{\infty}^V \rightarrow\R$ to denote the
\lovasz extension $\hat{f}$ of $f$.
The Lov\'asz extension $\hat{f}$ of a submodular function $f$ has many desirable properties. 
In particular, $\hat{f}$ is a convex relaxation of $f$ and it can be evaluated efficiently.

\begin{theorem}[Properties of Lov\'asz Extension, Theorem 6.3 in \cite{J22}]
\label{thm:lovasz_extension_properties}
Let $f: 2^{V} \rightarrow \mathbb{R}$ be a submodular function and $\hat{f}$ be its Lov\'asz extension. Then, 
\begin{itemize}
	\item [(a)] $\hat{f}$ is convex and $\min_{x \in B_\infty^V} \hat{f}(x) = \min_{S \subseteq V} f(S)$;
	\item [(b)] $f(S) = \hat{f}(I_S)$ for any $S \subseteq V$, where $I_S$ is the indicator vector for $S$; 
	\item [(c)] Suppose $x \in B_\infty^V$ satisfies $1 \ge x_{\pi(1)} \geq \ldots \geq x_{\pi(n)} \ge 0$ for a permutation $\pi: [n] \rightarrow V$, then $\hat{f}(x) = \sum_{i \in [n]} (f(\pi[i]) - f(\pi[i-1])) x_{\pi(i)}$ where $\pi[j] \defeq \{\pi(1), \cdots, \pi(j)\}$;
	\item [(d)] The set of minimizers of $\hat{f}$ is the convex hull of the set of minimizers of $f$. 
\end{itemize}
\end{theorem}

In particular, \Cref{thm:lovasz_extension_properties} (c) implies that a subgradient of $\hat{f}$ at $x \in \R^V$ is given by 
\[
(g_x)_{\pi_x(i)} \defeq f(\pi_x[i]) - f(\pi_x[i-1]) ,
\]
where $\pi_x: [n] \rightarrow V$ is the permutation corresponding to decreasing order of the coordinates of $x$ as in \Cref{thm:lovasz_extension_properties} (c). Moreover, given $x \in \R^V$, the subgradient $g_x$ can be computed in time $O(n \cdot \EO + n \log n)$ by sorting the coordinates of $x$ in decreasing order and applying the formula above. Note that $g_x$ only depends on the permutation $\pi_x$. Therefore, given  a permutation $\pi: [n] \rightarrow V$, we  also define $g_\pi$ as the subgradient induced, i.e., $(g_\pi)_{\pi(i)}=f(\pi[i])-f(\pi[i-1])$.

For $P \subseteq V$,
we define $\pi_{\leftarrow P}$ to be the permutation where the set $P$
is moved to the front of the permutation $\pi$ (with the relative order of elements in $P$ preserved). Formally, if $P = \{\pi(i_1), \cdots, \pi(i_\ell)\}$ for indices $i_1 < \cdots < i_\ell$ and the remaining indices are $j_1 < \cdots < j_{n-\ell}$, then 
\[
\pi_{\leftarrow P}(k) \defeq 
\begin{cases}
\pi(i_k) \qquad & \text{if } k \in [\ell], \\
\pi(j_{k - \ell}) & \text{if } k > \ell.
\end{cases}
\]
In particular, denote $\pi_{\leftarrow i} \defeq \pi_{\leftarrow \{i\}}$.
For a permutation $\pi: [n] \rightarrow V$ and $P \subseteq V$, we use $\Delta_{\pi,P} \in \mathbb{R}_{\geq 0}^V$ to denote the decrease of coordinates $V \setminus P$ in $g_{\pi}$ when we move $P$ to front, i.e., 
\[
(\Delta_{\pi,P})_q \defeq 
\begin{cases}
0 \qquad &\text{if $q \in P$}, \\
(g_{\pi})_q - (g_{\pi_{\leftarrow P}})_q \quad &\text{if $q \in V \setminus P$}.
\end{cases} 
\]
The \textit{base polytope} of a submodular function $f$ is defined as $B(f)\defeq \{y \in \mathbb{R}^V: y(S) \leq f(S), \forall S \subseteq V \text{ and } y(V) = f(V) \}$.  
Any vector $y \in B(f)$ can be represented as $y = \sum_{t \in [m]} \alpha_t g_{\pi_t}$, where $\pi_t$ are permutations and the coefficients $\alpha \in \Delta^{[m]}$. For any $P \subseteq V$, and $y \in B(f)$ represented as $y = \sum_{t \in [m]} \alpha_t g_{\pi_t}$, we define the vector $y_{\leftarrow P} \in \mathbb{R}^{V \setminus P}$ by restricting the vector $\sum_{t \in [m]} \alpha_t g_{(\pi_t)_{\leftarrow P}}$ to the coordinates in $V \setminus P$, i.e., $y_{\leftarrow P}$ is obtained from $y$ by moving $P$ to the front. In other words, $y_{\leftarrow P}$ is obtained by moving $P$ to front in every permutation $\pi_t$. Note that this operation depends on the permutations $\{\pi_t\}_{t \in [m]}$ in $y$'s representation. Whenever we write $y_{\leftarrow P}$, it will be clear from the context what representation we are referring to when performing the move-to-front operation.

%% file: framework.tex
\section{Our Framework}
\label{sec:framework}
In this section, we describe our general framework for $\sparsity$-sparse SFM. This framework essentially reduces the problem of $\sparsity$-sparse SFM to ``robustly'' minimizing the \lovasz extension to $\|u_f\|_\infty / \poly(k)$ accuracy. 
Our framework
bears resemblance to one in \cite{DVZ21}; we
build upon \cite{DVZ21} and introduce several new concepts and techniques in the following subsections. 

In \Cref{sec:framework:extension}, we discuss arc constraints and a natural extension of a submodular function (in accord with a given set of arc constraints), which is encapsulated through a general data structure we call the \emph{extension maintainer} (\Cref{alg:extension_maintainer});
then in \Cref{sec:framework:meta-algorithm}, we describe our meta algorithm behind our parallel and sequential algorithms (presented in \Cref{sec:parallel} and \ref{sec:sequential} respectively) for $\sparsity$-sparse SFM. 
Our meta algorithm contains two critical procedures, \textsf{Dimensionality-Reduction} and \textsf{Arc-Finding}, whose implementations differ for our parallel and sequential algorithms. Both implementions leverage a common concept we introduce, called a \emph{sparse dual certificate}, which might have broader SFM applications. 
We define sparse dual certificates and discuss its relationship to dimensionality reduction and arc finding in \Cref{subsec:cert}. Henceforth, we whenever we use the word ``algorithm'' in this section, we refer to our meta algorithm, unless specified otherwise.

\subsection{Arc Information and Extension Maintainer}
\label{sec:framework:extension}

\paragraph{Arc Information and Extension.} Our algorithm proceeds by finding elements $p,q \in V$ such that any  $\sparsity$-sparse minimizer of $f$ that contains $p$ must also contain $q$. We call such a constraint an \emph{arc}, denoted as $(p,q)$. 
Using arc information\footnote{\label{footnote:arc_difference} The arc definition we use here is slightly different from the standard definition in the literature \cite{IFF01,IO09,LSW15,DVZ21}. In the literature, an arc $(p,q)$ means that any minimizer of $f$ containing $p$ must contain $q$, while we only require that this holds for $\sparsity$-sparse minimizers. See \Cref{sec:ring_family_extensions} for more details.} for SFM was introduced by Iwata, Fleischer, and Fujishige \cite{IFF01} and is used in many SFM algorithms, e.g., \cite{IO09, LSW15,DVZ21}. 
For any element $p$, the set of all endpoints of arcs from $p$ is denoted as $p^\downarrow$
, where we adopt the convention that $p \in p^\downarrow$. 
Our algorithm maintains the set $E$ all arcs $(p, q)$ at a given state, as well as a set of contracted elements $W$
and a set of discarded elements $D$. $E, W, D$ naturally induce a \textit{ring family} (a set closed under complement and introduction), which we denote by $\rfam(E, W, D) \subseteq 2^{V}$. $\rfam(E, W, D)$\footnote{Our inclusion of $W$ and $D$ in the definition of $\rfam$ is perhaps non-standard but facilitates notation and explanations.} consists of all the sets $S$ so that $W \subseteq S \subseteq V \setminus D$, with the property that $S$ respects the arcs in $E$ (i.e. if $(p, q) \in E$, then $p \in S$ implies $q \in S$). Since our algorithms only maintain ring families through explicit $E, W, D$ we overload notation and sometimes use call $\rfam$ to refer to a ring family $\rfam(E,D,W)$, where the $(E, D, W)$ is dropped when clear from context.

We make the convention that are arcs in $E$ are only for elements in $V \setminus (W \cup D)$. 
We say a ring family $\rfam$ is $k$-\textit{consistent} with submodular $f$ if the arcs it captures are indeed valid (i.e., they preserve the structure of its $k$-sparse minimizers). The precise definition is as follows.

\begin{definition}[$k$-consistent ring family]\label{def:kconsistent}
    
    A \emph{ring family $\rfam(E, W, D)$} for finite set $V$ is determined by a set $W \subseteq V$ of contracted elements, a set $D \subseteq V$ of discarded elements (with $D \cap W = \emptyset$), and a set $E$ of arcs $(p, q)$ between elements of $V \setminus (W \cup D)$. 
    We say that \emph{$S \subseteq V$ is consistent with $\rfam$} if $W \subseteq S$, $D \cap S = \emptyset$, and if $p \in S$ and $(p,q)$ is a arc in $E$, then $q \in S$. Additionally, 
    we say that $\rfam$ is \emph{$k$-consistent (with $f$)} if every $k$-sparse minimizer of $f$ is consistent with $\rfam$. 
\end{definition}

Given a $k$-consistent ring family $\rfam$, to optimize $f$ it suffices to optimize the \lovasz extension over certain restrictions of the hypercube, see, e.g., \cite{LSW15}. However, to simplify algorithm development, following \cite{DVZ21}, we instead provide a new submodular function $f^{\sharp \rfam}: 2^V \rightarrow \mathbb{R}$ (see \Cref{defn:extension}) which encapsulates ring family constraints in a way that facilitates algorithm development. Our algorithms only require certain properties of $f^{\sharp \rfam}$, described in \Cref{lem:properties_extension} below (and proved in \Cref{sec:ring_family_extensions}), along with the extension maintainer, which is our data structure for accessing the extension.

\begin{restatable}[Properties of Extension $f^{\sharp \rfam}$]{lemma}{ExtensionProperty}
\label{lem:properties_extension}
Let $f: 2^{V} \rightarrow \R$ be a submodular function with a $\sparsity$-sparse minimizer and let $\rfam(E, W, D)$ a ring family $k$-consistent with $f$. Then, the following properties hold for the extension $f^{\sharp \rfam}:2^{V \setminus (W \cup D)} \to \R$: 
\begin{enumerate}
    \item \textbf{Submodularity}: $f^{\sharp \rfam}$ is a submodular function.
    \item \textbf{Extension}: $f^{\sharp \rfam}(S) \geq f(W \cup S^{\sharp \rfam})$ for any set $S \subseteq V \setminus W$, where $S^{\sharp \rfam} \subseteq S$ is defined so that $S^{\sharp \rfam} \cup W$ is the unique maximal subset of $S \cup W$ consistent with all the arcs in $E$; $f^{\sharp \rfam}(S) = f(W \cup S)$ for any set $S \subseteq V \setminus W$ with $S \cup W$ consistent with $\rfam$. 
    \item \textbf{$k$-Consistency}: For any $\sparsity$-sparse minimizer $S$ of $f$, $S \setminus W$ is also a $\sparsity - |W|$-sparse  minimizer of $f^{\sharp \rfam}$; for any minimizer $S^*$ of $f^{\sharp \rfam}$, $(S^*)^{\sharp \rfam} \cup W$ is a minimizer of $f$. 
    \item \textbf{Marginals}: 
    For any $p \in V$, $(u_{f^{\sharp \rfam}})_p = f(W \cup p^\downarrow) - f(W \cup p^\downarrow \setminus \{p\})$ if either $f(W \cup p^\downarrow) - f(W \cup p^\downarrow \setminus \{p\}) \ge 0$ or $p^\downarrow = \{p\}$, and $(u_{f^{\sharp \rfam}})_p = 0$ otherwise. 
    Consequently, $(u_{f^{\sharp \rfam}})_p < 0$ if and only if $p^\downarrow = \{p\}$ and $f(W \cup \{p\}) < f(W)$. Additionally, when new arcs are added or elements are added to $W$ or $D$, the value of $(u_{f^{\sharp \rfam}})_p$ does not increase for any $p \in V$. 
\end{enumerate}
\end{restatable}

Note that the second half of $k$-consistency in \Cref{lem:properties_extension} implies that in order to find a minimizer of the submodular function $f$ that is consistent with all the arcs, it suffices to minimize the submodular extension $f^{\sharp \rfam}$. 
In particular, any element that belongs to every minimizer of $f^{\sharp \rfam}$ must also belong to every minimizer of $f$. 
Moreover, by our definition of arcs, the first half of $k$-consistency in \Cref{lem:properties_extension} shows that any arc for the submodular extension $f^{\sharp \rfam}$ is also an arc for $f$. 
Therefore, to deduce either arcs or dimensionality reduction, it suffices for our algorithm to find new arcs or dimensionality reduction for the submodular extension $f^{\sharp \rfam}$.

\paragraph{Maintaining the Extension.} 
A core component of our framework, which is used in our main algorithm, as well as in our dimensionality reduction and arc-finding subroutines, is 
a data structure that maintains the parameters $E, W, D$ determining the ring family $\rfam(E, W, D)$, the marginals $(u_{f^{\sharp \rfam}})_p$, as well as access to the submodular extension $f^{\sharp \rfam}$ of $f$. We call this data structure the \emph{extension maintainer} and give the main theorem we use about it here.

Note that the maintainer itself may discover new elements of $W$ or new elements to discard (i.e. add to $D$). For instance, if $|p^\downarrow| > k$, then $p$ cannot belong to any $k$-sparse minimizer of $f$ and thus can be discarded. Additionally, when an element $p$ has an arc to any discarded element $q \in D$, then $p$ can also be discarded. 

The extension maintainer supports an operation \textsf{UpdateSpace} which allows for augmenting sets $W$ and $D$, and ensures that the ring family is $k$-consistent and the end of its call. 
Every call to the extension maintainer that can change quantities the parameters of the ring family ends with a call to \textsf{UpdateSpace}. 
Through a call to \textsf{UpdateSpace}, the maintainer first updates set $D$, then set $W$ to $W \cup W^{\mathrm{add}}$, computes all marginals $(u_{f^{\sharp \rfam}})_p$ for $p \in V$,
and then adds to $W$ every element $p$ with $(u_{f^{\sharp \rfam}})_p < 0$. Adding such an element $p$ to any set $S \in V \setminus (W \cup \{p\})$ decreases its value thanks to submodularity, so it must lie inside every minimizer. 
It then repeats the process of computing all marginals $(u_{f^{\sharp \rfam}})_p = f(W \cup \{p\}) - f(W)$ for $p \in V \setminus W$, and adding elements $p$ with $(u_{f^{\sharp \rfam}})_p < 0$, until no such element exists. 
Hence, the call eventually reaches a state where all marginals are non-negative (see property 1 of \Cref{thm:ext_maintainer}).

The following theorem summarizes the operations and costs for the extension maintainer data structure. We defer the details and proofs to \Cref{sec:ring_family_extensions}.

\begin{restatable}[Extension Maintainer]{theorem}{ExtMaintainer}
\label{thm:ext_maintainer}
Given a submodular function $f: 2^V \rightarrow \mathbb{R}$ with $n = |V|$ and minimal minimizer $S^*_{\min}$ with $|S^*_{\min}| \le k$, accessed through an evaluation oracle $\EO$, there is a data structure that maintains the set of elements that must be in every minimizer $W$, the set of discarded elements $D$ (with $D \cap W = \emptyset$), the set $E$ of all the arcs for elements in $V \setminus (D \cup W)$, the values $(u_{f^{\sharp \rfam}})_p$ for all $p \in V \setminus (D \cup W)$, and the corresponding submodular extension $f^{\sharp \rfam}: 2^{V \setminus (D \cup W)} \rightarrow \mathbb{R}$. Access to the extension is through the following operations. 
\begin{enumerate}
    \item \textsf{Init}$(V,k, f)$ initializes the data structure with $W = D = \emptyset$ and then 
   calls \textsf{UpdateSpace}$(\emptyset,\emptyset)$ defined below. The operation takes $O(n)$ time plus the time to call to \textsf{UpdateSpace}$(\emptyset,\emptyset)$. 

    \item \textsf{UpdateSpace}$(W^\mathrm{add}, D^\mathrm{add})$ updates the set of contracted elements $W$ to $W \cup W^\mathrm{add}$, as well as the set of discarded elements $D$ to $D \cup D^\mathrm{add}$
    for $D^\mathrm{add}$ and $W^\mathrm{add}$ such that the resulting ring family is $k$-consistent. The procedure then adds elements $p$ with $(u_{f^{\sharp \rfam}})_p < 0$ to $W$ until the ring-family is $k$-consistent and that each coordinate of $u_{f^{\sharp \rfam}}$ is non-negative.
    During a call of \textsf{UpdateSpace}, if the set $W$ changes its value from $W^{\mathsf{begin}}$ to $W^{\mathsf{end}}$, the depth and runtime of that call are $O(|W^{\mathsf{end}} \setminus W^{\mathsf{begin}}|)$ and $O(n |W^{\mathsf{end}} \setminus W^{\mathsf{begin}}| \cdot \EO + n |W^{\mathsf{end}} \setminus W^{\mathsf{begin}}|)$ respectively.
    
    \item \textsf{UpdateArcs}$(\{S_p\}_{p \in V})$ updates the data structure by adding, for each $p$, a set of new arcs $S_p$
    The procedure may then add additional arcs, augment $W$ or $D$. This takes 
    $O(k)$-depth, $O(m \cdot \EO + n k)$ time, where $m$ is the number of elements $p \in V \setminus (D \cup W)$ (after the update) that acquire new arcs, plus one call of \textsf{UpdateSpace}. 
    
    \item \textsf{Subgrad}$(\pi)$ outputs the subgradient of $f^{\sharp \rfam}$ restricted to $V \setminus (D \cup W)$ in $O(1)$-depth, $O(n \cdot \EO + n)$ time. 
    
    \item \textsf{Partial}$(i, \pi)$ outputs the $i$th coordinate of the subgradient of $f^{\sharp \rfam}$ in $O(1)$-depth, $O(\EO + n)$ time.

    \item $f^{\sharp \rfam}(S)$ outputs the value of $f^{\sharp \rfam}$ at set $S$ in constant depth and and $O(\EO + |S|)$ time. If $S = p^\downarrow$ for some $p \in V \setminus (W \cup D)$, the runtime is $O(\EO + 1)$. 

    \item Sets $W, D, \{p^\downarrow\}_{p \in V \setminus (W \cup D)}$, as well as vector $u_{f^{\sharp \rfam}}$, are explicitly stored and accessible. 
\end{enumerate}
    During each of call of \textsf{Init}, \textsf{UpdateArcs} and \textsf{UpdateSpace}, additional elements may be added to to $W$ and $D$. However, after the operations the data structure ensures that the ring-family is $k$-consistent and that each coordinate of $u_{f^{\sharp \rfam}}$ is non-negative.
\end{restatable}

\subsection{Meta Algorithm} 
\label{sec:framework:meta-algorithm}

In this section, we present our meta algorithm in \Cref{alg:main}, which is common to both our parallel and sequential algorithms for Theorems \ref{thm:main_parallel} -  \ref{thm:main_strongly_sequential}. As mentioned earlier, the difference between our parallel and sequential algorithms lie in their different implementations of the procedures $\textsf{Dimensionality-Reduction}$ and $\textsf{Arc-Finding}$ in \Cref{alg:main}. 
The guarantees of these two procedures are summarized as follows. 

\begin{itemize}
    \item \textsf{Dimensionality-Reduction}$(f,k)$: this procedure takes as input an integer $k>0$ and a submodular function $f$ with a $\sparsity$-sparse minimizer. If $f^* \leq -\|u_f\|_\infty / 6k$ (i.e., $f^*$ is sufficiently negative), then \textsf{Dimensionality-Reduction}$(f,k)$ returns a set $T \neq \emptyset$ that belongs to every minimizer of $f$ with high probability, i.e., what we call a \emph{dimensionality reduction}. 

    \item \textsf{Arc-Finding}$(f,k,\mathsf{Scale})$: this procedure takes as input an integer $k > 0$, a submodular function $f$ with a $\sparsity$-sparse minimizer, and a parameter $\mathsf{Scale} > 0$. If $f^* \geq -\|u_f\|_\infty / 6k$ (i.e., $f^*$ is not too negative) and $\mathsf{Scale} \geq \|u_f\|_\infty$, then, with high probability, \textsf{Arc-Finding}$(f,k,\mathsf{Scale})$ returns a non-empty set $S_p$ of endpoints of arcs from $p$ for every element $p \in V$ such that $(u_f)_p \geq \mathsf{Scale}/2$ and $p$ belongs to a $k$-sparse minimizer. (Consequently, if $(u_f)_p \geq \mathsf{Scale}/2$ and $S_p = \emptyset$, then $p$ is not in any $k$-sparse minimizer.)
\end{itemize}

The parallel and sequential implementations of $\textsf{Dimensionality-Reduction}$ are given in \Cref{subsec:dimensionality_reduction_parallel} and \Cref{subsec:dimensionality_reduction_sequential} respectively.
The parallel and sequential implementations of $\textsf{Arc-Finding}$ are given in \Cref{subsec:arc_finding_parallel} and \Cref{subsec:arc_finding_sequential} respectively.

\paragraph{The Meta Algorithm.} We now describe our meta \Cref{alg:main}. Like the extension maintainer (\Cref{alg:extension_maintainer}), the algorithm maintains its own copy of sets $W$ and $D$, the set $E$ of all the arcs for elements in $V \setminus (W \cup D)$, and the marginals $(u_{f^{\sharp \rfam}})_p$. It works with extension function $f^{\sharp \rfam}$ which encodes the arc constraints, to which it has query, subgradient, and subgradient coordinate access via the extension maintainer (\Cref{thm:ext_maintainer}). 
Every time there is a change in any of $E, W, D$, or the $u_{f^{\sharp \rfam}}$ vector, due to a function call to the extension maintainer, the meta algorithm's copies of these quantities are also dynamically updated to match these values. 

We now describe the steps of our meta-algorithm. First, through the call $\mathsf{Init}(V,k,f)$, the extension maintainer (\Cref{alg:extension_maintainer}) first initializes $E, W, D$ and the $u_{f^{\sharp \rfam}}$ vector so that the ring-family is $k$-consistent and that each coordinate of $u_{f^{\sharp \rfam}}$ is non-negative (meaning $(u_{f^{\sharp \rfam}})_q \ge 0, \forall q \in V \setminus (W \cup D)$).

In each iteration of the while loop \Cref{line:arc_finding_main}, the algorithm attempts to either find a dimensionality reduction (i.e., an element that is in every minimizer), or find arcs, as discussed in \Cref{sec:approach}. 
Since every marginal is positive after every function call to \Cref{alg:extension_maintainer}, in each iteration of the while loop, 
our procedures \textsf{Dimensionality-Reduction} and \textsf{Arc-Finding} always work with $u_{f^{\sharp \rfam}}$ so that $(u_{f^{\sharp \rfam}})_q \geq 0, \forall q \in V \setminus (W \cup D)$. Inside the while loop in \Cref{line:arc_finding_main}, 
we check whether\footnote{For our parallel algorithm in \Cref{sec:parallel} we actually use the threshold $\frac{\|u_{f^{\sharp \rfam}}\|_{\infty}}{4}$ to improve the oracle complexity by a $\poly(k)$ factor. This larger threshold doesn't work for our sequential algorithm in \Cref{sec:sequential}.} $f^* > -\|f^{\sharp \rfam}\|_\infty / 6k$ or not using the procedure $\textsf{Dimensionality-Reduction}$. 
In particular, if the minimum value $f^*$ is very negative in the sense that $f^* \leq -\|f^{\sharp \rfam}\|_\infty / 6k$, then $\textsf{Dimensionality-Reduction}$ will find a set of elements $T \subseteq V \setminus W$ with $T \neq \emptyset$ that belong to every minimizer. This set of elements $T$ can then be contracted. 
Note that whenever $|W| \geq k$, or every element in $V$ is either contracted or discarded (i.e., $W \cup D = V$), then $W$ must be a $\sparsity$-sparse minimizer so \Cref{alg:main} will return $W$. 

On the other hand, if $f^*$ is close to $0$ in the sense that $f^* > -\|f^{\sharp \rfam}\|_\infty / 6k$, then \Cref{alg:main} calls the procedure \textsf{Arc-Finding} to find a set $S_p \neq \emptyset$ of endpoints of arcs from any element $p$ with large marginal $(u_{f^{\sharp \rfam}})_p > \|u_{f^{\sharp \rfam}}\|_\infty / 2$. Since the value of $\|u_{f^{\sharp \rfam}}\|_\infty$ does not increase after adding arcs due to \Cref{lem:properties_extension}, we can continue finding arcs until $\|u_{f^{\sharp \rfam}}\|_\infty$ drops by more than a constant factor. 
Whenever $\|u_{f^{\sharp \rfam}}\|_\infty$ becomes smaller than $\epsilon/n$, then the set $V \setminus D$ must be an $\epsilon$-approximate minimizer for $f^{\sharp \rfam}$ by applying the following claim to the submodular function $f^{\sharp \rfam}: 2^{V \setminus D} \rightarrow \mathbb{R}$.

\begin{claim} \label{claim:eps_approx_minimizer}
    Let $f: 2^V \rightarrow \mathbb{R}$ be a submodular function such that $\|u_{f}\|_\infty \leq \epsilon / n$, then the set $V$ is an $\epsilon$-approximate minimizer of $f$.
\end{claim}
\begin{proof}
For any $S^* \subseteq V$ that minimizes $f$ we have that $f(V) \leq f(S^*) + \sum_{p \in V \setminus S^*} (u_f)_p \leq f^* + \epsilon$.
\end{proof}

The correctness of the meta algorithm follows from \Cref{claim:eps_approx_minimizer} and its description above.

\begin{corollary}[Correctness of Meta Algorithm] \label{cor:correctness_meta_algorithm}
Assuming the correctness of the instantiations of $\textsf{Dimensionality-Reduction}$ and $\textsf{Arc-Finding}$, 
\Cref{alg:main} always outputs an $\epsilon$-approximate minimizer of $f$. 
Additionally, if \Cref{alg:main} outputs the set $W$, then it is an exact minimizer of $f$. 
\end{corollary}
\begin{proof}
First, note that after each function call of data structure \Cref{alg:extension_maintainer}, all marginals $(u_{f^{\sharp \rfam}})_p$ are non-negative. This holds since each call that could change $E, W, D$ or $(u_{f^{\sharp \rfam}})_p$ values ends with a call to \textsf{UpdateSpace}, which by \Cref{thm:ext_maintainer}, ensures all marginals $(u_{f^{\sharp \rfam}})_p$ are non-negative at the end of it. Thus, whenever we call our $\textsf{Dimensionality-Reduction}$ and $\textsf{Arc-Finding}$ methods, the corresponding vector $u_{f^{\sharp \rfam}}$ is non-negative.

Next, note that assuming the correctness of $\textsf{Dimensionality-Reduction}$ and $\textsf{Arc-Finding}$, the elements added to $W$ (i.e., \Cref{line:add_to_W}) are always in every $k$-sparse minimizer, and the elements added to $D$ are not in every $\sparsity$-sparse minimizers. This holds because our $\textsf{Arc-Finding}$ method is guaranteed to find at least one arc from every $p$ that belongs to a $k$-sparse minimizer and has $(u_f)_p \geq \mathsf{Scale}/2$. Consequently, each call to \textsf{UpdateSpace} preserves the property that at the end of it, $W \subseteq S_{\min}$ (by \Cref{thm:ext_maintainer}). 

Finally, having proved that $E, W, D$ are updated correctly, note that
when the stopping condition of the while loop in \Cref{line:arc_finding_main} is $|W| \geq k$ or $D \cup W = V$, then the set $W$ is a $\sparsity$-sparse minimizer. If the while loop is ended because $\|u_{f^{\sharp \rfam}}\|_\infty \leq \epsilon / n$, then by \Cref{claim:eps_approx_minimizer}, the set $V \setminus D$ is an $\epsilon$-approximate minimizer of $f^{\sharp \rfam}$, which is also an $\epsilon$-approximate minimizer of $f$ by \Cref{lem:properties_extension}. This follows because $V \setminus D$ is consistent with the ring family, which implies that $f^{\sharp \rfam}(V \setminus D) \ge f(V \setminus D)$ by the extension property of \Cref{lem:properties_extension}. 
\end{proof}

Note how sparsity helps us in our meta algorithm: for \textsf{Dimensionality-Reduction}, the algorithm can find at most $k$ contracted elements $W$ since $f$ is guaranteed to have a $k$-sparse minimizer; for \textsf{Arc-Finding}, we can find at most $k$ arcs from each element $p$ with $(u_{f^{\sharp \rfam}})_p \ge \mathsf{Scale} / 2$ before concluding that $p$ does not lie in any $k$-sparse minimizer and can be safely discarded. This guarantees that $\|u_{f^{\sharp \rfam}}\|_\infty$ has to go down by a factor of $2$ after at most $k$ iterations of the inner while loop in \Cref{line:arc_finding_main}.

\begin{algorithm}[htp!]
\caption{Meta Algorithm}\label{alg:main}
\SetKwFunction{stdimred}{StochasticFindContraction}
\SetKwFunction{stdualcert}{StochasticDualCertificate}
\SetKwFunction{updatestruct}{UpdateStructure}
\SetKwFunction{solvescale}{SolveScale}
\SetKwFunction{findarcs}{FindArcs}
\KwData{Integer $k>0$, submodular function $f : 2^V \rightarrow \R$ with a $k$-sparse minimizer, and accuracy $\epsilon > 0$
}
\KwResult{An $\epsilon$-approximate minimizer of $f$}
 \medskip
Extension maintainer $\textsf{Ext} \leftarrow \mathsf{Init}(V,k,f)$\;
$W \gets \textsf{Ext}.W, D \gets \textsf{Ext}.D, u_{f^{\sharp \rfam}} \gets \textsf{Ext}.u_{f^{\sharp \rfam}}$\label{line:update_infos1}
 \tcp*{$f^{\sharp \rfam}: 2^{V \setminus (D \cup W)} \rightarrow \mathbb{R}$}
\While{$|W| < k, D \cup W \ne V, \|u_{f^{\sharp \rfam}}\|_{\infty} > \frac{\epsilon}{n}$ \label{line:outer_while_main}}{
    $T \gets \textsf{Dimensionality-Reduction}(f^{{\sharp \rfam}}, k)$ \label{line:solve_scale_sub-procedure} \tcp*{Find dim reduction $T \neq \emptyset$ if $f^* \leq -\|u_{f^{\sharp \rfam}}\|_\infty / 6k$}
    \If(\tcp*[f]{Add $T$ to contracted elements $W$ \label{line:add_to_W}}){$T \ne \emptyset$}{$\mathsf{UpdateSpace}(T, \emptyset)$\;
    $W \gets \textsf{Ext}.W, D \gets \textsf{Ext}.D, u_{f^{\sharp \rfam}} \gets \textsf{Ext}.u_{f^{\sharp \rfam}}$\label{line:update_infos2}
    }
    \Else{
        $\mathsf{Scale} = \|u_{f^{\sharp \rfam}}\|_{\infty}$ \tcp*{Find arcs if $f^* > -\|u_{f^{\sharp \rfam}}\|_\infty / 6k$}
        \While{$\|u_{f^{\sharp \rfam}}\|_{\infty} > \frac{\mathsf{Scale}}{2}$ \label{line:arc_finding_main}}{
             $\{S_p\}_{p \in V \setminus (D \cup W)} \leftarrow \textsf{Arc-Finding}(f^{\sharp \rfam}, k, \mathsf{Scale})$ \tcp*{Find arcs if $(u_{f^{\sharp \rfam}})_p \geq \mathsf{Scale}/2$}
             \For{$p \in V \setminus (D \cup W)$, $(u_{f^{\sharp \rfam}})_p \geq \mathsf{Scale}/2$ and $S_p = \emptyset$}{
             $D \gets D \cup \{p\}$ \tcp*{Discard $p$ if $(u_{f^{\sharp \rfam}})_p \geq \mathsf{Scale}/2$ and no arcs found}
             }
             $\textsf{Ext}.\mathsf{UpdateSpace}(\emptyset, D \setminus \textsf{Ext}.D)$\;
             $\textsf{Ext}.\mathsf{Update}(\{S_p\}_{p \in V \setminus (D \cup W)})$ \tcp*{$\| \textbf{}u_{f^{\sharp \rfam}}\|_{\infty}$ decreases after adding arcs}
             $W \gets \textsf{Ext}.W, D \gets \textsf{Ext}.D, u_{f^{\sharp \rfam}} \gets \textsf{Ext}.u_{f^{\sharp \rfam}}$\label{line:update_infos3}
        }
    } 
}
\Return $V \setminus D$ \label{line:return}
\end{algorithm}

\subsection{Sparse Dual Certificate}
\label{subsec:cert}

Both of our \textsf{Dimensionality-Reduction} and \textsf{Arc-Finding} procedures crucially rely on a core notion we call $(\delta, k)$ dual certificate. 
To motivate this concept, we recall that Edmonds' minimax theorem \cite{E70} states that
\begin{align} \label{eq:min-max-SFM}
\max_{y \in B(f)} y_-(V) =  \min_{S \subseteq V} f(S) ,
\end{align}
where $B(f)$ is the base polytope, as defined in \Cref{subsec:prelim_submodular}. 
In the beautiful framework established by \cite{DVZ21}, dimensionality reduction and arc information are deduced from an approximately optimal dual solution $y \in B(f)$ that satisfies
\[
y_-(V) \leq f^* \leq y_-(V) + \delta .
\]
However, in order to find a dimensionality reduction or arc information, the approximation quality of the dual solution $y$ is required to be $\delta = \|u_f\|_\infty/\poly(n)$ in \cite{DVZ21}. To the best of the authors knowledge, to achieve such a small accuracy, all known gradient descent methods would take at least a comparable $\poly(n)$ number of iterations, which would be prohibitively expensive.

In our framework, we relax the quality of the approximate dual solution, and, by exploiting the sparsity structure of our submodular function, show that for our dimensionality reduction and arc finding procedures, it is sufficient. 
More specifically, our procedures only need the approximation error to be $\delta = \frac{\|u_f\|_\infty}{\poly(k)}$ instead of $\frac{\|u_f\|_\infty}{\poly(n)}$, which allows us to only run $\poly(k)$ iterations of gradient descent methods.

\begin{definition}[$(\delta, k)$ Dual Certificate]
\label{def:certificate}
    $y \in \R^V$ is a \emph{$(\delta, k)$ dual certificate} for submodular function $f : 2^V \rightarrow \R$ and $\delta > 0$ if
    \begin{enumerate}
        \item $f^* \le y^{k+1}_{-}(V) + \delta$ and
        \item $f(S) \ge y(S)$ for every $k$-sparse $S \subseteq V$.
    \end{enumerate} 
\end{definition}

Note that the approximation to $f^*$ we use in the above definition is $y^{k+1}_{-}(V)$ instead of $y_{-}(V)$ in \eqref{eq:min-max-SFM}. 
This is because if $f$ has a $\sparsity$-sparse minimizer $S^*$, then any dual solution $y \in B(f)$ satisfies $y_-^k(V) \leq f^*$. To see this, let $g_\pi \in B(f)$ be the BFS corresponding to permutation $\pi$. Then, 
\begin{align*}
(g_\pi)^k_-(V) = \min_{S \subseteq V, |S|\leq k} g_\pi^\top \indicVec{S} \leq g_\pi^\top \indicVec{S^*} \leq f^* ,
\end{align*}
where the last inequality follows from Lemma 63 in \cite{LSW15}. Since any $y \in B(f)$ is a convex combination of BFSs, we have $y_-^k(V) \leq f^*$.
This suggests that the quantity $y_-^k(V)$ is intuitively a natural dual characterization of $f^*$ in the sparse setting. 

We will present efficient parallel and sequential algorithms for computing $(\delta,k)$ dual certificates in \Cref{subs_MD} and \Cref{subs_stMD} respectively.
For the remainder of this subsection, we discuss how the notion of $(\delta,k)$ dual certificate is useful for dimensionality reduction and arc finding.

\paragraph{Dimensionality Reduction.} Here, we show that given a $(\delta, k)$ dual certificate for a suitably chosen $\delta$, we can find an element $y \in V$ that is contained in every minimizer of $f$, i.e., a dimensionality reduction. 
Before doing so, we first claim that any element that lies in every $k$-sparse minimizer of $f$ must also lie in every minimizer of $f$. 

\begin{claim}
\label{claim:every_minimizer}
Let $f: 2^V \rightarrow \mathbb{R}$ be a submodular function with a $\sparsity$-sparse minimizer. If element $p \in V$ lies in every $\sparsity$-sparse minimizer of $f$, then it must lie in every minimizer of $f$. 
\end{claim}

\begin{proof}
Let $S_{\min}$ be the unique minimal minimizer of $f$, which satisfies $|S_{\min}| \leq k$ by assumption. It follows that $p \in S_{\min} \subseteq S^*$ for any other minimizer $S^*$ of $f$. 
\end{proof}

The next lemma is key to the development of our implementations of the \textsf{Dimensionality-Reduction} procedure. 

\begin{lemma}
\label{claim:certif_dim_red}
Let $f : 2^V \rightarrow \R$ be a submodular function with a $k$-sparse minimizer such that $f(\emptyset) = 0$, and $\delta > 0$. Let $y \in \mathbb{R}^V$ be a $(\delta, k)$ dual certificate for $f$. Then every $p \in V$ of $y$ with $y_p < - \delta$ must be in every minimizer of $f$. In particular, if $\delta < \frac{|f^*|}{k}$, then the index of the most negative component of $y$ must be in every minimizer of $f$. 
\end{lemma}

\begin{proof}
We start with the first statement. Let $p \in V$ be a coordinate with $y_p < - \delta$ and assume for the purpose of contradiction that $p$ does not lie in a $\sparsity$-sparse minimizer $S^*$ of $f$. 
Then, 
\[f(S^*) \ge y(S^*) \ge y_-^{k+1}(V) - y_p \geq f^* - \delta - y_p > f^* ,\]
contradicting to the assumption that $S^*$ is a $\sparsity$-sparse minimizer of $f$. This implies that $p$ belongs to every $\sparsity$-minimizer of $f$. Then by \Cref{claim:every_minimizer}, $p$ must be in every minimizer of $f$. 

For the second statement, let $S^*$ be any $k$-sparse minimizer of $f$. Then $y(S^*) \leq f^*$ so there most negative coordinate $p \in S^*$ of $y$ must satisfies $y_p \leq f^* / k < -\delta$. The second statement then follows from the first statement.
\end{proof}

Thus far, our discussion applied to both the parallel and sequential algorithms. From this point onward, the arc-finding technique described in the paragraphs that follow, as well as the corresponding lemma (\Cref{lem:move_to_front_dual_cert}) applies exclusively to the sequential algorithm.

\paragraph{Arc Finding.} The next lemma of this subsection is crucial to our \textsf{Arc-Finding} procedure. It states that given a $(\delta,k)$ dual certificate $y$ of $f$, moving an element $p$ to the front of the permutations corresponding to $y$ produces a dual certificate for the contracted function $f_p(\cdot): 2^{V \setminus \{p\}} \rightarrow \mathbb{R}$. Then any dimensionality reduction $S_p \subseteq V \setminus \{p\}$ for the function $f_p$ is a set of arcs from $p$. 

\begin{lemma}\label{lem:move_to_front_dual_cert}
Let $f$ be a submodular function, $y \in B(f)$ be a $(\delta, k)$ dual certificate, and $P \subseteq V$ be a subset, 
then $y_{\leftarrow P} \in B(f_P)$. 
Moreover, if $P \subseteq S^*$, where $S^*$ is a $\sparsity$-sparse minimizer of $f$, then $y(P) \leq \delta$ and 
$y_{\leftarrow P}$ is a $(\delta - y(P), k)$ dual certificate for $f_P$. 
\end{lemma}

\begin{proof}
Let $y = \sum_{t \in [m]} \alpha_t g_{\pi_t}$ be a representation of $y$, for coefficients $\alpha \in \Delta^{[m]}$. Let $y' \defeq \sum_{t \in [m]} \alpha_t g_{(\pi_t)_{\leftarrow P}}$. 
To prove $y_{\leftarrow P} \in B(f_P)$, we first note that 
\[
y_{\leftarrow P}(V \setminus P) = y'(V) - y'(P) = f(V) - f(P) = f_P(V \setminus P).
\]
For every $S \subseteq V \setminus P$, we have 
\[
y_{\leftarrow P}(S) =  y'(S \cup P) - y'(P) \le f(S \cup P) - f(P) = f_P(S) .
\]
The above two observations imply that $y_{\leftarrow P} \in B(f_P)$. 

Next, we show that  $y(P) \leq  \delta$ for any $P \subseteq S^*$ for $\sparsity$-sparse minimizer $S^*$ of $f$.
This follows from
\[
f^* \geq y(S^*) = y(P) + y(S^* \setminus P) \geq y(P) + y^{k+1}_-(V) \geq y(P) + f^* - \delta , 
\]
where the second inequality is because $|S^* \setminus P| \leq |S^*| \leq k$ and the last inequality uses the assumption that $y$ is a $(\delta, k)$ dual certificate.  
The above implies that $y(P) \leq \delta$.

Finally, we further show that $y_{\leftarrow P}$ is a $(\delta - y(P), k)$ dual certificate for $f_P$ when $P \subseteq S^*$. Since $y_{\leftarrow P} \in B(f_P)$, we have $y_{\leftarrow P}(S) \le f_P(S)$ for any $S \subseteq V \setminus P$. Hence, we are left to show 
\begin{align} \label{eq:dual_move_1}
(f_P)^* \le (y_{\leftarrow P})^{k+1}_{-}(V \setminus P) + \delta - y(P). 
\end{align}
Define $T^* \defeq S^* \setminus P$. 
Note that 
\[
f_P(T^*) = f(S^*) - f(P) = f^* - f(P) . 
\]
On the other hand, any set $T' \subseteq V \setminus P$ satisfies 
\[
f_P(T') = f(T' \cup P) - f(P) \geq f^* - f(P) ,
\]
which implies that $(f_P)^* = f^* - f(P)$. 
Therefore, to prove \eqref{eq:dual_move_1}, it suffices to prove that
\begin{align} \label{eq:dual_move_2}
f^* - f(P) \leq (y_{\leftarrow P})^{k+1}_{-}(V \setminus P) + \delta - y(P) = (y')^{k+1}_{-}(V \setminus P) + \delta - y(P). 
\end{align}
To prove \eqref{eq:dual_move_2}, we seek to compare $(y')_-^{k+1}(V \setminus P)$ to $y_-^{k+1}(V \setminus P)$. As we move $P$ to the front for each permutation in the subgradient $g_{\pi_t}$ to obtain $y'$ from $y$, we have $y'_q \leq y_q$ for each coordinate $q \in V \setminus P$ due to submodularity. 
Also note that $y'(P) - y(P) = f(P) - y(P)$ and $y'(V) = y(V) = f(V)$.  
These imply that 
\[
\sum_{q \in V \setminus P} (y_q - y'_q)  =  (y(V) - y(P)) - (y'(V) - y'(P)) = f(P) - y(P) ,
\]
i.e., the total amount of decrease for coordinates in $V \setminus P$ when changing from $y$ to $y'$ is $f(P) - y(P)$. In particular, the most negative $k+1$ coordinates of $y$ in $V \setminus P$ can decrease by at most $f(P) - y(P)$. Therefore, we have
\[
(y')^{k+1}_{-}(V \setminus P) \ge y^{k+1}_{-}(V \setminus P) - (f(P) - y(P)) \ge f^* - \delta - (f(P) - y(P)),
\]
where the last inequality uses the assumption that $y$ is a $(\delta, k)$ dual certificate for $f$. 
This proves \eqref{eq:dual_move_2} and completes the proof of the lemma.
\end{proof}

%% file: parallel_algo.tex
\section{$\mathsf{Poly}(k)$-Depth Parallel Algorithm for $k$-sparse SFM}
\label{sec:parallel}

In this section, we present our parallel algorithm for $k$-sparse SFM and prove \Cref{thm:main_parallel}. 
We apply mirror descent with an entropic regularizer to the \lovasz extension of $f$ restricted to $S_k^V$, the subset of vectors in $[0,1]^V$ with $\ell_1$ norm at most $k$, up to an accuracy $\|u_f\|_\infty / \poly(k)$ to obtain a $(\delta,k)$ dual certificate for  dimensionality reduction or arc finding.

However, naively using mirror descent yields an algorithm whose number of iterations depends polynomially on the ratio between the $\ell_\infty$-norm of the subgradient and the accuracy we need, which may be up to $n$ if the $\ell_\infty$-norm of the gradient is large.
To get around this large dependence on $n$, we employ a novel technique of truncating the subgradient values. We show that we can cap the negative values at some threshold depending on $f^*$ and $\norm u_{\infty}$ so that running mirror descent up to the same desired accuracy still allows us to find a contraction or an arc.

The section is organized as follows. 
We start by presenting the classic mirror descent algorithm and its convergence guarantee in \Cref{subsec:mirror_descent}. 
In \Cref{subs_MD}, we present our mirror descent with truncation method for obtaining $(\delta, k)$ dual certificates.
Then in \Cref{subsec:dimensionality_reduction_parallel}, we present our parallel implementation of \textsf{Dimensionality-Reduction}. Our parallel implementation of \textsf{Arc-Finding} is contained in  \Cref{subsec:arc_finding_parallel}.
Finally in \Cref{subs_correctness}, we provide a proof of \Cref{thm:main_parallel}.

\subsection{Mirror Descent}
\label{subsec:mirror_descent}

Before detailing our algorithm, we first present the classic mirror descent algorithm (\Cref{MirrorDescent}) of Nemirovski and Yudin \cite{NY83}. 
We refer to the excellent monograph \cite{B15} for more background and details on mirror descent. 
The following lemma is the standard performance guarantee of mirror descent, which is a slight adaptation\footnote{The statement of \Cref{lem:MD_guarantee} comes from the RHS of the last inequality in the proof of Theorem 4.2 in \cite{B15}. Note that the bound on regret $\sum_{t=0}^{m-1} \langle h_t, x_t - w \rangle$ does not require $h_t$ being the subgradient of $f$ at $x_t$.} of Theorem 4.2 in \cite{B15}.

\begin{algorithm}[htp!]
\caption{Mirror Descent}\label{MirrorDescent}
    \KwData{A convex function $f$ on $\mathbb{R}^n$, a convex domain $D \subseteq \mathbb{R}_{>0}^n$, an initial point $x_0 \in D$, a step size $\eta > 0$, and number of iterations $m$}
    \KwResult{Sequence of iterates $\{x_0, x_1, \ldots, x_m\} \subseteq D$ which satisfies \Cref{lem:MD_guarantee}}
    \SetKwFunction{MD}{MirrorDescent}
    \SetKwProg{Fn}{Function}{:}{}
    \Fn{\MD{$f, D, x_0, \eta, m$}}{
    \For{$t = 0, 1, \ldots, m$}{
        Compute a vector $h_t \in \mathbb{R}^n$ \label{line:compute_ht} \tcp*{Approximate subgradient $h_t$ can depend on $\{x_i\}_{i=0}^t$}
        $x_{t+1} = \argmin_{x \in D} \eta h_t^\top x + V_{x_t}(x)$  \tcp*{Recall $V_{x_t}(x)\defeq r(x)-\left(r(x_t)+\grad r(x_t)^{\top}(x-x_t)\right)$}
    }
    \Return $\{x_0, x_1, \ldots, x_m\}$\;
    }
\end{algorithm}

\begin{lemma}[Mirror Descent, Theorem 4.2 in \cite{B15}]
\label{lem:MD_guarantee}
Let function $f$ be convex on $\mathbb{R}^n$ and $r: \R^n \rightarrow \R$ be $\rho$-strongly convex on $D \subseteq \mathbb{R}^n$ with respect to norm $\|\cdot\|$. If the vectors $h_0, \ldots, h_{m-1}$ in \Cref{line:compute_ht} of \Cref{MirrorDescent} satisfy $\|h_t\|_* \leq L$ with respect to the dual norm $\|\cdot\|_*$ for all iterations $t = 0, \ldots, m-1$. Then for any point $w \in D$, the iterates of \Cref{MirrorDescent} satisfy 
\[
\sum_{t=0}^{m-1} \langle h_t, x_t - w \rangle \le \eta \frac{L^2 m}{2 \rho} + \frac{V_{x_0}(w)}{\eta } .
\]
\end{lemma}

\subsection{Dual Certificate via Mirror Descent with Truncation}
\label{subs_MD}

In this subsection, we show how to efficiently compute a dual certificate as in \Cref{def:certificate} using mirror descent with truncation. We start with the formal definition of truncation. 

\begin{definition}[Truncation]
Given $s > 0$, we define the function $\mathrm{trunc}_{s}(\cdot): \R^V \to \R^V$ as \[ (\mathrm{trunc}_{s}(g))_p \defeq \max \{-s, g_p\}, \enspace \forall p \in V .\]  
\end{definition}
Note that for any permutation $\pi: [n] \rightarrow V$, the subgradient $g_\pi$ satisfies $f(S) \geq g_\pi(S)$ for any $S \subseteq V$. 
Unfortunately, applying truncation to $g_\pi$ doesn't preserve this property. However, the following claim shows that truncation does preserve a sparse counterpart of this property, i.e. the the second condition in \Cref{def:certificate}.

\begin{claim}
\label{claim:truncated_subgr}
Let $f : 2^V \rightarrow \R$ be a submodular function, $\pi: [n] \rightarrow V$ a permutation, and $s > 0$. Let $h \defeq \mathrm{trunc}_{s}(g_{\pi})$ be the truncated subgradient. If $f^* \ge -s + (k-1) \cdot \max_p (g_{\pi})_p$, then $f(S) \ge h(S)$ for any $k$-sparse $S \subseteq V$. 
\end{claim}

\begin{proof}
Fix any $S \subseteq V$ and consider two cases. The first case is when no $q \in S$ was truncated to get $h$. In this case $h(S) = g_{\pi}(S) \le f(S)$ for any $S \subseteq V$.

The second case is when there exists $q \in S$ so that $(g_\pi)_q$ is truncated to get $h$. In this case we must have $h_q = -s$, and therefore 
\[
h(S) \le -s + (k-1) \max_p (g_{\pi})_p \le f^* \le f(S) ,
\]
where the second inequality follows from the assumed lower bound on $f^*$.
\end{proof}

\Cref{claim:truncated_subgr} allows us to use truncated subgradients of the \lovasz extension in \Cref{MirrorDescent} to compute a $(\delta,k)$ dual certificate. 
A formal description of this procedure is given in \Cref{dual_certificate_alg}. 
The correctness and parallel depth guarantee of \Cref{dual_certificate_alg} is given in the following \Cref{lem:certif_runtime}.

\begin{algorithm}[htp!]
\caption{Mirror Descent with Truncation for Parallel Algorithm}\label{dual_certificate_alg}
    \KwData{A submodular function $f$, a sparsity parameter $k \in \mathbb{Z}_{> 0}$,  a lower bound $\optlb$, and an accuracy parameter $\delta > 0$ \tcp*{Lower bound $-\phi \leq f^*$}} 
    \KwResult{A $(\delta, k)$ dual certificate $y$ for $f$}
    \SetKwFunction{dualcert}{DualCertificate}
    \SetKwProg{Fn}{Function}{:}{}
    \Fn{\dualcert{$f, k, \optlb, \delta$}}{
    $s \gets k \|u_f\|_\infty + \optlb$ \tcp*{Truncation threshold}
    $x_0 \gets \frac{k}{n} \cdot \indicVec{V}$ \tcp*{Initial point }
    $m \gets \frac{s^2 k(k+1) \log n}{\delta^2}$ 
    \tcp*{Number of iterations}
    $\eta \gets \frac{2\sqrt{k \log n}}{s \sqrt{m (k+1)}}$ \tcp*{Step size}
    Run \textsf{MirrorDescent}$(f, S_{k+1}^V, x_0, \eta, m)$ with $h_t \gets \mathrm{trunc}_{s}(g_{x_t})$ for each iteration $t = 0,\ldots, m-1$ in \Cref{line:compute_ht} of \Cref{MirrorDescent} \label{line:call_mirror_descent}  \;
    $y = \frac{1}{m } \sum_{t=0}^{m -1} h_t$\;
    \Return $y$\;
    }
\end{algorithm}

\begin{lemma}[Mirror Descent with Truncation]
\label{lem:certif_runtime}
Given a sparsity parameter $k \in \mathbb{Z}_{> 0}$, a submodular function $f: 2^V \rightarrow \mathbb{R}$, a lower bound $- \phi \leq f^*$, and an accuracy parameter $\delta > 0$, \Cref{dual_certificate_alg} outputs a $(\delta, k)$ dual certificate $y$ in $\Tilde{O}(k^2 (k \norm{u_f}_{\infty} + \phi)^2 / \delta^2)$ parallel depth and $\Tilde{O}(n k^2 (k \norm{u_f}_{\infty} + \phi)^2 / \delta^2 \cdot \EO + \poly(n, (k \|u_f\|_\infty + \phi)/\delta)))$ time. 
\end{lemma}

\begin{proof}
Since $h_t = \mathrm{trunc}_{s}(g_{x_t})$, we have $(h_t)_p \geq -s$ for each coordinate $p \in V$. Also note that for each coordinate $p$ such that $(h_t)_p \geq 0$, we have $(h_t)_p \leq (g_{x_t})_p \leq \|u_f\|_\infty$. These imply that $\|h_t\|_\infty \leq s$ in every iteration $t$ of \textsf{MirrorDescent} in \Cref{line:call_mirror_descent} of \Cref{dual_certificate_alg}. 

Next, we note that the negative entropy function $r(x) = \sum_{i \in V} x_i \log x_i$ is $1/(k+1)$-strongly convex\footnote{This standard fact can be obtained by showing that whenever $x \in S_{k+1}^V$, the Hessian $\nabla^2 r(x) = \mathsf{diag}(1/x_1, \ldots, 1/x_{|V|})$ satisfies $v^\top \nabla^2 r(x) v \geq \|v\|_1^2/(k+1)$ for all $v \in \mathbb{R}^n$.} on $S^V_{k+1}$. 
Thus we can apply \Cref{lem:MD_guarantee}, with $\|\cdot\|$ being the $\ell_1$-norm $\|\cdot\|_1$ and the parameters $\rho = 1/(k+1)$ and $L = s$, to obtain that for every point $w \in S_{k+1}^V$, 
\[
\sum_{t=0}^{m-1} h_t^\top(x_t - w) \leq \frac{\eta s^2 m (k+1)}{2} + \frac{V_{x_0}(w)}{\eta} \leq \frac{\eta s^2 m (k+1)}{2} + \frac{2 k \log n}{\eta} .
\]
Here, the last inequality follows because $x_0 = \frac{k}{n} \cdot \indicVec{V}$ in \Cref{dual_certificate_alg}, $w \in S_k^V$, and 
 by \eqref{eq:breg_formula} that 
\begin{align*}
V_{x_0}(w) & = \sum_{i \in V} w_i \log(w_i n/k) + \sum_{i \in V} (k/n - w_i) \\
& \leq k \cdot \sum_{i \in V} (w_i/k) \log(w_i/k) + \sum_{i \in V} w_i \cdot (\log n - 1) + k \le 2k \log n .
\end{align*}
The best choice of $\eta$ above is $\eta = \frac{2\sqrt{k \log n}}{s \sqrt{m (k+1)}}$, as in \Cref{dual_certificate_alg}. Since we have also set $m = s^2 k(k+1) \log n/\delta^2$. the above bound becomes
\begin{align} \label{eq:dual_cert}
\frac{1}{m} \sum_{t=0}^{m-1} h_t^\top (x_t - w) 
\le \frac{s \sqrt{k (k+1) \log n}}{\sqrt{m }} \le \delta .
\end{align}
Now we consider the output  $y = \frac{1}{m} \sum_{t=0}^{m-1}  h_t$ of \Cref{dual_certificate_alg}. Note that $h_t \ge g_{x_t}$ for every iteration $t$, which implies that $h_t^\top x_t \ge g_{x_t}^\top x_t$ since $x_t \in [0,1]^V$. It follows that 
\[
f^* \le \frac{1}{m } \sum_{t=0}^{m-1} f(x_t) = \frac{1}{m} \sum_{t=0}^{m-1} g_{x_t}^\top x_t \le \frac{1}{m} \sum_{t=0}^{m-1} h_t^\top  x_t ,
\]
where the equality above is due to \Cref{thm:lovasz_extension_properties}. 
Now, pick $w = \indicVec{T}$, where $T$ is the set of the $k+1$ most negative coordinates of $y$, which achieves the minimum value of $y^\top w = y^{k+1}_{-}(V)$. Plugging this choice of $w$ into \eqref{eq:dual_cert}, we obtain 
\[
f^* \leq \frac{1}{m} \sum_{t=0}^{m-1} h_t^\top x_t \leq y^\top w + \delta \le y^{k+1}_{-}(V) + \delta ,
\]
which gives the first condition in \Cref{def:certificate}. To obtain the second condition, note that since $f^* \geq -\phi$, the choice of $s = k \|u_f\|_\infty + \phi$ satisfies 
\[
f^* \geq - \phi = -s + k \|u_f\|_\infty \geq -s + (k-1) \cdot \max_p (g_{x_t})_p .
\]
It then follows from \Cref{claim:truncated_subgr} that $f(S) \ge h_t(S)$ for all $\sparsity$-sparse $S \subseteq V$. This implies $y(S) = \frac{1}{m} h_t(S) \leq f(S)$ for all $\sparsity$-sparse $S$ and proves that $y$ is a $(\delta,k)$ dual certificate.

Finally, note that \Cref{dual_certificate_alg} takes $m = \Tilde{O}(k^2 (k \norm{u_f}_{\infty} + \phi)^2 / \delta^2)$ iterations and that each $h_t$ can be computed using one round of $n$ parallel $\EO$ queries to $f$. This implies that the parallel depth of \Cref{dual_certificate_alg} is $\Tilde{O}(k^2 (k \norm{u_f}_{\infty} + \phi)^2 / \delta^2)$ and its runtime is $\Tilde{O}(n k^2 (k \norm{u_f}_{\infty} + \phi)^2 / \delta^2 \cdot \EO + \poly(n, (k \|u_f\|_\infty + \phi)/\delta))$. 
\end{proof}

\subsection{Dimensionality Reduction for Parallel Algorithm}
\label{subsec:dimensionality_reduction_parallel}

In this subsection, using the algorithm for computing $(\delta,k)$ dual certificates given in \Cref{subs_MD}, we give an implementation of the procedure $\textsf{Dimensionality-Reduction}$ for our parallel algorithm in \Cref{alg:dim_red_parallel}.
Recall from \Cref{sec:framework:meta-algorithm} that assuming $f$ has a $\sparsity$-sparse minimizer, the procedure \textsf{Dimensionality-Reduction}$(f,k)$ either outputs a subset $T \neq \emptyset$ that belongs to every minimizer of $f$, or certify that $f^* > \|u_f\|_\infty / 4$.

\Cref{alg:dim_red_parallel} aims at finding a value $\phi>0$ such that $-\phi \leq f^* \leq -\phi/2$. 
Since $f^* > - (\|u_f\|_1 - f(V))$ by submodularity, \Cref{alg:dim_red_parallel} starts by guessing $\phi = \phi^{(0)} \defeq \|u_f\|_1 - f(V)$ in \Cref{line:init_guess} and keeps decreasing the value of $\phi$ by a factor of 2 in each iteration of the while loop in \Cref{line:solvescale_whileloop} until $\phi$ falls under the threshold $\|u_f\|_\infty/4$. In the case where $f^* \leq -\|u_f\|_\infty/4$, then along this halving process there must be an iteration $i$ where the value $\phi = \phi^{(i)}$ satisfies $-\phi^{(i)}\leq f^* \leq -\phi^{(i)}/2$. The following claim shows that in such an iteration, the set $T$ in \Cref{line:find_dim_red_parallel} is a non-emtpy dimensionality reduction. 

\begin{claim}[Correctness of \Cref{alg:dim_red_parallel}]
\label{claim:correctness_dim_reduction_parallel}
In any iteration of the while loop of \Cref{alg:dim_red_parallel}, the set $T$ in \Cref{line:find_dim_red_parallel} lies in every minimizer of $f$. Moreover, if the value $\phi = \phi^{(i)}$ in an iteration $i$ satisfies $-\phi^{(i)} \leq f^* \leq -\phi^{(i)}/2$, then the set $T$ in \Cref{line:find_dim_red_parallel} is non-empty in that iteration. 
\end{claim}
\begin{proof}
Since $y$ is a $(\delta, k)$ dual certificate in each iteration of the while loop, and the set $T$ in \Cref{line:find_dim_red_parallel} is consisted of all the element $p \in V$ with $y_p < -\delta$, the first statement of the claim immediately follows from \Cref{claim:certif_dim_red}. For the second statement, note that if $-\phi^{(i)} \leq f^* \leq -\phi^{(i)}/2$ and $\delta = \phi^{(i)} / 3k$, it follows that $\delta < |f^*|/k$. So the second part of \Cref{claim:certif_dim_red} implies that $T \neq \emptyset$. 
\end{proof}

\Cref{claim:correctness_dim_reduction_parallel} implies that each non-empty set $T$ in \Cref{line:find_dim_red_parallel} found in the while loop lies in every $\sparsity$-sparse minimizer of $f$. 
If no such non-empty set $T$ is found throughout the while loop, then by \Cref{claim:correctness_dim_reduction_parallel} this means $f^* > \|u_f\|_\infty/4$, so \Cref{alg:dim_red_parallel} will simply output $\emptyset$.

\begin{algorithm}[htp!]
\caption{Dimensionality Reduction for Parallel Algorithm}\label{alg:dim_red_parallel}
    \KwData{A sparsity parameter $k$, and a submodular function $f$ with a $\sparsity$-sparse minimizer}
    \KwResult{A subset $T \subseteq V$ that must be in every minimizer of $f$, or $T = \emptyset$ certifying that $f^* > - \|u_f\|_\infty/4$}
    \SetKwFunction{dimred}{Dimensionality-Reduction}
    \SetKwProg{Fn}{Function}{:}{}
    \Fn{\dimred{$f, k$}}{
        $\phi \gets \|u_f\|_1 - f(V)$ \label{line:init_guess} \tcp*{Lower bound $f^* \geq -(\|u_f\|_1 - f(V))$}
        \While(\tcp*[f]{Implement while loop in parallel} ){$\phi \geq \|u_f\|_\infty / 4$ \label{line:solvescale_whileloop}}{
        $\delta = \frac{\phi}{3k}$\;
         $y \gets$ \dualcert{$f, k, \phi, \delta$} \tcp*{$y$ is a $(\delta,k)$ dual certificate}
         $T \gets \{p: y_p < -\delta\}$ \label{line:find_dim_red_parallel} \;
        \lIf(\tcp*[f]{When $-\phi \leq f^* \leq -\phi/2$ then $T \neq \emptyset$}){$T \ne \emptyset$}{
                \Return $T$ 
        }
        $\phi \gets \phi / 2$ \tcp*{$T = \emptyset$ indicates $f^* > -\phi/2$}
        }
         \Return $\emptyset$\; 
    }
\end{algorithm}

The correctness and runtime guarantees of \Cref{alg:dim_red_parallel} is formally given in the following lemma.

\begin{lemma}[$\textsf{Dimensionality-Reduction}$ for Parallel Algorithm]
\label{lem:dim_red_parallel}
    Let $k \in \mathbb{Z}_{>0}$ and $f : 2^V \rightarrow \R$ be a submodular function with a $k$-sparse minimizer. Then \Cref{alg:dim_red_parallel} outputs a set $T \subseteq V$ that lies in any minimizer such that $T = \emptyset$ implies $f^* > -\|u_f\|_\infty/4$. 
    Moreover, the algorithm uses $\Tilde{O}(k^6)$ parallel depth and runs in time $\Tilde{O}(n k^6 \cdot \EO + \poly(n))$. 
\end{lemma}

\begin{proof}
The correctness of \Cref{alg:dim_red_parallel} follows immediately from \Cref{claim:correctness_dim_reduction_parallel}, so we only need to prove the bound on depth and runtime. 

By \Cref{lem:certif_runtime}, the number of parallel rounds due to \textsf{DualCertificate} calls is 
\[
\Tilde{O}(k^2 (k \norm{u_f}_{\infty} + \phi)^2 / \delta^2) = \Tilde{O}\left(k^4 \left(\frac{k \norm{u_f}_{\infty} + \phi}{\phi} \right)^2 \right) \leq \Tilde{O}(k^6) ,
\]
where the last inequality is because we always have $\phi \geq \|u_f\|_\infty / 4$ and we always run with accuracy parameter $\delta = \frac{\phi}{3k}$. Note that there are at most $\log((\|u_f\|_1 - f(V)) / \|u_f\|_\infty) = \widetilde{O}(1)$ iterations of the while loop and these can be implemented in parallel. It follows that \Cref{lem:certif_runtime} also only uses $\Tilde{O}(k^6)$ parallel depth, $\Tilde{O}(n k^6)$ queries to $\EO$, and $\poly(n)$ runtime.
\end{proof}

\subsection{Arc Finding for Parallel Algorithm}
\label{subsec:arc_finding_parallel}

In this subsection, we describe the procedure $\textsf{Arc-Finding}$ for the parallel algorithm (formally given in \Cref{alg:arc_finding_parallel}), which is based on the procedure $\textsf{Dimensionality-Reduction}$ from the previous \Cref{subsec:dimensionality_reduction_parallel}. 
For each element $p \in V$ such that $f(p^\downarrow) \geq \mathsf{Scale}/2$, \Cref{alg:arc_finding_parallel} simply runs $\textsf{Dimensionality-Reduction}$ in \Cref{alg:dim_red_parallel} with the contracted submodular function $f^{\sharp \rfam}_{p^\downarrow}:2^{V \setminus p^\downarrow} \rightarrow \mathbb{R}$, which was defined in \Cref{subsec:prelim_submodular} as
\[
f^{\sharp \rfam}_{p^\downarrow}(S) \defeq f^{\sharp \rfam}(S \cup p^\downarrow) - f^{\sharp \rfam}(p^\downarrow) ,
\]
and the remaining sparsity $k - |p^\downarrow|$. This subset $S_p$ returned from $\textsf{Dimensionality-Reduction}$ will be the set of endpoints of arcs we found for the element $p$. 

\begin{algorithm}[htp!]
\caption{Arc Finding for Parallel Algorithm}\label{alg:arc_finding_parallel}
    \KwData{A sparsity parameter $k$, extension $f^{\sharp \rfam}$ of a submodular function $f$ with a $\sparsity$-sparse minimizer, and parameter $\mathsf{Scale} \geq \|u_{f^{\sharp \rfam}}\|_\infty \geq -4 f^*$ \tcp*{Arcs are input through $f^{\sharp \rfam}$}}
    \KwResult{A non-empty set $S_p \subseteq V$ of endpoints of arcs from every $p\in V$ s.t. $f(p^\downarrow) \geq \mathsf{Scale}/2$}
    \SetKwFunction{dimred}{Arc-Finding}
    \SetKwProg{Fn}{Function}{:}{}
    \Fn{\dimred{$f^{\sharp \rfam}, k, \mathsf{Scale}$}}{
        \For(\tcp*[f]{Implement the for loop in parallel }){$p \in V$}{
        \If{$(u_{f^{\sharp \rfam}})_p \geq \mathsf{Scale}/2$}
        {$S_p \leftarrow \textsf{Dimensionality-Reduction}(f^{\sharp \rfam}_{p^\downarrow}, k - |p^\downarrow|)$ \tcp*{Find arcs from $p$ with large $(u_{f^{\sharp \rfam}})_p$}}
        \lElse{$S_p \leftarrow \emptyset$
        }
        }
         \Return $\{S_p\}_{p \in V}$\; 
    }
\end{algorithm}

The following lemma summarizes the guarantees and runtime of the procedure $\textsf{Arc-Finding}$.

\begin{lemma}[$\textsf{Arc-Finding}$ for Parallel Algorithm]
\label{lem:arc_finding_parallel}
    Let $k \in \mathbb{Z}_{>0}$, $f : 2^V \rightarrow \R$ be a submodular function with a $k$-sparse minimizer, $f^{\sharp \rfam}$ be its extension w.r.t. a ring family $\rfam(E, W, D)$ that is $k$-consistent, and $\mathsf{Scale} \geq \|u_{f^{\sharp \rfam}}\|_\infty \geq -4 f^*$. 
    Then \Cref{alg:arc_finding_parallel} outputs, for every $p \in V$ with $f(p^\downarrow) \geq \mathsf{Scale}/2$ such that $p$ belongs to a $k$-sparse minimizer, a non-empty set $S_p \subseteq V \setminus p^\downarrow$ of endpoints of arcs from $p$. 
    Moreover, the algorithm uses $\Tilde{O}(k^6)$ parallel depth and $\Tilde{O}(n^2 k^6 \cdot \EO + \poly(n))$ time.
\end{lemma}
\begin{proof}
We start with the first statement of the lemma. Fix any element $p\in V$ that lies in a $\sparsity$-sparse minimizer $S^*$ such that $(u_{f^{\sharp \rfam}})_p \geq \mathsf{Scale} / 2$. By \Cref{lem:properties_extension}, we have 
\[
(u_{f^{\sharp \rfam}})_p = f(W \cup p^\downarrow) - f(W \cup p^\downarrow \setminus \{p\}) = f^{\sharp \rfam}(p^\downarrow) - f(W \cup p^\downarrow \setminus \{p\}) , 
\]
where the last equality follows since $(p^\downarrow)^{\sharp \rfam} = p^\downarrow$. Since $f^* \geq - \mathsf{Scale}/4$ by assumption, we have $f(W \cup p^\downarrow \setminus \{p\}) \geq  - \mathsf{Scale}/4$ and therefore,
\[
f^{\sharp \rfam}(p^\downarrow) = (u_{f^{\sharp \rfam}})_p + f(p^ \downarrow \setminus \{p\}) \geq \mathsf{Scale}/4 . 
\]
We also have for any element $q \in V \setminus p^\downarrow$, 
\[
(u_{f^{\sharp \rfam}_{p^{\downarrow}}})_q = f^{\sharp \rfam}_{p^{\downarrow}}(\{q\}) - f^{\sharp \rfam}_{p^{\downarrow}}(\emptyset) = f^{\sharp \rfam}(\{q\} \cup p^{\downarrow}) - f^{\sharp \rfam}(p^\downarrow) \leq (u_{f^{\sharp \rfam}})_q  \leq  \mathsf{Scale} ,
\]
where the first inequality uses submodularity, and the last inequality is by our assumption that $\|u_{f^{\sharp \rfam}}\|_\infty \leq \mathsf{Scale}$. 
It then follows that assuming $u_{f^{\sharp \rfam}_{p^{\downarrow}}} \geq 0$, the contracted function $f^{\sharp \rfam}_{p^\downarrow}$ satisfies (1) $f^{\sharp \rfam}_{p^\downarrow}$ has a $(k-|p^\downarrow|)$-sparse minimizer $S^* \setminus p^\downarrow$, and (2) its minimum value $(f^{\sharp \rfam}_{p^\downarrow})^* \leq - \mathsf{Scale}/4 \leq -\|u_{f^{\sharp \rfam}_{p^{\downarrow}}}\|_\infty / 4$. \Cref{lem:dim_red_parallel} then implies that a non-empty set $S_p \subseteq V \setminus p^\downarrow$ of endpoints of arcs from $p$ can be found for such an element $p$ using $\widetilde{O}(k^6)$ parallel depth and $\widetilde{O}(n k^6 \cdot \EO + \poly(n))$ runtime. 

The second statement of the lemma follows immediately since for each one of the $n$ elements $p \in V$, the for loop in \Cref{alg:arc_finding_parallel} can be implemented in parallel. This proves the lemma.  
\end{proof}

\subsection{Putting It All Together: Proof of \Cref{thm:main_parallel}}
\label{subs_correctness}

Finally, we are ready to prove \Cref{thm:main_parallel}, which we restate below for convenience. 

\parallel*

\begin{proof}[Proof of \Cref{thm:main_parallel}]
Consider the meta algorithm (\Cref{alg:main}) with procedures \textsf{Dimensionality-Reduction} and \textsf{Arc-Finding} given in Algorithms \ref{alg:dim_red_parallel} and \ref{alg:arc_finding_parallel}. 
The correctness of \Cref{alg:main} is already given in \Cref{cor:correctness_meta_algorithm} so we only need to analyze its parallel depth and runtime. 

Note that in each iteration of the outer while loop in \Cref{line:outer_while_main} of \Cref{alg:main}, one of the following three things will happen: (1) the size of the contracted elements $W$ will increase due to \textsf{Dimensionality-Reduction} in \Cref{line:solve_scale_sub-procedure} when $f^* \leq -\|u_{f^{\sharp \rfam}}\|_\infty / 4$, or (2) the size of the contracted elements $W$ will increase because there exists element $p \in V\setminus (W \cup D)$ such that $(u_{f^{\sharp \rfam}})_p < 0$ in \Cref{line:negative_upper} of \Cref{alg:extension_maintainer} (extension maintainer), or (3) $\|u_{f^{\sharp \rfam}}\|_\infty$ will decrease by a factor of $2$, and a set $S_p$ of endpoints of arcs is found for every element $p \in V \setminus (D \cup W)$ with $(u_{f^{\sharp \rfam}})_p > \mathsf{Scale}/2$ in the while loop in \Cref{line:arc_finding_main} when $f^* > -\|u_{f^{\sharp \rfam}}\|_\infty / 4$. 
Note that (1) and (2) can happen at most $k$ times before $|W| \geq k$ and \Cref{alg:main} outputs $W$; (3) can happen at most $\log (|f| n/ \epsilon)$ times before $\|u_{f^{\sharp \rfam}}\|_\infty \leq \epsilon / n$. 
So the total number of iterations of the while loop in \Cref{line:outer_while_main} will be at most $O(k + \log(|f| n / \epsilon))$.

We next bound the parallel depth and runtime due to \textsf{Dimensionality-Reduction}. Note that the total number of times \textsf{Dimensionality-Reduction} is called in \Cref{line:solve_scale_sub-procedure} is at most $k + \log(|f| n / \epsilon)$ by the above. By \Cref{lem:dim_red_parallel}, each \textsf{Dimensionality-Reduction} for (1) can be done in $\widetilde{O}(k^6)$ depth and runtime $\widetilde{O}(n k^6 \cdot \EO + \poly(n))$. 
So the total depth and runtime due to \textsf{Dimensionality-Reduction} is $\widetilde{O}(k^7 \log(|f| /\epsilon))$ and $\widetilde{O}(n k^7 \cdot \EO + \poly(n)) \cdot \log(|f| /\epsilon)$ respectively. 

Next, we bound the parallel depth and runtime due to \textsf{Arc-Finding}. By \Cref{lem:arc_finding_parallel}, each call to \textsf{Arc-Finding} takes $\widetilde{O}(k^6)$ parallel depth and runtime $\widetilde{O}(n^2 k^6 \cdot \EO + \poly(n))$. We perform at most $k$ calls to \textsf{Arc-Finding} before $\|u_{f^{\sharp \rfam}}\|_\infty$ decreases by a factor of $2$, so the total depth and runtime due to \textsf{Arc-Finding} is $\widetilde{O}(k^7 \log(|f| /\epsilon))$ and $\widetilde{O}(n k^7 \cdot \EO + \poly(n)) \cdot \log(|f| /\epsilon)$ respectively.  

Finally, note that each update to the \textsf{RingFamily} can be implemented in $\widetilde{O}(k)$ depth and $O(m \cdot \EO + nk)$ time where $m$ is the total number of elements $p$ from which arcs are found. Additionally, each one of \Cref{line:update_infos1}, \Cref{line:update_infos2} and \Cref{line:update_infos3} can be implemented in $O(1)$ depth and $O(n \cdot \EO + n)$ time. 
Combining everything above, \Cref{alg:main} finds an $\epsilon$-approximate minimizer for $\sparsity$-sparse SFM in parallel depth $\widetilde{O}(k^7 \log(|f|/\epsilon))$ and runtime $\widetilde{O}(n^2 k^7 \log (|f|/\epsilon) \cdot \EO + \poly(n) \cdot \log (|f| / \epsilon))$.
\end{proof}

%% file: sequential_algo_dual_certificate.tex
\section{$\otilde(n \cdot \poly(k))$-Query Algorithm for $k$-sparse SFM}
\label{sec:sequential}

In this section, we present our randomized sequential algorithm and prove Theorems \ref{thm:main_sequential} and \ref{thm:main_strongly_sequential}. The section is organized as follows. In \Cref{subs_stMD}, we present our optimization method for computing $(\delta, \sparsity)$ dual certificates via a stochastic version of the follow-the-regularized-leader\footnote{The FTRL algorithm is also referred to as lazy mirror descent or dual averaging in the literature.} (FTRL) algorithm (e.g., \cite{N03book,SS07}) where the subgradient implementation is  tailored to submodular structure. In \Cref{subsec:dimensionality_reduction_sequential} and \Cref{subsec:arc_finding_sequential}, we present our sequential implementations of the subprocedures \textsf{Dimensionality-Reduction} and \textsf{Arc-Finding} in \Cref{alg:main} respectively. Finally, in \Cref{subs_stcorrectness}, we present the proofs of Theorems \ref{thm:main_sequential} and \ref{thm:main_strongly_sequential}.

\subsection{Dual Certificate via Stochastic Follow-the-Regularized-Leader}
\label{subs_stMD}

In this subsection, we present our sequential algorithm (see Algorithms \ref{st_dual_cert} and \ref{submod_st_FTRL}) that efficiently computes a $(\delta, k)$ dual certificate. 
The main result of this subsection is the following theorem.

\begin{restatable}[Stochastic Dual Certificate]{theorem}{StochDualCertificate}
\label{lem:st_certif_runtime_FTRL}
Let $f$ be a submodular function, $k \in \mathbb{Z}_{>0}$ a sparsity parameter, 
$\phi \ge |f^*|$ and $0 < \delta \leq \phi$ an accuracy parameter. 
Assuming $\phi = \Omega(\|u_f\|_\infty / k)$, then 
\Cref{st_dual_cert} outputs a set of permutations $\{\pi^{(t)}\}_{t \in [m]}$ in time $\widetilde{O}(m) \cdot ( \EO + \poly(n))$, where 
\[
m = \otilde(k^6 \delta^{-4} \phi^2 (\|u_f\|_{\infty} + \phi) (\|u_f\|_1 + \phi)), 
\]
such that $y \defeq \frac{1}{m} \sum_{t \in [m]} g_{\pi^{(t)}}$ is a $(\delta,k)$ dual certificate for $f$ with high probability. 
\end{restatable}

The idea behind \Cref{lem:st_certif_runtime_FTRL} is to run a variant of stochastic follow-the-regularized-leader (stochastic FTRL)
with an entropy regularizer whose process of generating stochastic subgradients is tailored to submodular structure (without subgradient truncation). 
While such a method, without controlling the $\ell_\infty$-norm of the subgradient via truncation as was done in the previous section, might seem too slow, a more careful local
norm analysis of the variance (\Cref{lem:eff_vsampling}) shows how this can be made to work. 
In particular, one of the main novelties of our algorithm is a sampling method we call $\mathsf{vSampling}$ (\Cref{defn:vSampling}), which helps us sample low-variance unbiased stochastic subgradients in only $\otilde(1)$ queries. 

We emphasize that while \Cref{dual_certificate_alg}  in \Cref{sec:parallel}  returns a $(\delta, k)$ dual certificate explicitly, \Cref{st_dual_cert} only returns a dual certificate implicitly as a set of permutations $\{\pi^{(t)}\}_{t \in [m]}$ with the property that $y \defeq \frac{1}{m}\sum_{t \in [m]} g_{\pi^{(t)}}$ is a $(\delta, k)$ dual certificate with high probability. 
This is because computing the dual certificate $y$ from the set of permutations $\{\pi^{(t)}\}_{t \in [m]}$ takes $\otilde(m n)$ evaluation oracle queries, where $m$ can be bigger than $n$. Nevertheless, as we show later (see \Cref{subsec:arc_finding_sequential}), it suffices to sample coordinates of each $g_{\pi^{(t)}}$ to access the dual certificate much more efficiently. 

Now we formally describe \Cref{st_dual_cert}. 
This algorithm repeatedly runs \Cref{submod_st_FTRL} on the Lov\'asz extension  
and outputs the union of all permutations correponding to the iterates generated. As discussed earlier, \Cref{submod_st_FTRL} is a variant of the
stochastic FTRL algorithm (\Cref{st_FTRL}), with the stochastic approximate subgradients generated via the $\mathsf{vSampling}$ method (\Cref{defn:vSampling}) in Lines \ref{line:vsample_begin}-\ref{line:vsample_end}. The convergence guarantee for executing \Cref{submod_st_FTRL} once is a statement that holds in expectation (\Cref{lem:submodularFTRL_expectation}). The reason behind calling \Cref{submod_st_FTRL} repeatedly in \Cref{st_dual_cert} is to obtain an analogous guarantee with high probability, which is crucial to proving that the average of the subgradients output by \Cref{st_dual_cert} is a sparse dual certificate.

\begin{algorithm}[htp!]
\caption{Stochastic Dual Certificate for Sequential Algorithm}\label{st_dual_cert}
    \KwData{A sparsity parameter $k \in \mathbb{Z}_{> 0}$, a submodular function $f$ with $u_f \in \R^V_{\geq 0}$, a parameter $\phi \geq |f^*|$, and an accuracy parameter $0< \delta \leq \phi$} 
    \KwResult{A set of $m = \otilde(k^6 \delta^{-4} \phi^2 (\|u_f\|_{\infty} + \phi) (\|u_f\|_1 + \phi))$ permutations $\{\pi^{(t)}\}_{t \in [m]}$ s.t.  $y = \frac{1}{m} \sum_{t\in[m]}g_{\pi^{(t)}}$ is a $(\delta, k)$ dual certificate whp.}
    \SetKwFunction{stdualcert}{StochDualCertificate}
    \SetKwFunction{submodstFTRL}{SubmodularFTRL}
    \SetKwProg{Fn}{Function}{:}{}
    \Fn{\stdualcert{$f, k, \phi, \delta$}}{
    $N \gets \widetilde{O}(\delta^{-2} k^5 \phi^2)$ \;
    \For{$\ell = 1, \ldots, N$}{
        $\mathcal{C}_\ell \gets$\submodstFTRL{$f, k, \phi, \delta/2$} \tcp*{$\mathcal{C}_\ell$ is a collection of permutations} \label{line:StochDualCert_Rep}
    }
    \Return $\bigcup_{\ell \in [N]} \mathcal{C}_\ell$\;
    }
\end{algorithm}

\begin{algorithm}[htp!]
\caption{Stochastic Submodular Follow-the-Regularized-Leader (FTRL)}\label{submod_st_FTRL}
    \KwData{A sparsity parameter $k \in \mathbb{Z}_{> 0}$, a submodular function $f$ with $u_f \in \R^V_{\geq 0}$, a parameter $\phi \geq |f^*|$, and an accuracy parameter $0 < \delta \leq \phi$}
    \KwResult{A set of $M = \widetilde{O}(\delta^{-2} (k\|u_f\|_{\infty} + \phi) (\|u_f\|_1 + \phi))$ permutations $\{\pi^{(t)}\}_ {t \in [M]}$} 
    \SetKwFunction{submodstFTRL}{SubmodularFTRL}
    \SetKwProg{Fn}{Function}{:}{}
    \Fn{\submodstFTRL{$f, k, \phi, \delta$}}{
    $x_0 \gets \frac{k}{n} \indicVec V$\;
    $U_\infty \gets 2k\|u_f\|_{\infty} + \phi$ and $U_1 \gets 2\|u_f\|_1 + \phi$ \tcp*{Useful quantities in \Cref{lem:eff_vsampling}}
    $M \gets \frac{U_\infty U_1 \log n}{\delta^2}$\;
    $\eta \gets \sqrt{\frac{k \log n}{M U_\infty U_1}} = \frac{\delta}{U_\infty U_1}$\;
    \For{$t = 0, 1, \ldots, M-1$}{
        Sample $j \propto p^{(t)}_{j} = \frac{v^{(t)}_j}{\|v^{(t)}\|_1}$, where $v^{(t)}_i = 2 (u_f)_i - (g_{x_t})_i$
        for all $i \in V$ \label{line:vsample_begin}\;
        $h_t \gets \indicVec j \cdot \frac{(g_{x_t})_j}{p^{(t)}_{j}}$ \label{line:vsample_end} \tcp*{The $\mathsf{vSampling}(f,x_t)$ method as in \Cref{defn:vSampling}}
        $x_{t+1} = \argmin_{x \in S_k^V} \eta \sum_{t'=0}^t h_{t'}^\top x +  r(x)$ \label{line:prox_step_submodular_FTRL}\;
    }
    \Return $\{\pi_{x_t}\}_{t =0}^{M-1}$\;
    }    
\end{algorithm}

Now we formally define the $\mathsf{vSampling}$ method that  is used to generate the stochastic subgradients $h_t$ in Lines \ref{line:vsample_begin}-\ref{line:vsample_end} of \Cref{submod_st_FTRL}. 

\begin{definition}[The $\mathsf{vSampling}$ method] \label{defn:vSampling}
Given a submodular function $f: V \rightarrow \R$ such that $u_f \in \R^V_+$ and a point $x \in S_k^V$. Define vector $v \in \R^V_+$ as $v_i \defeq 2 (u_f)_i - (g_x)_i$. 
Then $\mathsf{vSampling}(f,x)$ samples a coordinate $j \in V$ with probability proportional to $p_j \defeq \frac{v_j}{\norm{v}_1}$ and returns the random vector $\indicVec j \cdot (g_x)_j p_{j}^{-1}$. 
\end{definition}

Let us briefly mention the motivation behind the $\mathsf{vSampling}$ method defined above. It is known that if one wants to sample a $1$-sparse unbiased estimator $h$ of the subgradient vector $g_x \in \R^V$, then sampling a coordinate $j \in V$ proportional to $|(g_x)_j|$ achieves the smallest second moment\footnote{It doesn't matter which norm we measure $h$ here since it is $1$-sparse.} $\E[\|h\|_\infty^2]$. 
However, this sampling method requires explicitly computing the values of all $|(g_x)_j|$, which takes $O(n)$ queries and is unfortunately too expensive. 
The main purpose of oversampling (with probabilities proportional to $v_i = 2(u_f)_i - (g_x)_i$ in $\mathsf{vSampling}$ is to make the sampling procedure more efficient while not significantly increasing the second moment. 
In particular, while it is prohibitively expensive to compute all the $|(g_x)_i|$, one can efficiently compute $\sum_{i \in I} (g_x)_i$ for a consecutive block of coordinates $I$ in the permutation $\pi_x$ using $O(1)$ queries. Therefore, leveraging a binary search idea, one can efficiently sample proportional to $v_i$ by computing $\sum_{i \in I} v_i$ for $O(\log n)$ consecutive blocks $I$, each of which takes only $O(1)$ queries. 

Formally, the following lemma provides upper bounds on the local norm of the stochastic subgradients generated by the $\mathsf{vSampling}$ method and the runtime of the method. 

\begin{lemma}
\label{lem:eff_vsampling} 
    Let $f: 2^V \rightarrow \R$ be a submodular function such that $u_f \in \R^V_+$, $\phi > 0$ satisfies $|f^*| \leq \phi$, and $x \in S_k^V$. 
  Then $\mathsf{vSampling}(f,x)$ given in \Cref{defn:vSampling} outputs a $1$-sparse random vector $h \defeq \indicVec j (g_x)_j p_{j}^{-1}$
  that satisfies $\E[h] = g_{x}$ and 
    \[
    \E[\|h\|_{x}^2] 
    \le (2k \|u_f\|_{\infty} - f^*) (2 \|u_f\|_1 - f(V)) \leq U_\infty \cdot U_1
    \,,\]
where $U_\infty \defeq 2k \|u_f\|_{\infty} + \phi$ and $U_1 \defeq 2 \|u_f\|_1 + \phi$. 
    Moreover, given the vector $u_f$ explicitly, $\mathsf{vSampling}$ can be implemented in time $O(\log n \cdot (\EO + n))$ time. 
\end{lemma}

\begin{proof}[Proof of \Cref{lem:eff_vsampling}]
    First, note that \[
    \E[h] = \sum_{j \in V} p_j \cdot \indicVec j (g_x)_j p_j^{-1} = g_{x}.
    \]
    Hence, it suffices to prove the bound on the variance of the stochastic subgradient $h$. Note that
    \begin{align*}
    \E[\|h\|_x^{2}] 
    = \sum_{i\in V} p_{i} x_{i} ((g_x)_{i}p_{i}^{-1})^{2}
    = \sum_{i\in V} x_{i} |(g_x)_{i}|\cdot \frac{|(g_x)_{i}|}{p_{i}}
    \leq |g_x|^{\top}x \cdot \max_{i \in V} \frac{|(g_x)_{i}|}{p_{i}}\, ,
    \end{align*}
    where $|g_x|$ is the vector with $(|g_x|)_i \defeq |(g_x)_i|$ for each coordinate $i \in V$. 
    Let vector $v \in \R^V$ be given by 
    $v_i \defeq 2 (u_f)_i - (g_x)_i$ for each $ i \in V$, for which $p_i = \frac{v_i}{\|v\|_1}$. Since $v_i \geq |(g_x)_i|$ for all $i \in V$, we have $\max_{i \in V} \frac{|(g_x)_{i}|}{p_{i}} = \max_{i \in V} \frac{|(g_x)_{i}| \cdot \|v\|_1}{v_i} \leq \|v\|_1$. Hence, we obtain 
    \[
    \E[\|h\|_x^{2}] \le |g_x|^{\top}x \cdot \|v\|_1 = |g_x|^{\top}x \cdot (2 \|u_f\|_1 - f(V)) \leq |g_x|^{\top}x \cdot U_1 .
    \]
    Let $g_{+} \defeq \max\{g_x,\vzero\}$ and $g_{-} \defeq \min\{g_x,\vzero\}$ entrywise. Since $f(x) = g_x^{\top}x$,  
    it follows that \begin{equation}\label{ineq:abs_val_inner_prod}
    |g_x|^{\top}x = g_{+}^{\top}x - g_{-}^{\top}x = 2g_{+}^{\top}x- f(x)
    \leq  2u_f^{\top}x - f^* \leq 2k \|u_f\|_\infty - f^* \leq U_\infty ,
    \end{equation}
    where the first inequality uses $(g_x)_i \leq (u_f)_i$ and $x_i \geq 0$ for every
    $i \in V$ and the last step follows from $\norm x_{1}\leq k$. This completes the proof of the first statement of the lemma. 
    
    Finally, we describe how to implement the $\mathsf{vSampling}$ procedure. For every $i \in [|V|]$, 
    \[
    \sum_{j \in \pi_x[i]} v_j = \Big ( \sum_{j \in \pi_x[i]} 2(u_f)_j \Big ) - f(\pi_x[i]) ,
    \]
    which can be computed in time $\EO + O(n)$ given the vector $u_f$ and the set $\pi_x[i]$. This allows us to sort the coordinates in $V$ according to $\pi_x$ and use binary search to sample the coordinate $j \sim p_j$ in $O(\log n)$ iterations. The total runtime for $\mathsf{vSampling}$ is therefore $O(\log n \cdot (\EO + n))$. 
\end{proof}

The next lemma gives an expectation bound on the regret for \Cref{submod_st_FTRL}.  

\begin{lemma}[Expected Regret Bound for \textsf{SubmodularFTRL}] \label{lem:submodularFTRL_expectation}
In an execution of \Cref{submod_st_FTRL}, the random vector $y = \frac{1}{M} \sum_{t=0}^{M-1} g_{x_t}$ satisfies that, for every $w \in S_k^V$ (independent of the randomness of the algorithm), $\E[\langle y, w \rangle] \geq f^* - \delta$. 
\end{lemma}

\begin{proof}
Note that \Cref{submod_st_FTRL} is an instantiation of \Cref{st_FTRL} with Lov\'asz extension $f$, domain $D = S_k^V$, step size $\eta$, number of iterations $M$, the entropy regularizer $r(x)$, and the stochastic subgradient $h_t$ generated as in \Cref{line:vsample_begin} - \ref{line:vsample_end}. 
Therefore, as discussed in \Cref{sec:approach:sequential}, the local norm analysis for stochastic FTRL (see the first statement of \Cref{lem:stMD_guarantee}) gives 
\[
\E\left[\sum_{t=0}^{M-1} \langle g_{x_t}, x_{t} - w \rangle\right] \le \frac{\sup_{x \in S_k^V} r(x) - \inf_{y \in S_k^V} r(y)}{\eta} + \eta \sum_{t=0}^{M-1} \E[\|h_t\|_{x_t}^2] .
\]
Since $\sup_{x \in S_k^V} r(x) = 0$ and $\inf_{y \in S_k^V} = - k \log(n/k)$, we can bound $\sup_{x \in S_k^V} r(x) - \inf_{y \in S_k^V} r(y) \leq k \log n$. Using \Cref{lem:eff_vsampling}, we can bound 
\[
\E[\|h_t\|_{x_t}^2] \leq (2k \|u_f\|_{\infty} - f^*) (2 \|u_f\|_1 - f(V)) \leq(2k \|u_f\|_{\infty} + \phi) (2 \|u_f\|_1 + \phi) ,
\]
where we used $\phi \geq |f^*|$. 
Dividing by $M$ on both sides, the above bound then becomes 
\[
\E\left[\Big(\frac{1}{ M}\sum_{t=0}^{M-1} \langle g_{x_t}, x_{t} \rangle \Big) - \langle y, w \rangle \right] \leq \frac{k \log n}{\eta M} + \eta (2k \|u_f\|_{\infty} + \phi) (2 \|u_f\|_1 + \phi).
\]
By optimally setting $\eta = \sqrt{\frac{k \log n}{M (2k \|u_f\|_{\infty} + \phi) (2 \|u_f\|_1 + \phi)}}$ above, we obtain 
\[
\E\left[\Big(\frac{1}{ M}\sum_{t=0}^{M-1} \langle g_{x_t}, x_{t} \rangle \Big) - \langle y, w \rangle \right] \leq \sqrt{\frac{k \log n}{M} \cdot (2k \|u_f\|_{\infty} + \phi) (2 \|u_f\|_1 + \phi)}  = \delta , 
\]
where the last equality is because of our setting $M = \frac{k (2k\|u_f\|_{\infty} + \phi) (2\|u_f\|_1 + \phi) \log n}{\delta^2}$. 
Then we have
\[
\E[\langle y, w \rangle] \geq \E\left[\frac{1}{ M}\sum_{t=0}^{M-1} \langle g_{x_t}, x_{t} \rangle \right] - \delta = \E\left[\frac{1}{ M}\sum_{t=0}^{M-1} f(x_t) \right] - \delta \geq f^* - \delta ,
\]
where the equality uses \Cref{thm:lovasz_extension_properties}. This completes the proof of the lemma. 
\end{proof}

Ideally, we would like to take $w$ to be the indicator vector corresponding to the $k$ most negative coordinates of $y$ in \Cref{lem:submodularFTRL_expectation} to argue that $y$ is a $(\delta, k)$ dual certificate. Unfortunately, \Cref{lem:submodularFTRL_expectation} only works when we fix $w$ ahead of time without depending on $y$. Therefore, to prove \Cref{lem:st_certif_runtime_FTRL}, we turn the expectation bound in \Cref{lem:submodularFTRL_expectation} to a high probability bound, from which we can then union bound over all $k$-sparse indicator vectors $w$. To this end, we start with the following high probability bound on $\|y\|_\infty$. 

\begin{lemma}[High Probability $\ell_\infty$ Bound for \textsf{SubmodularFTRL}]\label{lem:submodularFTRL_highprob_linfty}
In an execution of \Cref{submod_st_FTRL}, the random vector $y = \frac{1}{M} \sum_{t=0}^{M-1} g_{x_t}$ satisfies that 
\[
\|y\|_\infty \leq \|u_f\|_\infty + \widetilde{O}\Big(\phi + \frac{k \delta^2}{U_\infty} \Big) ,
\]
with probability at least $1 - \frac{1}{n^{C}}$, where $U_\infty \defeq 2k \|u_f\|_\infty + \phi$ and $C$ is a large enough constant. 
\end{lemma}

To prove \Cref{lem:submodularFTRL_highprob_linfty}, we will use the following martingale inequality due to Freedman \cite{freedman1975tail}.

\begin{theorem}[Freedman's Inequality, Theorem 1.6 in \cite{freedman1975tail}]
\label{thm:FreedmanInequality}
Consider a real-valued martingale sequence $\{Y_t\}_{t\geq 0}$ such that $X_0=0$, and $\E[Y_{t+1}|\mathcal{F}_t]=0$ for all $t$, where $\{\mathcal{F}_t\}_{t\geq 0}$ is the filtration defined by the martingale. Assume that the sequence is uniformly bounded, i.e., $|Y_t|\leq R$ almost surely for all $t$. Now define the predictable quadratic variation process of the martingale to be $W_t=\sum_{j=1}^t \E[Y_j^2|\mathcal{F}_{j-1}]$ for all $t\geq 1$. Then for all $\lambda \geq 0$ and $\sigma^2>0$, we have 
\[
\mathbb{P} \Big[ \exists \tau \geq 0 \text{ s.t. } \sum_{j=0}^\tau Y_j \geq \lambda  \text{ and } W_\tau \leq \sigma^2  \Big] \leq \exp\Big(- \frac{\lambda^2/2}{\sigma^2+ R \lambda /3} \Big).
\]
\end{theorem}

\begin{proof}[Proof of \Cref{lem:submodularFTRL_highprob_linfty}]
For notational convenience let $g_t \defeq g_{x_t}$.
Note that for every coordinate $i \in V$, we have \[
y_i \leq \frac{1}{M} \sum_{t=0}^{M-1} (g_t)_i \leq (u_f)_i \leq \|u_f\|_\infty .
\]
So we only need to prove that $\sum_{t=0}^{M-1} -(g_t)_i = \widetilde{O}(M \|u_f\|_\infty)$ for each $i \in V$, i.e., the coordinates of $y$ don't become too negative. We start by bounding the empirically sampled vectors $h_t$. 

\smallskip
\noindent \textbf{Bounding the empirical process $h_t$.} Note that \Cref{submod_st_FTRL} is an instantiation of \Cref{st_FTRL} with Lov\'asz extension $f$, domain $D = S_k^V$, step size $\eta$, number of iterations $M$, the entropy regularizer $r(x)$, and the stochastic subgradient $h_t$ generated as in \Cref{line:vsample_begin} - \ref{line:vsample_end}. Therefore, by the second statement in \Cref{lem:stMD_guarantee}, with probability at least $1 - \rho$, we have
\begin{align} \label{eq:SubmodularFTRL_empirical}
\max_{i \in V} \sum_{t=0}^{M-1} - (h_t)_i \leq \frac{M v^*}{\eta} + \frac{1}{\eta} \log \Big( \frac{n^2}{\rho} \Big) .
\end{align}
Here, $v^*$ is a number such that for every $t \in \{0, \ldots, M-1\}$, 
\[
-\eta \langle p^{(t)}, g_t \rangle + \eta^2 \E_t\left[\norm{h_t}_{p^{(t)}}^2\right] \leq v^* ,
\]
where $p^{(t)} \in \R^V$ is defined as $p^{(t)}_i \defeq \frac{w^{(t)}_i}{\norm{w^{(t)}}_1}$ and $w^{(t)} \in \R^V$ is given by $w_i^{(t)} \defeq \exp(-\eta \sum_{t'=0}^{t-1} (h_{t'})_i)$. 
These definitions allow us to view $x_t = \argmin_{x \in S_k^V} \eta \sum_{t'=0}^{t-1} h_{t'}^\top x +  r(x)$ in \Cref{line:prox_step_submodular_FTRL} of \Cref{submod_st_FTRL} as a proximal step with the vector $\eta \sum_{t'=0}^{t-1} h_{t'} = - \log w_i^{(t)}$ in  \Cref{lem:FTRL_update}. 
The second statement in \Cref{lem:FTRL_update} then implies that the decreasing order of the coordinates in $x_t$ and $p^{(t)}$ are the same, i.e., $\pi_{x_t} = \pi_{p^{(t)}}$ and thus $g_{x_t} = g_{p^{(t)}}$. It follows that we can bound
\[
- \langle p^{(t)}, g_t \rangle = - \langle p^{(t)}, g_{p^{(t)}} \rangle = - f(p^{(t)}) \leq \phi ,
\]
and by \Cref{lem:eff_vsampling} that 
\[
\E_t\left[\norm{h_t}_{p^{(t)}}^2\right] \leq (2k \|u_f\|_{\infty} - f^*) (2 \|u_f\|_1 - f(V)) \leq U_\infty \cdot U_1 ,
\]
where we denote $U_1 \defeq 2 \|u_f\|_1 + \phi$ as before. 
Therefore, we can set $v^* = \eta \phi + \eta^2 U_\infty U_1$ and the bound in \eqref{eq:SubmodularFTRL_empirical} becomes
\begin{align} \label{eq:SubmodularFTRL_emp_2}
\max_{i \in V} \sum_{t=0}^{M-1} - (h_t)_i \leq M \phi + M \eta U_\infty U_1 + \frac{1}{\eta} \log \Big( \frac{n^2}{\rho} \Big) . 
\end{align}

\smallskip
\noindent \textbf{Bounding the difference process $(-g_t)_i-(-h_t)_i $.} We have so far upper bounded how negative the empirical process $(h_t)_i$ can become. We next seek to bound the difference $\sum_{t=0}^{M-1} X_t$, where we denote $X_t \defeq (-g_t)_i - (-h_t)_i$ for a fixed coordinate $i$. Note that $\E_t[(h_t)_i] = (g_t)_i$, so the stochastic process $\sum_{t=0}^{T-1} X_t$ is a martingale. 
For notational simplicity, we shall also denote $G_T \defeq \sum_{t=0}^{T-1} -(g_t)_i$ and $H_T \defeq \sum_{t=0}^{T-1} -(h_t)_i$ so that $\sum_{t=0}^{T-1} X_t = G_T - H_T$. For this step, we will make use of \Cref{thm:FreedmanInequality}.

To apply Freedman's inequality, we first note that $(h_t)_i$ either takes value $\frac{(g_t)_i}{p^{(t)}_i}$ (if $i$ gets sampled) or $0$. Therefore, with probability 1, 
\[
|X_t| \leq \max \Big\{ |(g_t)_i|, \Big|\frac{(g_t)_i}{p^{(t)}_i} - (g_t)_i \Big| \Big\} \leq \Big|(g_t)_i \cdot \frac{\|v^{(t)}\|_1}{v^{(t)}_i} \Big| \leq U_1 , 
\]
where the last inequality follows from $v^{(t)}_i = 2 (u_f)_i - (g_t)_i \geq |(g_t)_i|$ and $\|v^{(t)}\|_1 = U_1$. 
We can also bound the quadratic variation as
\begin{align}
\E_t[X_t^2] &= p^{(t)}_i \cdot \Big( - (g_t)_i + \frac{(g_t)_i}{p^{(t)}_i} \Big)^2 + (1 - p^{(t)}_i) \cdot ((g_t)_i)^2 \nonumber \\
& \leq ((g_t)_i)^2 \cdot \frac{1 - p^{(t)}_i}{p^{(t)}_i} = ((g_t)_i)^2 \cdot \frac{\|v^{(t)}\|_1 - v^{(t)}_i}{v^{(t)}_i} \leq |(g_t)_i| \cdot U_1  \label{eq:quad_var} . 
\end{align}
Then by \Cref{thm:FreedmanInequality}, we have
for all $\lambda, \sigma^2 > 0$,
\begin{align} \label{eq:submodularFTRL_freedman}
\p \Big[ \sum_{t=0}^{M-1} X_t \geq \lambda \text{ and } \sum_{t=0}^{M-1} \E_t[X_t^2] \leq \sigma^2 \Big] \leq \exp \Big( - \frac{\lambda^2/2}{\sigma^2 + U_1 \lambda/3} \Big) . 
\end{align}
But as the quadratic variation bound in \eqref{eq:quad_var} depends on $(g_t)_i$, the quantity we want to control, we shall use \eqref{eq:submodularFTRL_freedman} with different values of $\lambda$ and $\sigma^2$. Note that by \eqref{eq:quad_var}, with probability 1, 
\[
\sum_{t=0}^{M-1} \E_t[X_t^2] \leq U_1 \cdot \sum_{t=0}^{M-1} |(g_t)_i| \leq U_1 \cdot \sum_{t=0}^{M-1} \|v^{(t)}\|_1 \leq M U_1^2 . 
\]
On the other hand, if $\sum_{t=0}^{M-1} |(g_t)_i| \leq \|u_f\|_\infty$, then we would be done in controlling $(g_t)_i$ already. 
Therefore, we define values $\sigma_i^2 \defeq \|u_f\|_\infty U_1 \cdot 2^i$, for $i = 0, \ldots, \ell$, where $\ell \defeq \lceil \log (\frac{M U_1}{\|u_f\|_\infty}) \rceil = O(\log n)$, and the value $\lambda_i \defeq \sqrt{2 \sigma_i^2 + U_1^2} \cdot \log(\frac{n}{\rho})$. 
For $i \in \{0,\ldots, \ell\}$, define events
\[
\mathcal{E}_i \defeq \Big\{ \sum_{t=0}^{M-1} X_t < \lambda_i \text{ or } \sum_{t=0}^{M-1} \E_t[X_t^2] > \sigma^2 \Big\} \quad \text{and} \quad \mathcal{E}_* \defeq \{\eqref{eq:SubmodularFTRL_emp_2} \text{ holds}\} .
\]
Then \eqref{eq:SubmodularFTRL_emp_2} and \eqref{eq:submodularFTRL_freedman} together imply that all the events $\mathcal{E}_*$ and $\{\mathcal{E}_i\}_{i=0}^\ell$ hold with probability at least $1 - 2 \rho$.
We now argue that conditioning on these events, $G_M = \sum_{t=0}^{M-1} -(g_t)_i$ has to be small.

\smallskip
\noindent \textbf{Controlling $-(g_t)_i$ conditioned on events $\mathcal{E}_*$ and $\{\mathcal{E}_i\}_{i=0}^\ell$.} To this end, we let $i^* \in \{0, \ldots, \ell\}$ be the smallest index (such an index must exist by our choice of $\ell$) such that 
\[
\sum_{t=0}^{M-1} \E_t[X_t^2] \leq U_1 \cdot \sum_{t=0}^{M-1} |(g_t)_i| \leq \sigma_{i^*}^2 . 
\]
Then as the event $\mathcal{E}_{i^*}$ hold, we must have
\begin{align*}
G_M - H_M &\leq \lambda_{i^*} = \sqrt{2 \sigma_{i^*}^2 + U_1^2} \cdot \log(\frac{n}{\rho}) 
\leq \sqrt{4 U_1 \sum_{t=0}^{M-1} |(g_t)_i| + U_1^2} \cdot \log(\frac{n}{\rho}) \\
& \leq \sqrt{4 U_1 (2 M \|u_f\|_\infty + G_M)) + U_1^2} \cdot \log(\frac{n}{\rho}),
\end{align*}
where the second inequality follows from the minimality of $i^*$ and the last inequality uses
\[
\sum_{t=0}^{M-1} |(g_t)_i| \leq \sum_{t=0}^{M-1} (2 (u_f)_i - (g_t)_i) \leq 2 M \|u_f\|_\infty + G_M . 
\]
Unraveling the above inequality and using the event $\mathcal{E}_*$ give the bound
\begin{align*}
G_M &\leq H_M + \sqrt{8M U_1 \|u_f\|_\infty + U_1^2} \cdot \log(\frac{n}{\rho}) + 4 U_1 \log^2(\frac{n}{\rho}) \\
&\leq M \phi + M \eta U_\infty U_1 + \frac{1}{\eta} \log \Big( \frac{n^2}{\rho} \Big) + \sqrt{8M U_1 \|u_f\|_\infty + U_1^2} \cdot \log(\frac{n}{\rho}) + 4 U_1 \log^2(\frac{n}{\rho}) .
\end{align*}
Finally, plugging in $M = \frac{U_1 U_\infty \log n}{\delta^2}$ and $\eta = \frac{\delta}{ U_1 U_\infty}$, and setting $\rho \defeq n^{-C-1}/2$, we have
\begin{align*}
\frac{1}{M} \sum_{t=0}^{M-1} - (g_t)_i 
& \leq \phi + \delta + O(\delta) + O(k \log n) \cdot \sqrt{\frac{ U_1 \|u_f\|_\infty}{k U_1 U_\infty \log n}} \cdot \delta  + O(k \log n)  \cdot \frac{\delta^2 }{ U_\infty} \\
& \leq \phi + O(\sqrt{\log n}) \cdot \delta + O(\log n) \cdot \frac{k \delta^2}{U_\infty} . 
\end{align*}
where the last inequality uses the assumption that $\delta \leq \phi \leq U_\infty$. 
Finally, apply the union bound over all $i\in V$ which completes the proof of the lemma. 
\end{proof}

Finally, we have gathered enough tools to present the proof of \Cref{lem:st_certif_runtime_FTRL}, which we first restate below for convenience. 

\StochDualCertificate*

\newcommand{\ysafe}{\overline{y}}

\begin{proof}[Proof of \Cref{lem:st_certif_runtime_FTRL}]
For each iteration $\ell \in [N]$ in \Cref{line:StochDualCert_Rep} of \Cref{st_dual_cert}, we let $y_\ell \defeq \frac{1}{M} \sum_{\pi \in \mathcal{C}_\ell} g_\pi$ denote the average subgradients in the $\ell$-th call to \submodstFTRL$(\cdot)$. 
\Cref{lem:submodularFTRL_highprob_linfty} implies that with probability at least $1 - n^{-C}$ where $C > 0$ is a large constant, for all $\ell \in [N]$, 
\[
\|y_\ell\|_\infty \leq \|u_f\|_\infty + \widetilde{O}(1) \cdot \Big(\phi + \frac{k \delta^2}{U_\infty} \Big) \defeq R. 
\]
For all $\ell \in [N]$ define
\[
\ysafe_\ell
\defeq 
\begin{cases}
y_\ell & \text{if } \norm{y_\ell}_\infty \leq R\\
u_f & \text{if } \norm{y_\ell}_\infty > R
\end{cases}
\,.
\]
Since $g_\pi \leq u_f$ for all permutations $\pi$ we see that $y_\ell \leq \ysafe_\ell$. Additionally, by the early fact regarding \Cref{lem:submodularFTRL_highprob_linfty} we have that, for every $\ell \in [N]$, $y_\ell = \ysafe_\ell$ with probability $1 - n^{-C}$, which translates to $y_\ell = \ysafe_\ell, \forall \ell \in [N]$ with high probability, by applying a union bound over $\ell \in [N]$ and using that $N \le n^{10}$. 

Next, for any fixed $w \in S_k^V$, define $X^w_\ell \defeq \langle \ysafe_\ell, w\rangle$. Note that each iteration of \Cref{line:StochDualCert_Rep} is an independent execution of \Cref{submod_st_FTRL} with accuracy $\delta/2$. Consequently, \Cref{lem:submodularFTRL_expectation} and the $\ysafe_\ell \geq y_\ell$ implies that 
\[
\E[X^w_\ell] \geq \E[\langle \ysafe_\ell, w\rangle] \geq f^* - \frac{\delta}{2}\,.
\]
Additionally, from the definition of $\ysafe$ and $R$ we see that $\norm{\ysafe}_\infty \leq \norm{u_f}_\infty$ and therefore 
\[
|X^w_\ell| \leq k \|y_\ell\|_\infty \leq k R\,.
\]
Consequently, applying Azuma-Hoeffding's inequality yields that 
\[
\p \Big(\frac{1}{N}\sum_{\ell \in [N]} (X^w_\ell - \E[X^w_\ell]) \leq -\frac{\delta}{2} \Big) \leq \exp\Big(- \frac{N \delta^2}{8k^2 R^2} \Big) .
\]
The above probability is smaller than $n^{-10k}$ by our choice of
\begin{align*}
N = 100 \delta^{-2} R^2 k \log n = \widetilde{O}(1) \cdot \delta^{-2} k \Big(k \|u_f\|_\infty + \widetilde{O}(k) \cdot \Big(\phi + \frac{k \delta^2}{U_\infty} \Big)\Big)^2  \leq \widetilde{O}(\delta^{-2} k^5 \phi) ,
\end{align*}
where we used the assumptions that $\|u_f\|_\infty \leq O(k \phi)$ and $0 < \delta \leq \phi$. 
This implies that with probability at least $1 - n^{-9k}$, we have
\[
\frac{1}{N}\sum_{\ell \in [N]} \langle \ysafe_\ell, w\rangle \geq f^* - \delta .
\]
Finally, taking the union bound over all the $O(n^k)$ vertices of $S_k^V$ for the choice of $w$, we obtain that \[\max_{w \in S_k^V} \frac{1}{N}\sum_{\ell \in [N]} \langle \ysafe_\ell, w\rangle \geq f^* - \delta,\]
with probability $1 - n^{-5k}$. Since $y_\ell = \ysafe_\ell$ with probability $1 - n^{-C}$, we obtain that $\frac{1}{N} \sum_{\ell \in [N]} y_\ell$ is a $(\delta, k)$ dual certificate with high probability. 

The total number of permutations used is 
\begin{align*}
m & = MN = \widetilde{O}(\delta^{-2} k (k\|u_f\|_{\infty} + \phi) (\|u_f\|_1 + \phi)) \cdot \widetilde{O}(\delta^{-2} k^5 \phi^2)\\
& = \widetilde{O}(k^6 \delta^{-4} \phi^2 (\|u_f\|_{\infty} + \phi) (\|u_f\|_1 + \phi)) . 
\end{align*}
Since for each permutation, we only sample a single coordinate in \Cref{line:vsample_begin} - \ref{line:vsample_end} in \Cref{submod_st_FTRL}, the total runtime is then $\widetilde{O}(m) \cdot (\EO + \poly(n))$ by the second statement of \Cref{lem:eff_vsampling}. This completes the proof of the theorem.
\end{proof}

%% file: sequential_algo_other.tex
\subsection{Dimensionality Reduction for Sequential Algorithm}
\label{subsec:dimensionality_reduction_sequential}

In this subsection, using the algorithm for computing $(\delta,k)$ dual certificates in \Cref{subs_stMD}, we give the implementation of the subprocedure \textsf{Dimenionality-Reduction} for our sequential algorithm in \Cref{alg:dim_red_sequential}.
Recall from \Cref{sec:framework:meta-algorithm} that assuming $f$ has a $\sparsity$-sparse minimizer, the subprocedure \textsf{Dimensionality-Reduction}$(f,k)$ either outputs a dimensionality reduction $T \neq \emptyset$ that belongs to every minimizer of $f$, or certifies that $f^* > -\frac{\|u_f\|_\infty}{12k}$.

Similarly to \Cref{alg:dim_red_parallel}, \Cref{alg:dim_red_sequential} aims at finding a value $\phi>0$ such that $-\phi \leq f^* \leq -\phi/2$, starting from the initial guess $\phi^{(0)} = \|u_f\|_1 - f(V)$ and halving the value of $\phi$ in each iteration of the while loop until the threshold $\frac{\|u_f\|_\infty}{12k}$ is reached. 
If $f^* \leq -\frac{\|u_f\|_\infty}{12k}$, then along this halving process there must be an iteration $i$ such that
$-\phi^{(i)} \leq f^* \leq -\phi^{(i)}/2$ for which we can obtain a non-empty dimension reduction $T \neq \emptyset$. 

\begin{algorithm}[H]
\caption{Dimensionality Reduction for Sequential Algorithm}\label{alg:dim_red_sequential}
    \KwData{A sparsity parameter $k$, and a submodular function $f$ with a $\sparsity$-sparse minimizer}
    \KwResult{A subset $T \subseteq V$ that must be in every minimizer of $f$, with $T = \emptyset$ certifying that $f^* > - \frac{\|u_f\|_\infty}{12k}$} 
    \SetKwFunction{stdimred}{Dimensionlity-Reduction}
    \SetKwFunction{stdualcert}{StochDualCertificate}
    \SetKwProg{Fn}{Function}{:}{}
    \Fn{\stdimred{$f$, $k$}}{
         $\phi \gets \|u_f\|_1 - f(V)$ \label{line:init_guess_seq} \tcp*{Lower bound $f^* \geq -(\|u_f\|_1 - f(V))$}
        \While{$\phi \geq \frac{\|u_f\|_\infty}{12k}$}{
            $\delta \leftarrow \frac{\phi}{8k}$, $N \leftarrow \frac{100 k^4 \|u_f\|_1}{\|u_f\|_\infty} \cdot \log n$\;
            $\{\pi^{(t)}\}_{t \in [m]} \gets \stdualcert{$f, k, \phi, \delta$}$ \tcp*{$\frac{1}{m} \sum_t g_{\pi^{(t)}}$ is $(\delta, k)$ dual cert.}
            \For{$t =1, \ldots, N$}
            {
            $z^{(t)} \leftarrow \mathsf{vSampling}(g_{\pi^{(i_t)}})$ for random $i_t \sim [m]$ \label{line:vsample} \tcp*{$\mathsf{vSampling}$ for random $g_{\pi^{(i_t)}}$}
            }
            $z \leftarrow \frac{1}{N} \sum_{t=1}^N z^{(t)}$ \label{line:estimate_y} \tcp*{Estimate for $(\delta,k)$ dual certificate} 
            $T \leftarrow \{p\in V: z_p \le - \frac{3\phi}{8k}\}$ \label{line:find_dim_red_sequential} \tcp*{Find dimensionality reduction}
            \lIf{$T \ne \emptyset$}{\Return $T$}
            $\phi \gets \phi / 2$\;
        }
        \Return $\emptyset$\; 
        
    }
\end{algorithm}

The following lemma for \Cref{alg:dim_red_sequential} is the main result of this subsection.

\begin{restatable}[$\textsf{Dimensionality-Reduction}$ for Sequential Algorithm]{lemma}{DimRedSequential}
\label{lem:dim_red_sequential}
Let $k \in \mathbb{Z}_{>0}$ and $f : 2^V \rightarrow \R$ be a submodular function with a $k$-sparse minimizer and $u_f \geq 0$. Then \Cref{alg:dim_red_sequential} outputs a set $T \subseteq V$ that must be in every minimizer of $f$ such that $T = \emptyset$ implies $f^* > -\|u_f\|_\infty/12k$. 
Moreover, the algorithm runs in time $\otilde\Big(\frac{k^{12} \|u_f\|_1}{\|u_f\|_\infty} \cdot \EO + \poly(n)\Big)$.
\end{restatable}

In the remainder of this subsection we prove \Cref{lem:dim_red_sequential}. 
To establish the correctness of \Cref{alg:dim_red_sequential}, we start with the following claim that the $\mathsf{vSampling}$ for a random permutation $g_{\pi^{(i_t)}}$ in \Cref{line:vsample} has small variance on each coordinate. This claim differs from \Cref{lem:eff_vsampling}, as it is a bound on the variance of one coordinate of $z^{(t)}$, rather than a bound on the local norm of $z^{(t)}$ at $x_{i_t}$.

\begin{claim}[Coordinate Variance Bound for $\mathsf{vSampling}$]
\label{lem:variance_bound}
Let $f: 2^V \rightarrow \mathbb{R}$ be a submodular function such that $f(\emptyset) = 0$ and $u_f \geq 0$. Let $\{\pi^{(t)}\}_{t \in [m]}$ be a set of permutations on $V$ and define the vector $y\defeq\frac{1}{m}\sum_{t\in[m]}g_{\pi^{(t)}}$. 
Define the random vector 
\[
w \defeq \frac{(2 \norm{u_f}_1 - f(V)) \cdot (g_{\pi^{(i)}})_p}{2 (u_f)_p - (g_{\pi^{(i)}})_p} \cdot \indicVec{p} ,
\]
where $i \in [m]$ is sampled uniformly at random and $p \in V$ is sampled with probability proportional to $2 (u_f)_p - (g_{\pi^{(i)}})_p$. 
Then we have 
\begin{enumerate}
\item $\E[w] = y$, 
\item $\|w\|_\infty \leq 2 \norm{u_f}_1 - f(V) $ with probability 1, and
\item $\E[w_q^2]
\leq (2 \norm{u_f}_1 - f(V)) (2 (u_f)_q - y_q)$ for any coordinate $q \in V$. 
\end{enumerate}
\end{claim}

\begin{proof}
Note that $\sum_{p \in V} (2 (u_f)_p - (g_{\pi^{(i)}})_p) = 2 \norm{u_f}_1 - f(V)$ and consequently $\E[w] = y$. Further, 
\[
\norm{w}_\infty 
\leq (2 \norm{u_f}_1 - f(V)) \cdot \max_{i \in [m], p \in V} \frac{(g_{\pi^{(i)}})_p}{2 (u_f)_p - (g_{\pi^{(i)}})_p} \leq 2 \norm{u_f}_1 - f(V) 
\]
with probability $1$, where the last inequality follows because $(g_{\pi^{(i)}})_p \leq (u_f)_p$ by submodularity and $(u_f)_p \geq 0$ by the assumption of the claim. 
Next we note that 
\begin{align*}
\E_{i,p}[|w_q |]
&=  \E_{i} \Big[
\Pr[p = q] \cdot 
\Big| \frac{(2 \norm{u_f}_1- f(V)) \cdot (g_{\pi^{(i)}})_q}{2 u_q - (g_{\pi^{(i)}})_q} \Big| \Big] \\
&=
\E_{i} [ | (g_{\pi^{(i)}})_q |]
\leq  (u_f)_q + \E_{i} [ (u_f)_q - (g_{\pi^{(i)}})_q]
= 2 (u_f)_q  - y_q .
\end{align*}
Consequently, we have
\[
\E[w_q^2] \leq \|w\|_\infty \cdot \E_{i,p}[|w_q|]
\leq  (2 \norm{u_f}_1 - f(V)) (2 (u_f)_q - y_q) .
\]
This completes the proof of the claim. 
\end{proof}

Using the above variance bound, we next prove that whp. the vector $z$ in \Cref{line:estimate_y} of \Cref{alg:dim_red_sequential} is an estimate of the $(\delta, k)$ dual certificate $y \defeq \frac{1}{m} \sum_{t \in [m]} g_{\pi^{(t)}}$ with additive accuracy $O(\phi/k)$.

\begin{claim}[Estimate for Dual Certificate]
\label{claim:simple_sample}
Under the same assumptions as in \Cref{lem:variance_bound}, the random vector $z$ in \Cref{line:estimate_y} of \Cref{alg:dim_red_sequential} satisfies that $\E[z] = y$, and $\|z - y\|_\infty \leq \alpha$ with high probability, where the accuracy 
\[
\alpha \defeq 10 \cdot \max\left\{ \frac{(2\|u_f\|_1 - f(V)) \cdot \log n}{N}, \sqrt{\frac{(2\|u_f\|_1 - f(V)) \cdot (\|u_f\|_\infty - f^*) \cdot \log n}{N}}\right\} .
\]
In particular, if parameter $\phi \geq \|u_f\|_\infty / 12k$ satisfies $ -\phi \leq f^*$, and $N = \Theta(\frac{k^4 \|u_f\|_1}{\|u_f\|_\infty} \cdot \log n)$ as in \Cref{alg:dim_red_sequential}, then the accuracy $\alpha \leq \phi /  8k$. 
\end{claim}

\begin{proof}
Note that for every $t \in [N]$, each random vector $z^{(t)}$ is obtained by applying $\mathsf{vSampling}$ to a uniformly random vector in $\{g_{\pi^{(t)}}\}_{t \in [m]}$ exactly as in \Cref{lem:variance_bound}. 
Therefore, it follows from \Cref{lem:variance_bound} that $\E[z] = y$ and  $\|z^{(t)}\|_\infty \leq M \defeq 2 \|u_f\|_1 - f(V)$ with probability 1, and that for every coordinate $q\in V$, the variance of each sampled vector $z^{(t)}$ satisfies 
\[
\E[(z^{(t)}_q - y_q)^2] \leq \E[(z^{(t)}_q)^2] \leq (2 \|u_f\|_1 - f(V)) (2 (u_f)_q - y_q)  \leq 2 (2 \|u_f\|_1 - f(V)) (\|u_f\|_\infty - f^*) .
\]
Next, to prove that $\|z - y\|_\infty \leq \alpha$ whp., we fix an arbitray coordinate $q \in V$. Since $z^{(t)}_q$ are i.i.d. samples with expectation $\E[z^{(t)}_q] = y_q$, it follows from Bernstein's inequality that 

\begin{align} \label{eq:estimate_y}
\Pr\left[\left|\sum_{t \in [N]} (z^{(t)}_q - y_q)\right| \ge \lambda\right] \leq 2 \cdot \exp\left(\frac{\lambda^2 / 2}{\sum_{t \in [N]} \E((z^{(t)}_q - y_q)^2) + \lambda M / 3}\right) .
\end{align}
where recall that $M \defeq 2 \|u_f\|_1 - f(V)$ is an upper bound on $\|w^{(t)}_q\|$ almost surely. 
Now we plug $\lambda = \alpha N$ into the bound \Cref{eq:estimate_y}.
Note that the definition of $\alpha$ in the claim satisfies that 
\begin{align*}
\sum_{t \in [N]} \E((z^{(t)}_q - y_q)^2) 
\leq 2N (2 \|u_f\|_1 - f(V)) (\|u_f\|_\infty - f^*) 
\leq 2N \cdot \frac{\alpha^2 N}{100 \log n} \leq \lambda^2 /2 \cdot \frac{1}{25 \log n} ,
\end{align*}
and that 
\[
\lambda M / 3 \leq \lambda (2 \|u_f\|_1 - f(V)) / 3 
\leq \frac{\alpha N}{30 \log n} \leq  \lambda/2 \cdot \frac{1}{15 \log n} .
\]
It thus follows from \Cref{eq:estimate_y} that 
\[
\Pr\left[\left|\sum_{t \in [N]} (z^{(t)}_q - y_q)\right| \ge \alpha N \right] \leq 2 n^{-9} .  
\]
The first statement of the claim then follows by applying a union bound to all coordinates $q \in V$. 

To prove the second statement of the claim, we can trivially upper bound 
\[
2 \|u_f\|_1 - f(V) \leq 2 \|u_f\|_1 + \phi \quad \text{and} \quad (2\|u_f\|_1 - f(V)) \cdot (\|u_f\|_\infty - f^*) \leq (2\|u_f\|_1 + \phi) \cdot (\|u_f\|_\infty + \phi) .
\]
Then the bound $\alpha \leq \phi/8 k$ follows by plugging in the value of $N$ and using $\phi \geq \|u_f\|_\infty / 12k$. 
\end{proof}

The following claim establishes the correctness of \Cref{alg:dim_red_sequential}. 

\begin{claim}
\label{claim:st_dim_red}
Let $f$ be a submodular function with a $k$-sparse minimizer, and $\phi \geq \|u_f\|_\infty / 12k$ satisfying $ -\phi \leq f^*$ in one iteration of the while loop. Then the subset $T \subseteq V$ in \Cref{line:find_dim_red_sequential} of  \Cref{alg:dim_red_sequential} is in every minimizer of $f$ whp. Moreover, if $-\phi \leq f^* \leq -\phi/2$, then $T \neq \emptyset$ whp. 
\end{claim}

\begin{proof}
Note that by \Cref{lem:st_certif_runtime_FTRL}, the set of permutations $\{\pi^{(t)}\}_{t \in [m]}$ returned by $\mathsf{StochDualCertificate}$ is such that $y = \frac{1}{m} \sum_{t\in[m]}g_{\pi^{(t)}}$is a $(\delta, k)$ dual certificate whp. 
Since $\phi \geq \|u_f\|_\infty / 12k$ and $-\phi \leq f^*$,  \Cref{claim:simple_sample} implies that the vector $z$ in \Cref{line:estimate_y} of \Cref{alg:dim_red_sequential} satisfies $\|z - y\|_\infty \leq \phi / 8k$ whp.  
In the following, we condition on the above two events. 

In the first part of the claim, for every element $p\in T$ defined in \Cref{line:find_dim_red_sequential}, we have 
\[
y_p \leq z_p + \|z - y\|_\infty \leq - \frac{3 \phi}{8k} + \frac{\phi}{8k} < -\delta .  
\]
It then follows from \Cref{claim:certif_dim_red} that $p$ must be in every minimizer of $f$. This proves the first statement of the claim. 

To prove the second statement, note that since $-\phi \leq f^* \leq -\phi/2$ and $y$ is a $(\delta, k)$ dual certificate, the most negative coordinate $p$ of $y$ satisfies $y_p \leq - |f^*|/k$. Therefore, 
\[
z_p \leq y_p + \phi/8k \leq - \phi/2k + \phi/8k \leq - 3\phi/8k , 
\]
which implies that $p$ must be inside the subset $T$ considered \Cref{line:find_dim_red_sequential}. This proves that $T \neq \emptyset$ whp. for such an iteration and completes the proof of the claim.  
\end{proof}

Finally, we are ready to prove \Cref{lem:dim_red_sequential} which we restate below for convenience. 

\DimRedSequential*

\begin{proof}[Proof of \Cref{lem:dim_red_sequential}]
The first statement of the lemma essentially follows from \Cref{claim:st_dim_red}. Since \Cref{alg:dim_red_sequential} starts from $\phi = \phi_0 \defeq \|u_f\|_1 - f(V) \geq |f^*|$,  \Cref{claim:st_dim_red} implies that the subset $T$ in \Cref{line:find_dim_red_sequential} is in every minimizer of $f$ whp. and that $-\phi \leq -f^*$ in every iteration of the while loop executed by \Cref{alg:dim_red_sequential}. 
In particular, if $f^* \leq - \|u_f\|_\infty/12k$, then there must be one iteration of the while loop in which $-\phi \leq f^* \leq -\phi/2$ (assuming $T = \emptyset$ in all previous iterations). For such an iteration, \Cref{claim:st_dim_red} states that the subset $T \neq \emptyset$ whp. and therefore \Cref{alg:dim_red_sequential} will never return $\emptyset$.  
This proves the first statement of the lemma. 

Next we prove the second statement on the runtime of \Cref{alg:dim_red_sequential}.  
We first note that the total number of iterations of the while loop is
\[
\log \Big(\frac{\|u_f\|_1 - f(V)}{\max\{\|u_f\|_\infty / 12k, |f^*|\}} \Big) \leq \log\Big(\frac{\|u_f\|_1 + |f^*|}{\max\{\|u_f\|_\infty / 12k, |f^*|\}} \Big)  = O(\log n).  
\]
For each iteration of the while loop, as we have $\norm{u_f}_{\infty} = \otilde(k) (\phi + \delta)$ each time we call $\mathsf{StochDualCertificate}$, by \Cref{lem:st_certif_runtime_FTRL}, the number of queries to $\EO$ made by calling the $\mathsf{StochDualCertificate}$ routine is given by 
\begin{align*}
&\otilde(k^6 \delta^{-4} \phi^2 (\|u_f\|_{\infty} + \phi) (\|u_f\|_1 + \phi)) \leq \widetilde{O}(k^{10}) \cdot \frac{(\|u_f\|_\infty + \phi) (2 \|u_f\|_1 + \phi)}{\phi^2}  
\leq \widetilde{O}(k^{12} \frac{\|u_f\|_1}{\|u_f\|_\infty}) ,
\end{align*}
where the first inequality is obtained by plugging in $\delta = \phi/8k$ and $\phi \geq |f^*|$, and the second inequality uses $\phi \geq \|u_f\|_\infty / 12k$. 
Next we observe that given the vector $u_f$, each $\mathsf{vSampling}$ in \Cref{line:vsample} can be implemented using only $O(\log n)$ additional queries. 
Thus the total number of queries due to estimating $y$ in \Cref{line:estimate_y} is at most $\widetilde{O}(N) \ll \widetilde{O}(k^{12} \|u_f\|_1 / \|u_f\|_\infty)$. 
Finally, note that every step of \Cref{alg:dim_red_sequential} uses polynomial additional runtime. 
\end{proof}

\subsection{Arc Finding for Sequential Algorithm}
\label{subsec:arc_finding_sequential}

In this section, we present our sequential \textsf{Arc-Finding} subprocedure. As mentioned in \Cref{sec:approach:sequential}, in prior work \cite{DVZ21,LSW15}, finding arcs relies on the ``move-to-front'' approach.
Given a vector $y$ as the average of subgradients $\{g_{\pi^{(t)}}\}_{t \in [m]}$ and an element $p \in V$, this approach considers the vector $y'$ obtained by moving $p^{\downarrow}$ to the front of every permutation $\pi^{(t)}$. A coordinate $q \notin p^{\downarrow}$ with negative enough entry $y'(q)$ is then concluded to be in every sparse minimizer containing $p$, i.e. an arc from $p$ to $q$.

Key to our algorithm is 
the following lemma which shows when the move-to-front procedure allows us to deduce arcs or dimensionality reductions. 
Recall from in \Cref{sec:prelim} that, for some permutation $\pi$, $\norm{\Delta_{\pi,. p}}_{1}$ measures the total decrease of coordinates of $g_{\pi}$ in $V \setminus p^{\downarrow}$ after moving $p^{\downarrow}$ to the front of $\pi$. 
Similarly, $\norm{\Delta_p}_{1}$ measures the total decrease of coordinates of $y$ in $V \setminus p^{\downarrow}$ after moving $p^{\downarrow}$ to the front in every permutation. 
Informally, \Cref{lem:Delta_thresholds} states that any element $p$  in some $\sparsity$-sparse minimizer $S^*$ with large $(u_{f^{\sharp \rfam}})_p$ must also have large $\norm{\Delta_p}_{1}$. Moreover, we can find $q \in S^* \setminus p^{\downarrow}$ such that $y_q$ decreases significantly after moving $p^\downarrow$ to the front, hence deducing an arc $p \to q$.

\begin{lemma}[Move to Front]
\label{lem:Delta_thresholds}
Let $f$ be a submodular function with a $k$-sparse minimizer and $f^{\sharp \rfam}$ be its extension w.r.t. a $k$-consistent ring family $\rfam(E, W, D)$. 
Let $S^*$ be a $k$-sparse minimizer of $f$ (and therefore $f^{\sharp \rfam}$), element $p \in S^*$, and 
$y \in B(f^{\sharp \rfam})$ be a $(\delta, k)$ dual certificate of $f^{\sharp \rfam}$. Define the vector $\Delta_p \in \mathbb{R}^V_{\geq 0}$  as
\[
(\Delta_p)_i \defeq 
\begin{cases}
0 \qquad &\text{if $i \in p^\downarrow$} \\
y_i - (y_{\leftarrow p^\downarrow})_i \quad &\text{otherwise} .
\end{cases}
\]
If $f^{\sharp \rfam}(p^\downarrow) > (2k+1) \delta$, then the following hold:
\begin{itemize}
\item (1) $\norm{\Delta_p}_1 = f^{\sharp \rfam}(p^\downarrow) - y(p^{\downarrow}) > 2 k \delta$ and $y(p^{\downarrow}) \le \delta$,  
\item (2) There exists $q \in S^* \setminus p^\downarrow$ so that $(\Delta_p)_q > \frac{1}{k} \norm{\Delta_p}_1$, and
\item (3) If an element $q \in V\setminus p^\downarrow$ satisfies $(\Delta_p)_q \geq \frac{1}{2k}\norm{\Delta_p}_{1}$, then $q \in S^*$.
\end{itemize}
In particular, the above hold if $(f^{\sharp \rfam})^* = f^*\geq - \frac{(u_{f^{\sharp \rfam}})_p}{2}$ and $\delta \leq \frac{(u_{f^{\sharp \rfam}})_p}{6k}$. 
\end{lemma}

\begin{proof}
Statement (1) is an immediate corollary of \Cref{lem:move_to_front_dual_cert} together with $p^\downarrow \in \mathcal{F}$.      

To prove the second statement, we observe that 
\[
y(S^* \setminus p^\downarrow) = y(S^*) - y(p^\downarrow) \geq y_-^{k+1}(V) - y(p^\downarrow) \geq f^* - \delta - y(p^\downarrow) ,
\]
where the first inequality is because $|S^*| \leq k$ and the last inequality uses that $y$ is a $(\delta, k)$ dual certificate for $f^{\sharp \rfam}$ and that $(f^{\sharp \rfam})^* = f^*$. By \Cref{lem:move_to_front_dual_cert}, we have $y_{\leftarrow p^\downarrow} \in B(f^{\sharp \rfam}_{p^\downarrow})$. It follows that 
\[
y_{\leftarrow p^\downarrow}(S^* \setminus p^\downarrow) \leq f^{\sharp \rfam}_{p^\downarrow}(S^* \setminus p^\downarrow) = (f^{\sharp \rfam})^* - f^{\sharp \rfam}(p^\downarrow) = f^* - f^{\sharp \rfam}(p^\downarrow) ,
\]
where the first equality uses that $S^* \setminus p^\downarrow$ is a minimizer of $f^{\sharp \rfam}_{p^\downarrow}$, and the second equality uses \Cref{lem:properties_extension}. 
Subtracting the above two inequalities, we have
\begin{align}\label{ineq:negative_mass_increase}
\Delta_p(S^* \setminus p^{\downarrow}) 
= y(S^* \setminus p^{\downarrow}) - y_{\leftarrow p^\downarrow}(S^* \setminus p^{\downarrow}) \ge f^{\sharp \rfam}(p^\downarrow) - \delta - y(p^{\downarrow}) = \|\Delta_p\|_1 - \delta .
\end{align} 
Since $|S^* \setminus p^\downarrow | \leq k-1$ and that $\|\Delta_p\|_1 > 2k \delta$ by (1), there exists $q \in S^* \setminus p^\downarrow$ such that 
\[
(\Delta_p)_q \geq \frac{\| \Delta_p\|_1 - \delta}{k-1} > \frac{\|\Delta_p\|_1}{k } . 
\]
This completes the proof of (2). 

Next, note that by \eqref{ineq:negative_mass_increase} and (1) we have $\Delta_p(V \setminus S^*) \leq \delta < \frac{\|\Delta_p\|_1}{2k}$. Since all entries of $\Delta_p$ are non-negative, any coordinate $(\Delta_p)_q > \frac{\|\Delta_p\|_1}{2k}$ must satisfy $q \in S^* \setminus p^\downarrow$.  This proves (3).

Finally, for the ``in particular'' part of the lemma, we note that
\[
f^{\sharp \rfam}(p^\downarrow) = (u_{f^{\sharp \rfam}})_p + f(W \cup p^\downarrow \setminus \{p\}) \geq (u_{f^{\sharp \rfam}})_p + f^* \geq \frac{(u_{f^{\sharp \rfam}})_p}{2} > (2k+1) \delta ,
\]
where in the last inequality we use $f^*\geq - \frac{(u_{f^{\sharp \rfam}})_p}{2}$ and $\delta \leq \frac{(u_{f^{\sharp \rfam}})_p}{6k}$. 
\end{proof}

In light of  \Cref{lem:Delta_thresholds}, we define $A \defeq \{p \in V: (u_{f^{\sharp \rfam}})_p \geq \|u_{f^{\sharp \rfam}}\|_\infty/2\}$ as the {\em active} set of elements. 
Then  \Cref{lem:Delta_thresholds} allows us to deduce arcs from any $p \in A$ whenever we have $f^* \geq - \frac{\|u_{f^{\sharp \rfam}}\|_\infty}{4}$ and a $(\frac{\|u_{f^{\sharp \rfam}}\|_\infty}{12k}, k)$ dual certificate. 
For the convenience of our presentation, we simplify the notation $f^{\sharp \rfam}$ as $f$ in the remainder of this section. However, it is important to keep in mind that our arc finding procedures are always applied to the extension $f^{\sharp \rfam}$.

Even though \Cref{lem:Delta_thresholds} shows how the move-to-front approach  allows us to deduce arcs or dimensionality reductions from elements in the active set $A$, the cost of computing $\Delta_p = \frac{1}{m} \sum_{t \in [m]} \Delta_{\pi^{(t)}, p^\downarrow}$ for every $p \in A$ is unfortunately too high.
In fact, even computing a single coordinate $(\Delta_{\pi^{(t)}, p^\downarrow})_q$ for all $t \in [m]$ and $p \in A$ naively takes $m \cdot |A|$ queries to the evaluation oracle. As $m \approx \poly(k) \cdot \frac{\|u_f\|_1}{\|u_f\|_\infty}$ in \Cref{lem:st_certif_runtime_FTRL} and $|A|$ can be as large as $n$, naively calculating $(\Delta_{\pi^{(t)}, p^\downarrow})_q$ for all $(t, p) \in [m] \times A$ requires $n^2 \cdot \poly(k)$ evaluation queries which we cannot afford. 

\medskip
\noindent \textbf{Estimating large $(\Delta_p)_q$ via sampling.}
Hence, we resort to estimating the large coordinates of $\Delta_p$ up to $O(\frac{1}{k}\|\Delta_p\|_1)$ error via sampling, which allows us to identify coordinates $q \in V \setminus p^\downarrow$ with $\Delta_p(q) > \frac{1}{k}\|\Delta_p\|_1$ or certify that none exists. 
Before presenting the procedure, let us start by gathering some intuition. 

For every $p \in A$, we consider the matrix $M_p \in \R^{m \times (V \setminus p^{\downarrow})}$ where each row $t \in [m]$ is exactly $\Delta_{\pi^{(t)}, p^{\downarrow}}$. 
The quantity we would like to estimate is the vector $\Delta_p = \frac{1}{m} \sum_{t \in [m]} \Delta_{\pi^{(t)}, p^\downarrow}$, which is the average of the rows of $M_p$.
To do so, the naive approach would be to uniformly sample a row $t \in [m]$, and then sample $q \propto \frac{(\Delta_{\pi^{(t)}, p^{\downarrow}})_q}{\|\Delta_{\pi^{(t)}, p^{\downarrow}}\|_1} =: \beta_q$ and outputs $\|\Delta_{\pi^{(t)}, p^{\downarrow}}\|_1 \cdot \indicVec{q}$. 
This can easily be verified to be an unbiased estimator for $\Delta_p$:
\[
\E_{t \sim [m], q \sim \beta_q}\Big[ \|\Delta_{\pi^{(t)}, p^{\downarrow}}\|_1 \cdot \ind_q \Big]  = \E_{t \sim [m]}[\Delta_{\pi^{(t)}, p^\downarrow}] = \Delta_p .
\]
However, the output estimate has entry $\|\Delta_{\pi^{(t)}, p^{\downarrow}}\|_1$ that can be on the order of $\|u_f\|_1$, while $(\Delta_p)_q \leq \|\Delta_p\|_1$ is typically on the order of $O(k) \cdot \|u_f\|_\infty$. Thus using the above sampling procedure to estimate $(\Delta_p)_q$ up to $O(\frac{1}{k}\|\Delta_p\|_1)$ error would require too many samples. 

A natural fix to the above idea is to sample $t \in [m]$ depending on $\|\Delta_{\pi^{(t)}, p^{\downarrow}}\|_1$. In particular, by sampling $t \propto \|\Delta_{\pi^{(t)}, p^{\downarrow}}\|_1$ and $q \propto \beta_q$,  the random vector $\|\Delta_p\|_1 \cdot \ind_q$ is an unbiased estimator of $\Delta_p$ whose non-zero entry now only has range $\|\Delta_p\|_1$. This allows us to obtain  $O(\frac{1}{k} \|\Delta_p\|_1)$ error for large entries $(\Delta_p)_q$ using only $\poly(k)$ samples. 
However, the obstacle to this approach is that computing $\|\Delta_{\pi^{(t)}, p^{\downarrow}}\|_1$ for all $p \in A$ and $t \in [m]$ would again require $m |A| = O(n^2 \poly(k))$ queries, which is computationally too expensive for us. 

\medskip
\noindent \textbf{Efficiently estimating $\Delta_p$ simultaneously for all $p \in A$.} 
To get around the aforementioned computational bottleneck, we present a novel oversampling technique that allows us to obtain $\poly(k)$ samples of $t$ for all $p \in A$ simultaneously. 
At first thought, this might appear impossible as the values $(\Delta_p)_q$ crucially depends on $p \in A$. 
Our crucial idea here is to perform an ``oversampling'' using the  uniform upper bound vector $2 u_f - g_{\pi^{(t)}} \geq \Delta_{\pi^{(t)}, p^\downarrow}$ for all $p \in A$. In particular, for any $t \in [m]$, by sampling $a \propto \frac{2(u_f)_a - (g_{\pi^{(t)}})_a}{\| 2 u_f - g_{\pi^{(t)}}\|_1}$, we simultaneously have for all $p \in A$ that 
\[
\p(a \in p^\downarrow) \propto B_p^{(t)} \defeq 2 u_f(p^\downarrow) - g_{\pi^{(t)}}(p^\downarrow) \geq \|\Delta_{\pi^{(t)}, p^\downarrow} \|_1 . 
\]
Therefore, by sampling $t \in [m]$ uniformly  and $a \propto \frac{2(u_f)_a - (g_{\pi^{(t)}})_a}{\| 2 u_f - g_{\pi^{(t)}}\|_1}$, and accepting this sample for each $p \in A$ only if $a \in p^\downarrow$, we obtain a way to sample $t \in [m]$ proportional to $B^{(t)}_p$ simultaneously for all $p \in A$. Moreover, as $B^{(t)}_p \geq \|\Delta_{\pi^{(t)}, p^\downarrow} \|_1$, this oversampling still requires only $\poly(k)$ samples $C_p$ to estimate $(\Delta_p)_q$ up to error $O(\frac{1}{k} \|\Delta_p\|_1)$. 

Having obtained samples $C_p$ proportional to $B_p^{(t)}$ for every $p \in A$, we feed them into the procedure $\textsf{Negative-Mass-Estimate}$ (given in \Cref{alg:find_arc}) to compute an unbiased estimate $\Tilde{\Delta}_p$ for all $p \in A$ with additive error $\frac{1}{4k} \norm{\Delta_p}_1$ for every coordinate $q \in V \setminus p^{\downarrow}$. To do so, as mentioned earlier, for each $t \in C_p$ it samples element $q \propto \beta_q$ and uses $\frac{\|\Delta_{\pi^{(t)}, p^{\downarrow}}\|_1}{B_p^{(t)}} \cdot \indicVec{q}$ as the estimate for $\Delta_p$. This estimate satisfies
\[
\E_{t \propto B_p^{(t)}, q \propto \beta_q}\Big [ \frac{\|\Delta_{\pi^{(t)}, p^{\downarrow}}\|_1}{B_p^{(t)}} \cdot \indicVec{q} \Big] = \E_{t \propto B_p^{(t)}} \Big[ \frac{\Delta_{\pi^{(t)}, p^\downarrow}}{B_p^{(t)}} \Big] = \frac{\Delta_p}{\frac{1}{m} \sum_{t \in [m]} B_p^{(t)}} ,
\]
which is a factor of $\frac{1}{m}\sum_{t \in [m]} B_p^{(t)}$ smaller than $\Delta_p$. To fix this, while performing the oversampling described, we also obtain estimates $\Tilde{z}_p$ of the values of $\frac{1}{m}\sum_{t \in [m]} B^{(t)}$ for all $p \in A$ and re-weight properly at the end. See \Cref{arc_finding_alg} and \ref{alg:find_arc} for formal descriptions of our arc-finding procedures.

\begin{algorithm}[htp!]
\caption{Arc Finding for Sequential Algorithm}\label{arc_finding_alg}
    \KwData{A sparsity parameter $\sparsity$, a submodular function $f$ with a $\sparsity$-sparse minimizer, permutations $\{\pi^{(t)}\}_{t \in [m]}$ such that $\frac{1}{m} \sum_{t \in [m]} g_{\pi^{(t)}}$ is a $(\delta, k)$ dual certificate with $\delta \leq \frac{\|u_f\|_\infty}{24k}$, and parameter $\mathsf{Scale} \geq \|u_f\|_\infty \geq -12k  f^*$}
    \KwResult{A set $S_p \subseteq V \setminus p^\downarrow$ of arcs from every $p \in V$  with $(u_f)_p \ge \mathsf{Scale} / 2$, 
    where $S_p = \emptyset$ certifies that $p$ is not in any $\sparsity$-sparse minimizer}
    \SetKwFunction{findarcs}{Arc-Finding}
    \SetKwFunction{negmassest}{NegativeMassEstimate}
    \SetKwProg{Fn}{Function}{:}{}
        \Fn{\findarcs{$f, k, \{\pi^{(t)}\}_{t \in [m]}, \mathsf{Scale}$}}{
        $A \gets \{p: (u_f)_p \ge \frac{\mathsf{Scale}}{2}\}$ \; $S_p, C_p \leftarrow \emptyset$ for all $p \in A$ \tcp*{Initialize set of arcs and samples from $p$ to be empty}
        $N \gets \Theta(k^4 (\log n) \cdot  \frac{\|u_f\|_1}{\|u_f\|_{\infty}})$ \tcp*{Draw $N$ samples in total}
        $\ N_p \gets \Theta(k^4 \log n)$ and $\mathsf{count}_p \gets 0$,  $\forall p \in A$
        \tcp*{Assign $N_p$ samples for each $p \in A$}
        \For{$\mathsf{itr} \in [N]$ \label{line:over_sampling_phase_begin}}{
            Sample a pair $(t, a) \in [m] \times V$ with probability $\frac{2 (u_f)_a - (g_{\pi^{(t)}})_a }{m \cdot (2\norm{u_f}_1 - f(V))}$\;
            \For{every $p \in A$ with $a \in p^{\downarrow}$}{
                $\mathsf{count}_p \gets \mathsf{count}_p + 1$\;
                \lIf{$|C_p| < N_p$}{$C_p \gets C_p \cup \{t\}$} 
            }
        }
        $\Tilde{z}_p \gets \frac{\mathsf{count}_p}{N} \cdot (2 \norm{u_f}_1 - f(V))$ for all $p \in A$ \tcp*{Multiplicative estimate of $\frac{1}{m} \sum_{t \in [m]} B_p^{(t)}$}
        $\{\Tilde{\Delta}_p\}_{p \in A} \gets \negmassest{$f, \{\pi^{(t)}\}_{t \in [m]}, \{C_p\}_{p \in A}, \{\Tilde{z}_p\}_{p \in A}$}$\;
        \For{all $p \in A$ and $q \in V \setminus p^{\downarrow}$}{
            \lIf{$\Tilde{\Delta}_p(q) \ge \frac{3}{4k} \|\Tilde{\Delta}_p\|_{1}$}{$S_p \gets S_p \cup \{q\}$}
        }
        \Return $\{S_p\}_{p \in A}$\;
    }
\end{algorithm}

\begin{algorithm}[htp!]
\caption{Procedure for Estimating $\Delta_p$ for Sequential Algorithm}\label{alg:find_arc}
    \KwData{A submodular function $f$, permutations $\{\pi^{(t)}\}_{t \in [m]}$, a collection of samples $C_p = \{t_i\}_{i \in [N_p]} \subseteq [m]$ and estimates $\Tilde{z}_p$ for all $p \in A$} 
    \tcp*{The samples $t_i$ in $C_p$ satisfies $\Pr[t_i = t] \propto B^{(t)}_p$ and $\E[\Tilde{z}_p] = \frac{1}{m} \sum_{t \in [m]} B_p^{(t)}$}
    \KwResult{An estimate $\Tilde{\Delta}_p$ of $\Delta_p$ for each $p \in A$}
    \SetKwFunction{negmassest}{NegativeMassEstimate}
    \SetKwProg{Fn}{Function}{:}{}
        \Fn{\negmassest{$f, \{\pi^{(t)}\}_{t \in [m]}, \{C_p\}_{p \in A}, \{\Tilde{z}_p\}_{p \in A}$}}{
        \For{$p \in A$}{
            \For{$t \in C_p$}{
                Sample $q \propto \frac{(\Delta_{\pi^{(t)}, p^{\downarrow}})_q}{\|\Delta_{\pi^{(t)}, p^{\downarrow}}\|_1
                } =: \beta_q$, set $\tilde{\Delta}_p^{(t)} \gets \frac{\|\Delta_{\pi^{(t)}, p^{\downarrow}}\|_1}{B^{(t)}_p} \cdot \indicVec q$ \tcp*{$B_p^{(t)} \defeq 2 u_f(p^\downarrow) - g_{\pi^{(t)}}(p^\downarrow)$} \label{line:sample_q}
            }
            $\Tilde{\Delta}_p \gets \frac{1}{|C_p|} \sum_{t \in C_p}\tilde{\Delta}_p^{(t)} \cdot \Tilde{z}_p$\;
        }
        \Return $\{\Tilde{\Delta}_p\}_{p \in A}$
    }
\end{algorithm}

The following is our main guarantee for the $\textsf{Negative-Mass-Estimate}$ routine in \Cref{alg:find_arc}.

\begin{lemma}\label{lem:neg_incr_estimate_guarantee}
Given a submodular function $f: 2^V \rightarrow \R$, subset $A \subseteq V$, 
and permutations $\{\pi^{(t)}\}_{t \in [m]}$.
For each $p \in A$, let $C_p \subseteq [m]$ be $N_p = 10^5 k^4 \log n$ i.i.d. samples of $t_i$ with $\p(t_i = t) \propto B^{(t)}_p$ and 
$\Tilde{z}_p$ be a random variable such that $\E[\Tilde{z}_p] = \frac{1}{m}\sum_{t \in [m]} B^{(t)}_p \leq 10 k \|\Delta_p\|_1$, and $\Tilde{z}_p \in \frac{1}{m}\sum_{t \in [m]} B^{(t)}_p(1 \pm \frac{1}{16 k})$ with high probability. 
Then, the output of \Cref{alg:find_arc} satisfies that with high probability, $(\Tilde{\Delta}_p)_q \in (\Delta_p)_q \pm \frac{1}{8k} \norm{\Delta_p}_1$ and $\|\Tilde{\Delta}_p\|_1 \in (1 + \frac{1}{8k}) \|\Delta_p\|_1$ for all $p \in A$ 
that belongs to a $k$-sparse minimizer and $q \in V \setminus p^\downarrow$. Moreover, \Cref{alg:find_arc} uses $\otilde(k \sum_{p \in A} N_p)$ queries and $\poly(n)$ additional runtime. 
\end{lemma}

\begin{proof}
Fix $p \in A$ that is in a $k$-sparse minimizer. 
For each sample $t \in C_p$ and $q \in V \setminus p^\downarrow$, we have
\begin{align*}
\E_{t \propto B_p^{(t)}, q \propto \beta_q} [\Tilde{\Delta}_p^{(t)} \cdot \Tilde{z}_p] 
& = \E_{t \propto B_p^{(t)}, q \propto \beta_q} \Big[ \frac{\|\Delta_{\pi^{(t)}, p^{\downarrow}}\|_1}{B^{(t)}_p} \cdot \indicVec q \cdot \Tilde{z}_p \Big] = \E_{t \propto B_p^{(t)}} \Big[ \frac{\Delta_{\pi^{(t)}, p^{\downarrow}}}{B^{(t)}_p} \cdot \Tilde{z}_p \Big] \\
& = \frac{\Delta_p}{\frac{1}{m} \sum_{t \in [m]} B_p^{(t)}} \cdot \E[\Tilde{z}_p] = \Delta_p .
\end{align*}
This implies $\E[\Tilde{\Delta}_p] = \Delta_p$, i.e. \Cref{alg:find_arc} outputs an unbiased estimator. 
As $\Tilde{z}_p \in \frac{1}{m}\sum_{t \in [m]} B^{(t)}_p (1 \pm \frac{1}{16 k})$ with high probability, it suffices to show that $\frac{1}{N_p}\sum_{t_i \in C_p} \tilde{\Delta}_p^{(t_i)}$ concentrates around its mean.  
    
To this end, fix a coordinate $q \in V \setminus p^{\downarrow}$. Note that for each sample $t_i \in C_p$, we have 
\[
0 \leq (\Tilde{\Delta}^{t_i}_p)_q \leq \frac{\norm{\Delta_{\pi^{t_i}, p^{\downarrow}}}_1}{B^{(t_i)}_p} \le 1.\]
Let $\mu_q \defeq \E[(\Tilde{\Delta}^{t_i}_p)_q]$.
Using Hoeffding's inequality with error $\gamma =\frac{1}{100k^2}$ gives that \[\Pr\Big[ \Big|\frac{1}{N_p}\sum_{t_i \in C_p} (\tilde{\Delta}_p^{(t_i)})_q -  \mu_q \Big| \geq \gamma \Big] \le 2 \exp\left(- \gamma^2 N_p / 2 \right) \leq 1/\poly(n).
\]
This shows that with high probability,  for all $p \in A$ and $q \in V \setminus p^\downarrow$ we have
\begin{align*}
\Big|(\Tilde{\Delta}_p)_q - (\Delta_p)_q \Big| &= \Big|\frac{1}{N_p}\sum_{t_i \in C_p} (\tilde{\Delta}_p^{(t_i)})_q \cdot \Tilde{z}_p - (\Delta_p)_q \Big| \leq \frac{1}{16k} (\Delta_p)_q + (1 + \frac{1}{16k}) \cdot
\frac{\gamma}{m}\sum_{t \in [m]} B^{(t)}_p \\
& \leq \frac{\|\Delta_p\|_1}{16k} + \frac{17}{16} \cdot \frac{1}{100k^2} \cdot 10k \cdot \|\Delta_p\|_1 \leq \frac{\|\Delta_p\|_1}{8k} ,
\end{align*}
where the second inequality uses the assumption that $\frac{1}{m} \sum_{t \in [m]} B^{(t)}_p \leq 10k \|\Delta_p\|_1$. 
The proof that $\|\Tilde{\Delta}_p\|_1 \in (1 + \frac{1}{8k}) \|\Delta_p\|_1$ is similar. This completes the proof of the first statement of the lemma.

Now we prove the ``moreover'' part of the lemma. Note that the only place where \Cref{alg:find_arc} queries the evaluation oracle is in \Cref{line:sample_q}.
In this line, computing $\|\Delta_{\pi^{(t)}, p^\downarrow}\|_1$ and $B_p^{(t)}$ takes $O(k)$ queries for each $p \in A$ and $t \in C_p$ which is a total of $O(k \sum_{p \in A} N_p)$ queries. 
It therefore suffices to show that we can sample $q$ proportional to $(\Delta_{\pi^{(t)}, p^{\downarrow}})_q$ in $\widetilde{O}(k \sum_{p \in A} N_p)$ queries.
Note that $(\Delta_{\pi^{(t)}, p^{\downarrow}})_q \geq 0$ for all $q \in V \setminus p^\downarrow$ and that for any $i \in V \setminus p^\downarrow$ computing $\sum_{j \in [i] \setminus p^\downarrow} (\Delta_{\pi^{(t)}, p^{\downarrow}})_j$ takes only $O(k)$ queries since $|p^\downarrow| \leq k$. It follows that we can apply binary search as in the proof of \Cref{lem:eff_vsampling} to sample $q \propto \beta_q$ in $O(k \log n)$ queries. For all $p \in A$ and $q \in V \setminus A$, this is a total of $\widetilde{O}(k \sum_{p \in A} N_p)$ queries. Finally, note that every operation in \Cref{alg:find_arc} requires at most $\poly(n)$ additional computational cost. This finishes the proof of the lemma.
\end{proof}

Now, we are ready to present our result regarding the arc-finding procedure. 

\begin{restatable}[Arc Finding for Sequential Algorithm]{lemma}{ArcFindingSequential}
\label{lem:st_arc_finding}
Let $f$ be a submodular function with $f^* \ge - \frac{\|u_f\|_{\infty}}{12 k}$ and $\{\pi^{(t)}\}_{t \in [m]}$ be permutations such that $y = \frac{1}{m}\sum_{t\in[m]}g_{\pi^{(t)}}$ is a $(\delta, k)$ dual certificate with $\delta \le \frac{\|u_f\|_{\infty}}{24 k}$.
Then with high probability, \Cref{arc_finding_alg} outputs arcs $S_p \neq \emptyset$ for each $p \in A$ that belongs to some $\sparsity$-sparse minimizer. 
Moreover, \Cref{arc_finding_alg} uses $\otilde(k^5 \frac{\|u_f\|_1}{\|u_f\|_{\infty}})$ queries and $\poly(n)$ additional runtime. 
\end{restatable}

\begin{proof}[Proof of \Cref{lem:st_arc_finding}]
First, we show that the input passed to \Cref{alg:find_arc} satisfies the conditions in \Cref{lem:neg_incr_estimate_guarantee}.  
We begin by proving that, at the end of the for loop in \Cref{line:over_sampling_phase_begin}, we have collected $N_p$ samples (i.e. $|C_p| = N_p$) for all $p \in A$ with high probability.
To prove this, we fix $p \in A$. In each iteration of \Cref{line:over_sampling_phase_begin}, if $|C_p| < N_p$, then the probability that $C_p$ is updated is $\frac{2 u_f(p^{\downarrow}) - g_{\pi^{(t)}}(p^{\downarrow})}{2 \|u_f\|_1 - f(V)}$. Since $(u_f)_q - (g_{\pi^{(t)}})_q \ge (u_f)_q$ for every $q \in V$, and $-\frac{\|u_f\|_\infty}{12k}\leq f(V) \leq \|u_f\|_1$, we have 
\[
\frac{2 u_f(p^{\downarrow}) - g_{\pi^{(t)}}(p^{\downarrow})}{2 \|u_f\|_1 - f(V)} \ge \frac{u_f(p^\downarrow)}{3 \|u_f\|_{1}} \ge \frac{\|u_f\|_{\infty}}{6 \|u_f\|_{1}} ,
\]
where the last inequality follows because $p \in A$ satisfies $(u_f)_p \ge \|u_f\|_{\infty} / 2$. Therefore, in $N = \Omega(k^4 \log n \frac{\|u_f\|_1}{\|u_f\|_\infty})$ iterations, by the multiplicative Chernoff bound, $|C_p|$ will be updated $\Omega(\frac{N \|u_f\|_{\infty}}{\|u_f\|_{1}}) = \Omega(k^4 \log n)$ times with high probability. 
Moreover, each $t_i \in C_p$ is indeed picked i.i.d. from $[m]$ with $\Pr(t_i = t) \propto B^{(t)}_p$. This is because as long as $|C_p| < N_p$, $C_p$ gets updated if and only if the sample $(t, a)$ drawn satisfies $a \in p^{\downarrow}$. Conditioned on the event that $a  \in p^{\downarrow}$, the sample $t$ has probability distribution over $[m]$ proportional to $B_p^{(t)} = 2 u_f(p^{\downarrow}) - g_{\pi^{(t)}}(p^{\downarrow})$. 

Next, we show that with high probability, $\Tilde{z}_p \in (1\pm \frac{1}{16k}) \cdot (2u_f(p^{\downarrow}) - y(p^{\downarrow})) = (1\pm \frac{1}{16k}) \cdot \frac{1}{m} \sum_{t \in [m]} B_p^{(t)}$ for all $p \in A$. This holds since each time we draw a sample $(t, a)$, the probability that $\mathsf{count}_p$ gets updated is \[
\E_{t \sim [m]}\Big[\frac{2 u_f(p^{\downarrow}) - g_{\pi^{(t)}}(p^{\downarrow})}{2 \|u_f\|_1 - f(V)} \Big] = \frac{2 u_f(p^{\downarrow}) - y(p^{\downarrow})}{2 \|u_f\|_1 - f(V)} \geq \frac{\|u_f\|_\infty}{6 \|u_f\|_1} .
\]
This implies that $\E[\Tilde{z}_p] = 2 u_f(p^\downarrow) - y(p^\downarrow) = \frac{1}{m} \sum_{t \in [m]} B_p^{(t)}$, i.e. $\Tilde{z}_p$ are unbiased estimators.
Moreover, since we draw $N = \otilde(k^4 \frac{\|u_f\|_1}{\|u_f\|_\infty})$ samples in total, multiplicative Chernoff bound again implies that $\Tilde{z}_p$ are indeed $(1\pm \frac{1}{16k})$-multiplicative estimates of $2 u_f(p^{\downarrow}) - y(p^{\downarrow})$ with high probability. 
Lastly, we show that $\E[\Tilde{z}_p] = O(k) \cdot \norm{\Delta_p}_1$ for all $p \in A$ that belongs to a $k$-sparse minimizer. 
To see why this holds, note that 
\[
\|\Delta_p\|_1 = f^{\sharp \rfam}(p^\downarrow) - y(p^\downarrow) = (u_f)_p + f(W \cup p^\downarrow \setminus \{p\}) - y(p^\downarrow) \geq (1 - \frac{1}{6k}) \cdot (u_f)_p - \delta,
\]
where in the last inequality we used that $f^* \geq - \frac{\|u_f\|_\infty}{12k} \geq - \frac{(u_f)_p}{6k}$ and $y(p^\downarrow) \leq \delta$ by \Cref{lem:Delta_thresholds}. 
Using $y_-^{k+1}(V) \geq f^* - \delta$ as $y$ is a $(\delta, k)$ dual certificate, this further implies that \[
\frac{2 u_f(p^{\downarrow}) - y(p^{\downarrow})}{\|\Delta_p\|_1} \le \frac{2k \|u_f\|_{\infty} - (f^* - \delta)}{(1 - \frac{1}{6k}) \cdot (u_f)_p - \delta} \leq O(k), 
\]
where the last inequality follows from $\delta = O(\frac{1}{k} \|u_f\|_\infty)$ and $p \in A$.
Hence, the input passed to \Cref{alg:find_arc} satisfies the assumptions in \Cref{lem:neg_incr_estimate_guarantee}. 
    
Next, we prove the correctness of our arc finding method. Fix $p \in A$ that belongs to some $\sparsity$-sparse minimizer. By \Cref{lem:neg_incr_estimate_guarantee}, we have $(\Tilde{\Delta}_p)_q \in (\Delta_p)_q \pm \frac{1}{8k} \|\Delta_p\|_1$ and $\|\Tilde{\Delta}_p\|_1 \in (1 \pm \frac{1}{8k}) \|\Delta_p\|_1$ with high probability for all $q \in V \setminus p^{\downarrow}$. Therefore, every $q$ with $(\Tilde{\Delta}_p)_q \ge \frac{3}{4k}\|\Tilde{\Delta}_p\|_1$ must have $(\Delta_p)_q \ge \frac{1}{2k}\|\Delta_p\|_1$, implying that $p \to q$ is an arc due to \Cref{lem:Delta_thresholds}. Hence, all arcs deduced are valid. Lastly, note that for each $p \in A$ that belongs to some $k$-sparse minimizer, there must be at least one such $q \in V \setminus p^\downarrow$, as otherwise we would have $(\Delta_p)_q \le \frac{1}{k}\|\Delta_p\|_1$ for all $q$ which is a contradiction to \Cref{lem:Delta_thresholds}. This implies that we can find valid arcs $S_p \neq \emptyset$ for every element $p \in A$ that belongs to some $\sparsity$-sparse minimizer.

Finally, we show that \Cref{arc_finding_alg} can be implemented using $\widetilde{O}(k^5 \cdot (\frac{\|u_f\|_1}{\|u_f\|_\infty} + |A|))$ queries. 
Note that each sample pair $(t,a)$ in \Cref{line:over_sampling_phase_begin} can be implemented in $O(\log n)$ queries and $O(n \log n)$ additional runtime by first sampling $t \sim [m]$ uniformly at random, and then use binary search to sample $a \in V$. Together with the runtime bound in \Cref{lem:neg_incr_estimate_guarantee}, we have that \Cref{arc_finding_alg} can be implemented 
using $\otilde(N + k \sum_{p \in A} N_p) = \otilde(k^5 \frac{\|u_f\|_{1}}{\|u_f\|_{\infty}})$ queries and $\poly(n)$ additional computation, where we used that $|A| \le 2 \frac{\|u_f\|_{1}}{\|u_f\|_{\infty}}$. 
This completes the proof of the lemma. 
\end{proof}

\subsection{Putting It All Together}
\label{subs_stcorrectness}

We are now ready to prove the correctness, query complexity and runtime guarantees for our sequential algorithm. We prove \Cref{thm:main_sequential} and \Cref{thm:main_strongly_sequential} in the following. 

\weaklysequential*

\begin{proof}[Proof of \Cref{thm:main_sequential}]
We prove that the meta algorithm (\Cref{alg:main}) where the subroutines 
 \textsf{Dimensionality-Reduction} and \textsf{Arc-Finding} are instantiated by \Cref{alg:dim_red_sequential} and \Cref{arc_finding_alg} respectively has the desired properties. The correctness of these instantiations of \textsf{Dimensionality-Reduction} and \textsf{Arc-Finding} was established in
 \Cref{lem:dim_red_sequential} and \Cref{lem:st_arc_finding} respectively. Consequently,
 the correctness of this algorithm follows from \Cref{cor:correctness_meta_algorithm} so we only need to analyze its query complexity and runtime here. 

In each iteration of the outer while loop in \Cref{line:outer_while_main} of \Cref{alg:main}, one of the following three conditions occur: (1)  $f^* \leq -\|u_{f^{\sharp \rfam}}\|_\infty / 12k$ and the size of the contracted elements $W$ increases due to \textsf{Dimensionality-Reduction} in \Cref{line:solve_scale_sub-procedure}, (2) there exists element $p \in V\setminus (W \cup D)$ such that $(u_{f^{\sharp \rfam}})_p < 0$ in \Cref{line:negative_upper} of \Cref{alg:extension_maintainer} (extension maintainer) and the size of the discarded elements $D$ increases, or (3) $f^* > -\|u_{f^{\sharp \rfam}}\|_\infty / 12k$  and $\|u_{f^{\sharp \rfam}}\|_\infty$ decreases by a factor of $2$ and a set of arcs $S_p$ are found for every element $p \in V \setminus (D \cup W)$ with $(u_{f^{\sharp \rfam}})_p > \mathsf{Scale}/2$ in the while loop in \Cref{line:arc_finding_main} (in this case, the discarded set $D$ might also increase due to $S_p = \emptyset$ or an element $p$ having more than $k$ arcs). 

(1) and (2) can happen at most $k$ times before $|W| \geq k$ and \Cref{alg:main} outputs $W$. (3) can happen at most $\log (|f| n/ \epsilon)$ times before $\|u_{f^{\sharp \rfam}}\|_\infty \leq \epsilon / n$. 
Consequently, the total number of iterations of the while loop in \Cref{line:outer_while_main} is at most $O(k + \log(|f| n / \epsilon))$. Moreover, each iteration makes $1$ call to \textsf{Dimensionality-Reduction} and at most $k$ calls to \textsf{Arc-Finding}, as otherwise the number of new arcs found from elements $p$ with $(u_{f^{\sharp \rfam}})_p > \mathsf{Scale}/2$ would be more than $k$ and such elements are discarded in our data structure in \Cref{line:discard_cond}.

Each call to \textsf{Dimensionality-Reduction} can be implemented in time $\widetilde{O}(\frac{\|u_{f^{\sharp \rfam}}\|_{1}}{\|u_{f^{\sharp \rfam}}\|_{\infty}} k^{12} \cdot \EO + \poly(n))$ by \Cref{lem:dim_red_sequential}. Additionally, the total number of times \textsf{Dimensionality-Reduction} is called in \Cref{line:solve_scale_sub-procedure} is at most $k + \log(|f| n / \epsilon)$ by the reasoning in the preceding paragraph. Next, each call to \textsf{Arc-Finding}, first computes a $(\delta, k)$ dual certificate $y$ (implicitly given by the permutations $\{\pi^{(t)}\}_{t \in [m]}$), where $\phi = \frac{\|u_{f^{\sharp \rfam}}\|_\infty}{12k}$, $\delta = \frac{\|u_{f^{\sharp \rfam}}\|_\infty}{24k} = \phi/2$, which by \Cref{lem:st_certif_runtime_FTRL} uses 
\[
\otilde(k^6 \delta^{-4} \phi^2 (\|u_{f^{\sharp \rfam}}\|_{\infty} + \phi) (\|u_{f^{\sharp \rfam}}\|_1 + \phi)) \leq \otilde(k^8 \frac{\|u_{f^{\sharp \rfam}}\|_1}{\|u_{f^{\sharp \rfam}}\|_\infty})  
\]
queries to $\EO$ and $\poly(n)$ additional runtime.
Given the $(\delta, k)$ dual certificate, \Cref{lem:st_arc_finding} states that each call to \textsf{Arc-Finding} can be done in $\widetilde{O}(\frac{\|u_{f^{\sharp \rfam}}\|_{1}}{\|u_{f^{\sharp \rfam}}\|_{\infty}} k^5 \cdot \EO + \poly(n))$ runtime. 
Since $\frac{\|u_{f^{\sharp \rfam}}\|_{1}}{\|u_{f^{\sharp \rfam}}\|_{\infty}} \le n$, we obtain that the runtime due to \textsf{Dimensionality-Reduction} and \textsf{Arc-Finding} is $\widetilde{O}((n k^{12} \cdot \EO + \poly(n)) \log(|f| /\epsilon))$ in total.

Finally, due to \Cref{thm:ext_maintainer}, each update to the \textsf{RingFamily} can be implemented in $O(m \cdot \EO + nk)$ time where $m$ is the total number of elements $p$ from which arcs are found. Additionally, each one of \Cref{line:update_infos1}, \Cref{line:update_infos2} and \Cref{line:update_infos3} can be implemented in $O(1)$ depth and $O(n \cdot \EO + n)$ time. 
Combining everything above, \Cref{alg:main} finds an $\epsilon$-approximate minimizer for $\sparsity$-sparse SFM in runtime $\widetilde{O}((n k^{12} \cdot \EO + \poly(n)) \cdot \log (|f| / \epsilon))$.
\end{proof}

We now prove our strongly-polynomial result for our sequential algorithm.

\stronglysequential*

\begin{proofof}{\Cref{thm:main_strongly_sequential}}
Consider the meta algorithm in \Cref{alg:main} with subprocedures \textsf{Dimensionality-Reduction} and \textsf{Arc-Finding} as in Algorithms \ref{alg:dim_red_sequential} and \ref{arc_finding_alg}, where we set $\epsilon = 0$. The correctness of these instantiations of \textsf{Dimensionality-Reduction} and \textsf{Arc-Finding} was established in
 \Cref{lem:dim_red_sequential} and \Cref{lem:st_arc_finding} respectively. Consequently,
 the correctness of this algorithm follows from \Cref{cor:correctness_meta_algorithm} so we only need to analyze its query complexity and runtime here. 
We show that this algorithm outputs an exact minimizer and has runtime bound $\otilde(n k^{13} \cdot \EO + \poly(n))$. 
    
First, the runtime due to \textsf{RingFamily} operations can be easily bounded by $O(n k \cdot \EO + \poly(n))$, since on average each arc takes at most $O(1 \cdot \EO + nk)$ to update, and there can be at most $nk$ arcs in total (see \Cref{thm:ext_maintainer}).
Hence, it suffices to reason about the complexity of \textsf{Dimensionality-Reduction} and \textsf{Arc-Finding} over the course of the algorithm. 

We have already argued that in the proof of \Cref{thm:main_sequential} above that one iteration of the while loop in \Cref{line:outer_while_main}, the number of calls to \textsf{Dimensionality-Reduction} and \textsf{Arc-Finding} is $1$ and $O(k)$ respectively. Hence, each such iteration takes $\widetilde{O}(\frac{\|u_{f^{\sharp \rfam}}\|_{1}}{\|u_{f^{\sharp \rfam}}\|_{\infty}} k^{12} \EO + \poly(n))$ time, since each call of \textsf{Arc-Finding} takes $\otilde(k^8 \frac{\|u_{f^{\sharp \rfam}}\|_1}{\|u_{f^{\sharp \rfam}}\|_\infty})$ (as shown in the proof of \Cref{thm:main_sequential}).

To obtain a strongly polynomial runtime, we bound the runtime of the calls to \textsf{Dimensionality-Reduction} and \textsf{Arc-Finding} in a more fine-grained way than in the proof of \Cref{thm:main_sequential}. We introduce a charging scheme to argue that the total cost of these calls is $\otilde(k^{13})$ per element $p$. We do so by charging the number of queries made during each iteration of the while loop to elements in $V$ according to the entries of the vector $u_{f^{\sharp \rfam}}$. In particular, we charge $\otilde(\frac{(u_{f^{\sharp \rfam}})_p}{\|u_{f^{\sharp \rfam}}\|_{\infty}} k^{12})$ queries to each $p \in V$. 
Note that such a charging scheme accounts for the total number of queries used in one iteration of the outer while loop. Hence, it suffices to show that the total charge to element $p$ is at most $\otilde(k^{13})$.

Now, fix an element $p \in V$ and observe that $(u_{f^{\sharp \rfam}})_p$ can only change when we deduce an arc from $p$ (directly or through transitive closure). Moreover, $(u_{f^{\sharp \rfam}})_p$ can change at most $k$ times as any element reaching more than $k$ arcs is discarded by \Cref{line:discard_cond} in \Cref{alg:extension_maintainer}.  
Now, fix a 
sequence of iterations $t_0, \cdots, t_1$ of the while loop in \Cref{line:outer_while_main} where the value $(u_{f^{\sharp \rfam}}^{(t)})_p$ in the $t$th iteration is constant throughout $t \in \{t_0, \cdots, t_1\}$. 
Note that $\|u_{f^{\sharp \rfam}}^{(t)}\|_\infty$ decreases by a factor of $2$ after each iteration $t$ and that $\|u_{f^{\sharp \rfam}}^{(t_1)}\|_\infty \geq (u_{f^{\sharp \rfam}}^{(t_1)})_p$. 
We can therefore bound the total number of queries charged to $p$ in iterations $t_0, \cdots, t_1$ as 
\[
\otilde(k^{12}) \cdot \sum_{t=t_0}^{t_1} \frac{(u_{f^{\sharp \rfam}}^{(t)})_p}{\|u_{f^{\sharp \rfam}}^{(t)}\|_\infty} \leq \otilde(k^{12}) \sum_{t=t_0}^{t_1} 2^{-(t_1 - t)} \leq \otilde(k^{12}) . 
\]
Since there are at most $k$ such sequences of iterations for $p$, the total number of queries charged to $p$ throughout the entire algorithm is at most $\otilde(k^{13})$. This proves that the algorithm makes at most $\otilde(n k^{13})$ queries to $\EO$ in total.  

Finally, to prove that the additional computation is also strongly polynomial, we note that each iteration of the while loop in \Cref{line:outer_while_main} deduces either a dimensionality reduction or an arc. Note each one of \Cref{line:update_infos1}, \Cref{line:update_infos2} and \Cref{line:update_infos3} can be implemented in $O(1)$ depth and $O(n \cdot \EO + n)$ time. 
Additionally, \Cref{line:outer_while_main} gets called at most $nk$ times. Since each iteration uses $\poly(n)$ additional computation, the total runtime we obtain is $\otilde(n k^{13} \cdot \EO + \poly(n))$. 
\end{proofof}

%% file: extension.tex
\section{Ring Family and Extensions} 
\label{sec:ring_family_extensions}

In this section, we present more details on arc information and the submodular extension $f^{\sharp \rfam}$ introduced in \Cref{sec:framework:extension}. 

As was illustrated in \Cref{sec:framework:extension}, besides finding elements which are contracted (added to $W$) or discarded (added to $D$), our algorithms also proceed by finding elements $p,q \in V$ such that any $\sparsity$-sparse minimizer of the submodular function $f$ containing $p$ must also contain $q$. Such information is captured by an \emph{arc} $(p,q)$. To incorporate all the arc information, we use the framework of SFM over ring families first introduced by Iwata, Fleischer, and Fujishige \cite{IFF01}. This framework was later used in many other algorithms, e.g., \cite{IO09, LSW15,DVZ21}.

Here, we emphasize once more the slight difference between the arc definition we use and the standard one in the literature. As said in \Cref{footnote:arc_difference}, an arc $(p, q)$ in our work means that every $k$-sparse minimizer containing $p$ contains $q$. Equivalently, this can be reformulated as ``if the minimal minimizer $S_{\min}$ contains element $p$, then it also contains element $q$''. Crucially, we use this definition to limit the number of arcs we expect to deduce from any element to at most $k$. Hence, whenever we deduce more than $k$ arcs from some element, we discard it, as it is not part of the minimal minimizer. This also affects how we implement the transitive closure in \textsf{UpdateArcs}, the function in \Cref{alg:extension_maintainer} used to add new arcs to the data structure. For more details, see the proof of \Cref{thm:ext_maintainer}. 
The modification in the definition of arcs, along with the corresponding algorithmic tweaks, allows a speed-up in both the depth and query complexity of calls to the data structure \Cref{alg:extension_maintainer}, making our runtimes (\Cref{thm:main_parallel}, \Cref{thm:main_sequential}, \Cref{thm:main_strongly_sequential}) possible. 
Our presentation and notations follow the ones in \cite{DVZ21}, with necessary adaptions to fit the arc definition of this paper. 

\medskip
\noindent \textbf{Directed graph and partial ordering.}
All arc information is maintained as a directed graph $G = (V', E)$, where $V':= V \setminus (W \cup D)$ with
the property that if $(p,q) \in E$, then every $k$-sparse minimizer $S \subseteq V$ of the submodular function $f$ that contains $p$ must also contain $q$. 
We may assume that $E$ is acyclic as any (directed) cycle of $G$ can be contracted (e.g., Section 4.1 of \cite{DVZ21}). 
We may also assume that $G$ is {\em transitive}, i.e., $(p,q) \in E$ if and only if there is a directed path from $p$ to $q$ in $G$, by maintaining its transitive closure. Our data structure, \Cref{alg:extension_maintainer}, maintains this property by adding arc $p \to q$ whenever there is a directed path from $p$ to $q$. 

The acyclic graph $G = (V', E)$ defines a partial order  $\preceq_E$, i.e., $q \preceq_E p$ if there exists a directed path in $E$ from $p$ to $q$ (and $(p,q) \in E$ when $G$ is transitive); in particular, $p \preceq_E p$ for all $p \in V$. We say that an ordering of the vertices is {\em consistent with} $\preceq_E$, if $q$ is ordered before $p$ whenever $q \preceq_E p$. We use the simpler notation $q \preceq p$ if $E$ is clear from the context.

\medskip
\noindent \textbf{Ring families.} The acyclic graph $G = (V',E)$ defines a collection of sets $\rfam(E, W, D) \defeq \{S: W \subseteq S \subseteq V \setminus D, (p \in S) \wedge (q \preceq p) \Rightarrow q \in S\}$, i.e., $\rfam(E, W, D)$ is the set of all lower ideals of the partial order $\preceq$. 
Note that $\rfam(E, W, D)$ contains $W$ and $V$ and is closed under set intersection and union. 

Our algorithms begin with no arc information, so $\rfam(E, W, D) = 2^V$. 
Throughout our algorithm, we maintain the ring family $\E \subseteq 2^V$ (by the arcs in $G = (V',E)$) to be {\em consistent} with all the $k$-sparse minimizers of the submodular function $f$, i.e., $\rfam(E, W, D)$ contains all the $\sparsity$-sparse minimizers of $f$.

Given a ring family $\rfam(E, W, D)$, for any element $p \in V'$, we define 
\[
p^\downarrow \defeq \{q \in V': q \preceq p\} \quad \text{and} \quad p^\uparrow \defeq \{q \in V': p \preceq q\}.
\]
Since we have assumed that $G = (V', E)$ is transitive, for all $p \in V'$, $p^\downarrow = \{p\} \cup \{q \in V': (p, q) \in E\}$ and $p^\uparrow = \{p\} \cup \{q \in V': (q, p) \in E \}$. 
Similarly, for any $X \subseteq V'$, we define 
\[
X^\downarrow \defeq \bigcup_{p \in X} p^\downarrow \quad \text{and} \quad X^\uparrow \defeq \bigcup_{p \in X} p^\uparrow.
\]
Note that $X^\downarrow \cup W$ is the unique minimal set in $\rfam(E, W, D)$ containing $X$, and $(V' \setminus X^\uparrow) \cup W$ is the unique maximal element of $\E$ disjoint from $X$. 
We also define, for any $X \subseteq V'$ the set 
\[
X^{\sharp \rfam} \defeq V' \setminus (V' \setminus X)^\uparrow\,.
\]
Note that $X^{\sharp \rfam} \cup W$ is the unique maximal set in $\rfam(E, W, D)$ that is contained in $X \cup W$.

\medskip
\noindent \textbf{Upper bound values and extensions of submodular functions.} 
Consider $f$ a submodular function and a consistent ring family with arc constraints captured by $\rfam$. 
For every $p \in V'$, we define the upper bound values 
\[(u_{f, \rfam})_p \defeq f(W \cup p^\downarrow) - f(W \cup p^\downarrow \setminus \{p\}) 
.\]
An intuitive property of $(u_{f, \rfam})_p$ is that for any set $X \in \rfam(E, W, D)$ that does not contain $p$ and that $X \cup \{p\} \in \rfam(E, W, D)$, $(u_{f, \rfam})_p$ is an upper bound on the marginal value $f(X \cup \{p\}) - f(X)$. 
Recall from \Cref{lem:properties_extension} that the upper bound values $(u_{f, \rfam})_p$ coincide with $(u_{f^{\sharp \rfam}})_p \defeq f^{\sharp \rfam}(\{p\}) - f^{\sharp \rfam}(\emptyset)$ whenever $(u_{f, \rfam})_p \geq 0$. 

Using the upper bound values, we define an extension $f^{\sharp \rfam}$ of the submodular function $f$ that captures the ring family structure $\rfam(E, W, D)$. 
\begin{definition}[Submodular Extension]
\label{defn:extension}
Given a submodular function $f$ and a ring family $\rfam(E, W, D)$ $k$-consistent with the structure of its minimizers, define, for $S \subseteq V \setminus \{W \cup D\}$ 
\[
f^{\sharp \rfam}(S) \defeq f_W(S^{\sharp \rfam}) + u_{f, \rfam}^{+}(S \setminus S^{\sharp \rfam}) = f(W \cup S^{\sharp \rfam}) + u_{f, \rfam}^{+}(S \setminus S^{\sharp \rfam}).\]
\end{definition}

To incorporate set $W$, the extension $f^{\sharp \rfam}$ replaces the function $f$ with $f_W: 2^{V \setminus W} \rightarrow \mathbb{R}$ defined as $f_W(S) \defeq f(S \cup W) - f(W)$, which is submodular.
For a submodular function $F:2^{V'} \to \R$, the operation $F^\downarrow$ with respect to a ring family was defined 
in \cite{DVZ21}; we denote by $\Bar{F}$ the complement of $F$, with the relationship $\Bar{F}(S) = F(V' \setminus S)$. The function $f^{\sharp \rfam}$ is the natural complement of the function $(f_W)^\downarrow$. In particular, it is easy to see that $f^{\sharp \rfam}(S) = (\Bar{f_W})^{\downarrow}(S)$. 
We now prove the following lemma collecting the properties of $f^{\sharp \rfam}$, which was previously stated in \Cref{sec:framework:extension}.

\ExtensionProperty*

\begin{proof}
To prove the submodularity property, we note from above that $f^{\sharp \rfam}(S) = (\Bar{f_W})^{\downarrow}(V \setminus S)$, where $\Bar{f_W}(S) = f_W(V \setminus S)$ is the complement set function of $f_W$. $\Bar{f_W}^{\downarrow}$ is the submodular extension defined in \cite{DVZ21} for function $F = \Bar{f_W}$. Since $\Bar{f_W}^{\downarrow}$ was shown to be submodular (Lemma 4.4 in \cite{DVZ21}), it follows that $f^{\sharp \rfam}$ is also a submodular function.

The first part of the extension property follows from \Cref{defn:extension} and the non-negativity of $u_{f, \rfam}^+$; the second part follows because $S^{\sharp \rfam} = S$ for any set $S$ so that $W \cup S$ is consistent with all the arcs. 

We now prove the first part of $k$-consistency. By the extension property, $f^{\sharp \rfam}(S) \geq f(W \cup S^{\sharp \rfam})$ and that $f^{\sharp \rfam}(S^*) = f(W\cup S^*)$ for any set $S^*$ so that $S^* \cup W$ is a $\sparsity$-sparse minimizer $f$. It follows that the minimum value of $f^{\sharp \rfam}$ agrees with $f^*$ and any $\sparsity$-sparse minimizer $S^*$ of $f$ is also a minimizer of $f^{\sharp \rfam}$. 
For the second part of $k$-consistency, we note that $f((S^*)^{\sharp \rfam} \cup W) \leq f^{\sharp \rfam}(S^*) = f^*$ which implies that $(S^*)^{\sharp \rfam}$ is a minimizer of $f^{\sharp \rfam}$.


To prove the marginal property, we consider two cases. The first case is when there is no element $q \neq p$ such that $q \preceq p$, or, equivalently, $\{p\} = p^\downarrow$. Then, $\{p\}^{\sharp \rfam} = \{p\}$, so $(u_{f^{\sharp \rfam}})_p = f^{\sharp \rfam}(\{p\}) - f^{\sharp \rfam}(\emptyset) = f(W \cup \{p\}) - f(W) = (u_{f, \rfam})_p$. 
In the other case, there exists $q \neq p$ such that $q \preceq p$, and so  $\{p\}^{\sharp \rfam} = \emptyset$. It follows that $f^{\sharp \rfam}(\{p\}) = f(W) + u_{f, \rfam}^{+}(\{p\})$, so $(u_{f^{\sharp \rfam}})_p = u_{f, \rfam}^{+}(\{p\})$. Hence, $(u_{f^{\sharp \rfam}})_p = (u_{f, \rfam})_p$ if $(u_{f, \rfam})_p \geq 0$, and $(u_{f^{\sharp \rfam}})_p = 0$ if $(u_{f, \rfam})_p < 0$. 
Since adding new arcs do not decrease the set $p^\downarrow$ for any element $p \in V$, the ``Moreover'' part immediately follows from the submodularity of $f$. 
\end{proof}

\subsection*{Missing Proofs for Extension Maintainer}
\label{subsec:ext_maintainer}

The proof of \Cref{lem:properties_extension} follows immediately from the above. The data structure for the extension maintainer is formally given in \Cref{alg:extension_maintainer}. 
We now prove that \Cref{alg:extension_maintainer} satisfies the properties in \Cref{thm:ext_maintainer}, which is restated below.

\ExtMaintainer*

\begin{algorithm}[htp!]
\caption{Data Structure: Extension Maintainer}\label{alg:extension_maintainer}
    \textbf{State maintained:} \tcp*{Read-only access} 
    Ground set $V$, discarded elements $D \subseteq V$, and contracted elements $W \subseteq V$
    \tcp*{Elements in $D$ not in any $k$-sparse minimizer, elements in $W$ belong to every minimizer}
    Explicit array of sets $p^\downarrow \subseteq V \setminus (D \cup W)$ for each $p \in V \setminus (D \cup W)$ \tcp*{Closure of all arcs}
    Sparsity parameter $k \in \mathbb{Z}_{\geq 0}$ \;
    Explicit arrays of $(u_{f, \rfam})_p, (u_{f^{\sharp \rfam}})_p \in \mathbb{R}$ for all $p \in V \setminus (D \cup W)$ \\
    \SetKwFunction{init}{Init}
    \SetKwProg{Fn}{Function}{:}{}
    \Fn{\init{$V,k, f$}}{
    $V \leftarrow V$, $D \leftarrow \emptyset$, $W \leftarrow \emptyset$, $k \leftarrow k$, $f \leftarrow f, p^\downarrow \leftarrow \{p\} \enspace  \forall p \in V$ \;
    $\mathsf{self}.\mathsf{UpdateSpace}(\emptyset, \emptyset)$
    }
    \SetKwFunction{updatespace}{UpdateSpace}
    \SetKwProg{Fn}{Function}{:}{}
    \Fn{\updatespace{$W^{\mathrm{add}}, D^{\mathrm{add}}$}}{
        \lFor{$p \in D^{\mathrm{add}}$}{$D \gets D \cup p^\uparrow$ \label{line:discard_more}}
        \If{$W^{\mathrm{add}} \ne \emptyset$}{
        $W \leftarrow W \cup_{p \in W^{\mathrm{add}}} p^\downarrow$\;
        \For{$p \in V \setminus (W \cup D)$}{
        $(u_{f, \rfam})_p \leftarrow f(W \cup p^\downarrow) - f(W \cup p^\downarrow \setminus p)$ and $(u_{f^{\sharp \rfam}})_p \gets f(W \cup \{p\}^{\sharp \rfam}) + u_{f, \rfam}^{+}(\{p\} \setminus \{p\}^{\sharp \rfam})$ \label{line:change_uppers}
        }
        }
        \While{$\exists p: (u_{f^{\sharp \rfam}})_p < 0$ \label{line:neg_upper_cond}}{
        $W \gets W \cup_{p: (u_{f^{\sharp \rfam}})_p < 0} p^\downarrow$ \label{line:negative_upper}\;
        \lFor{$p \in V \setminus (W \cup D)$}{update $(u_{f, \rfam})_p, (u_{f^{\sharp \rfam}})_p$ as in \Cref{line:change_uppers}
        }
        }
    }
    \SetKwFunction{updatearcs}{UpdateArcs}
    \SetKwProg{Fn}{Function}{:}{}
    \Fn{\updatearcs{$\{S_p\}_{p \in V}$}}{
    $p^\downarrow \leftarrow p^\downarrow \cup S_p$ for each $p \in V \setminus (D \cup W), W^{\mathrm{new}}, D^{\mathrm{new}} \gets \emptyset$\; \label{line:transclose_start}
    Update each $p^\downarrow$ to its transitive closure (possibly adding elements to $W^{\mathrm{new}}$) \tcp*{Parallel implementation in \Cref{thm:ext_maintainer}} \label{line:transclose_end}
    Update $(u_{f, \rfam})_p, (u_{f^{\sharp \rfam}})_p$ for every $p$ with new arcs \label{line:update_uppers_for_new_arcs}\;
    \For{$p \in V\setminus D$}
    {
    \If{$|p^\downarrow| > k$ or $p^\downarrow \cap D \neq \emptyset$}
        {
        $D^{\mathrm{new}} \leftarrow D^{\mathrm{new}} \cup \{p\}$ \tcp*{$p$ is discarded if it has $>k$ arcs, or an arc to $D$} \label{line:discard_cond}
        }
    }
    $\mathsf{self}.\mathsf{UpdateSpace}(W^{\mathrm{new}}, D^{\mathrm{new}})$
    }
    
    \SetKwFunction{subgrad}{Subgrad}
    \SetKwProg{Fn}{Function}{:}{}
    \Fn{\subgrad{$\pi$}}{
    \lIf{$D$ is not at the end of $\pi$}
    {\Return ``Error''}
    $g_i \leftarrow (f(\pi[i]^{\sharp \rfam}) + u_{f, \rfam}^{+}(\pi[i] \setminus \pi[i]^{\sharp \rfam})) - (f(\pi[i-1]^{\sharp \rfam}) + u_{f, \rfam}^{+}(\pi[i-1] \setminus \pi[i-1]^{\sharp \rfam}))$ for all $i \in V \setminus (D \cup W)$ \;
    \Return vector $g \in \mathbb{R}^{|V\setminus D|}$ \tcp*{The full sub-gradient, for parallel algorithm}
    }
    
\SetKwFunction{partial}{Partial}
\SetKwProg{Fn}{Function}{:}{}
\Fn{\partial{$i, \pi$}}{
\lIf{$i \in D$}
{\Return ``Error''}
$g_i \leftarrow (f(\pi[i]^{\sharp \rfam}) + u_{f, \rfam}^{+}(\pi[i] \setminus \pi[i]^{\sharp \rfam})) - (f(\pi[i-1]^{\sharp \rfam}) + u_{f, \rfam}^{+}(\pi[i-1] \setminus \pi[i-1]^{\sharp \rfam}))$ \;
\Return $g_i \in \mathbb{R}$ \tcp*{The $i$th coordinate of sub-gradient $g_{\pi}$, for sequential algorithm}
}   
\SetKwFunction{fsharp}{$f^{\sharp \rfam}$}
\SetKwProg{Fn}{Function}{:}{}
\Fn{\fsharp{$S$}}{
\Return $f(S^{\sharp \rfam}) + u_{f, \rfam}^{+}(S \setminus S^{\sharp \rfam})$
}
\end{algorithm}

\begin{proof}[Proof of \Cref{thm:ext_maintainer}]
First, we show that \Cref{alg:extension_maintainer} indeed preserves the invariants regarding $W, D, u_{f^{\sharp \rfam}}$ claimed in the theorem ($k$-consistency and the non-negativity of $u_{f^{\sharp \rfam}}$), and then we prove the runtime and depth bounds claimed. 
Note that \Cref{line:neg_upper_cond} ensures
that after any call of \textsf{UpdateSpace}, we have $(u_{f^{\sharp \rfam}})_p \ge 0, \forall p \in V \setminus (D \cup W)$. This proves that after every call to any function of the algorithm, the values $(u_{f^{\sharp \rfam}})_p$ remain non-negative. Next, note that as long as $W \subseteq S^*_{\min}$, the minimizer of $f_W$ is $S^*_{\min} \setminus W$. As argued before, any element $p$ with $f(W \cup \{p\}) < f(W)$ belongs to every minimizer of $f_W$, which implies that if $W \subseteq S^*_{\min}$ every time \Cref{line:negative_upper} is called, the new value of $W$ will also be a subset of $S^*_{\min}$. By the assumption that $W^{\mathsf{begin}}, W^\mathrm{new} \subseteq S^*_{\min}$, adding every $p^\downarrow$ with $p \in W^\mathrm{new}$ to set $W^{\mathsf{begin}}$ yields a subset of $S^*_{\min}$, which means that $W \subseteq S^*_{\min}$ is maintained throughout the call \textsf{UpdateSpace}. 
We now prove the runtime and depth bounds for the functions in \Cref{alg:extension_maintainer}.

For \textsf{Init}$(V,k)$, we need $O(n)$ time for the initialization of $V$, $D$, $k$ and $p^\downarrow$ (which can all be done in $O(1)$-depth in parallel), and one call of \textsf{UpdateSpace}. 

For \textsf{UpdateSpace}, we separately bound the time it takes to augment the set $W$ from $W^{\mathsf{begin}}$ to $W^{\mathsf{end}}$ and the set $D$ during one function call. 
First, each time we augment $W$ by at least one element, it takes $O(1)$ depth and $O(n \cdot \EO + n)$ time to recompute the $(u_{f, \rfam})_p, (u_{f^{\sharp \rfam}})_p$ values. Since we make at most $|W^{\mathsf{end}} \setminus W^{\mathsf{begin}}|$ updates to $W$, this gives the desired bound. For augmenting $D$, we need not recompute any $(u_{f, \rfam})_p, (u_{f^{\sharp \rfam}})_p$ values, as discarding elements does not affect them. Moreover, excluding elements that have arcs to elements newly discarded (\Cref{line:discard_more}) can be done with no queries to $f$, by simply traversing the edges of the graph via BFS (using the edge directions for the former operations and the inverse directions for the latter), which takes $O(n)$ time.

For \textsf{Subgrad}$(\pi)$, note that $\pi[i]^{\sharp \rfam}$ is monotonitically increasing with $i$, so we can compute them one by one using a total of $O(n)$ time. We then need to make $n$ queries, one to each $f(\pi[i]^{\sharp \rfam})$. These operations can be done in $O(1)$ parallel round. 
The analysis for \textsf{Partial}$(i)$ is similar to \textsf{Subgrad}$(\pi)$. For $f^{\sharp \rfam}(S)$, it's easy to see that at most $1$ function call to $f$ is needed. Constructing $S^{\sharp \rfam}$ takes $O(n)$ time, and so does computing $u_{f, \rfam}^{+}(S \setminus S^{\sharp \rfam})$. When $S = p^\downarrow$, we simply return $f(S)$, meaning that we only do one query, with no additional computation.

Finally, we prove the lemma statement for \textsf{Update}$(\{S_p\}_{p \in V})$, which requires a very careful analysis. 
In Line \ref{line:transclose_start} - \ref{line:transclose_end}, we need to first compute the transitive closure when the new arcs corresponding to vertices in each $S_p$ are added. Direct implementation of this step could take $O(n)$ parallel depth, e.g., when the new arcs form a chain of length $O(n)$ and propagating the information of the last arc along the chain takes $O(n)$ depth. But clearly, such a large parallel depth is unnecessary for the correctness of the entire function, since each element $p \in V\setminus D$ can have at most $k$ arcs or else it will be discarded in \Cref{line:discard_cond}. 

Therefore, we can implement Line \ref{line:transclose_start} - \ref{line:transclose_end} in depth $k$ by using a parallel BFS approach where each $p$ runs the following procedure in parallel for $k$ iterations: $p$ checks every element $q \in p^\downarrow$ and update new arcs in each $q^\downarrow$. Note that this will not compute the full transitive closure, but suffices for the condition in \Cref{line:discard_cond}, and can correctly compute $p^\downarrow$ for all element $p$ that will not be discarded after the update.  

Note that the above strategy can also be implemented in total time $O(nk)$, since we can freeze (and discard it in \Cref{line:discard_cond}) each element $p$ whenever $|p^\downarrow| > k$. This way each element adds at most $k$ new arcs, with a total runtime of $O(kn)$ for Line \ref{line:transclose_start} - \ref{line:transclose_end}. 

Lastly, \Cref{line:update_uppers_for_new_arcs} of the function \textsf{Update}$(\{S_p\}_{p \in V})$ can be computed using $O(1)$ depth, $O(m)$ queries to $\EO$, one for each $p$ whose $(u_{f, \rfam})_p, (u_{f^{\sharp \rfam}})_p$ needs to be updated, and additional $O(nk)$ time. 
\end{proof}

%% file: appendix.tex
\newpage
\section{Optimization Methods}
\label{sec:missing_proofs}

\subsection{Stochastic Follow-the-Reglarized-Leader}

In this subsection, we discuss the stochastic follow-the-reglarized-leader (stochastic FTRL) algorithm (see \Cref{st_FTRL}) and its corresponding performance guarantees.

\begin{algorithm}[htp!] 
\caption{Stochastic Follow-the-Regularized-Leader}\label{st_FTRL}
    \KwData{A convex function $f$ on $\mathbb{R}^V$, a convex domain $D \subseteq \mathbb{R}^V$, the step size $\eta > 0$ for the update, and the number of iterations $m$}
    \KwResult{$\{x_0, x_1, \ldots, x_m\} \subseteq D$, the sequence of iterates generated, which satisfy the guarantees stated in \Cref{lem:stMD_guarantee}}
    \SetKwFunction{stFTRL}{StochasticFTRL}
    \SetKwProg{Fn}{Function}{:}{}
    \Fn{\stFTRL{$f, D, \eta, m$}}{
    $x_0 \gets \argmin_{x \in D} r(x)$ \tcp*{$x_0$ is the initial point, $r(x)$ is the entropy regularizer}
    \For{$t = 0, 1, \ldots, m$}{
        Sample stochastic subgradient $h_t$, with $\E[h_t] = g_{x_t}$\;
        $x_{t+1} = \argmin_{x \in D} \eta \sum_{j=0}^t h_j^\top x +  r(x)$ \tcp*{Recall that $r(x) = \sum_{i \in V} x_i \log x_i $}
    }
    \Return $\{x_0, x_1, \ldots, x_m\}$\;
    }
\end{algorithm}

We prove the following guarantees for \Cref{st_FTRL}. 

\begin{lemma}
\label{lem:stMD_guarantee}
    For any point $w \in D$, where $D$ is as in \Cref{st_FTRL}. Assume that $\eta \norm{h_t}_{\infty} < 1/2$ with probability $1$ 
    for all $t \in \{0, \ldots, m-1\}$. Then, the iterates of \Cref{st_FTRL} satisfy \begin{equation}\label{ineq:expectation_bound}
        \E\left[\sum_{t=0}^{m-1} \langle g_{x_t}, x_{t} - w \rangle\right] \le \frac{\sup_{x \in D} r(x) - \inf_{y \in D} r(y)}{\eta} + \eta \sum_{t=0}^{m-1} \E[\|h_t\|_{x_t}^2] .
    \end{equation} 
    Furthermore, define the stochastic process $\{w^{(0)}, w^{(1)}, \ldots, w^{(m)}\} \subseteq \R^V$ as $w^{(0)} \defeq \Vec{1}$ and $w^{(t+1)}_i \defeq w^{(t)}_i \exp(- \eta (h_t)_i)$, for all $i \in V$ and $t \in \{0,\ldots, m-1\}$. Define stochastic vector $p^{(t)} \in \R^V$ with coordinates $p^{(t)}_i \defeq \frac{w^{(t)}_i}{\norm{w^{(t)}}_1}$. Let $v^*$ be a number such that with probability 1, 
\[
-\eta \langle p^{(t)}, g_{x_t} \rangle + \eta^2 \E_t\left[\norm{h_t}_{p^{(t)}}^2\right] \leq v^* , \enspace \forall t \in \{0, \ldots, T-1\} ,  
\]  
where the expectation $\E_t$ is taken over $h_t$ conditioned on $\{h_0,\ldots, h_{t-1}\}$. 
Then with probability at least $1 - \rho$,  we have 
\begin{align}\label{ineq:ell_infinity_bound_gen}
        \max_{i \in V} \sum_{t = 0}^{m-1} - \eta (h_t)_i \le m v^* + \log \Big(\frac{n^2}{\rho} \Big).
\end{align}  
\end{lemma}

The proof of \Cref{lem:stMD_guarantee} relies on \Cref{lem:MP_proof_tweak} and \Cref{lem:stocastic_MW} which we prove below.

\begin{proof}[Proof of \Cref{lem:stMD_guarantee}]
First, we prove \eqref{ineq:expectation_bound}. For each $T \in \{0, \ldots, m-1\}$, define potential function $\Phi_T(x) \defeq \langle \eta \sum_{t=0}^T h_t, x \rangle + r(x)$ and $\Phi_{-1}(x) \defeq r(x)$. We later show that for any outcome of the randomness of $x_T$ and $h_T$, we have 
\begin{equation} \label{x_t_ineq}
    \Phi_T(x_{T+1}) \ge \Phi_{T-1}(x_{T}) + \eta h_T^{\top} x_T - \eta^2 \norm{h_T}_{x_T}^2, \quad \forall T \in \{0, \ldots, m-1\} .
\end{equation}
    
    For now, assume \eqref{x_t_ineq} holds. Summing up inequalities \eqref{x_t_ineq} for all $T \in \{0, \ldots, m-1\}$ yields \[\Phi_{m-1}(x_{m}) - \Phi_0(x_1) \ge \eta \sum_{t = 0}^{m-1} h_t^{\top} x_t - \eta^2 \sum_{t = 0}^{m-1} \|h_t\|_{x_t}^2 .\]
    Since $x_{m}$ minimizes $\Phi_{m-1}(\cdot)$, we have $\Phi_{m-1}(w) \geq \Phi_{m-1}(x_m)$ and therefore 
    \[
    \eta \sum_{t = 0}^{m-1} h_t^{\top} w + r(w) - r(x_1) \ge \eta \sum_{t = 0}^{m-1} h_t^{\top} x_t - \eta^2 \sum_{t = 0}^{m-1}\|h_t\|_{x_t}^2. 
    \]
    This in turn implies  that
    \[\sum_{t=0}^{m-1} \langle \eta h_t, x_{t} - w \rangle \le r(w) - r(x_1) +\eta^2 \sum_{t=0}^{m-1} \|h_t\|_{x_t}^2 .
    \]
    Bounding $r(w) - r(x_1) \le \sup_{x \in D} r(x) - \inf_{y \in D} r(y)$ and 
    taking the expectation\footnote{More formally, by Tower law, we have $\E[\langle h_t, x_{t} - w \rangle] = \E[\E[\langle h_t, x_{t} - w \rangle|x_t]] = \E[\langle g_{x_t}, x_{t} - w \rangle]$.} over the randomness of  $h_0,\ldots, h_{m-1}$, we obtain the inequality in \eqref{ineq:expectation_bound}: \[
    \E\Big[\sum_{t=0}^{m-1} \langle \eta g_{x_t}, x_{t} - w \rangle \Big ] \le \sup_{x \in D} r(x) - \inf_{y \in D} r(y) + \eta^2 \sum_{t=0}^{m-1}  \E[\|h_t\|_{x_t}^2] . \]

    Now, to prove that inequality \eqref{x_t_ineq} is true, fix $T \in \{0, \ldots, m-1\}$. 
    Note that it suffices to show that $\Phi_{T-1}(x_{T+1}) - \Phi_{T-1}(x_{T}) \ge \eta h_T^{\top} (x_T - x_{T+1}) -\eta^2 \norm{h_T}_{x_T}^2$. Observe that
    \begin{align*} 
    \Phi_{T-1}(x_{T+1}) - \Phi_{T-1}(x_{T}) &= \langle \eta \sum_{t=0}^{T-1} h_t , x_{T+1} - x_T \rangle + r(x_{T+1}) - r(x_T)  \\ 
    & = \langle \eta \sum_{t=0}^{T-1} h_t + \nabla r(x_T), x_{T+1} - x_T \rangle + V_{x_T}(x_{T+1}) ,
    \end{align*}
    where we define $\sum_{t=0}^{T-1} h_t \defeq \vzero$ for $T = 0$. 
    Since $x_T$ minimizes $\Phi_{T-1}$ for any $T \in \{ 0,\ldots, m-1\}$, we have that $\langle \eta \sum_{t=0}^{T-1} h_t + \nabla r(x_T), x_{T+1} - x_T \rangle \ge 0$. Hence, we obtain \[\Phi_{T-1}(x_{T+1}) - \Phi_{T-1}(x_{T}) \ge V_{x_T}(x_{T+1}) .
    \]
    Using \Cref{lem:MP_proof_tweak} with $g = - \eta h_T$, we bound
    \[
    \langle \eta h_T, x_T - x_{T+1} \rangle - V_{x_T}(x_{T+1}) \le \|\eta h_T\|_{x_T}^2.
    \]
    Thus, we have \[
    \Phi_{T-1}(x_{T+1}) - \Phi_{T-1}(x_{T}) \ge \eta h_T^{\top} (x_T - x_{T+1}) - \eta^2 \norm{h_T}_{x_T}^2 .
    \]
    By moving the term $\eta h_T^\top x_{T+1}$ to the left side, this implies \eqref{x_t_ineq}
    \[\Phi_T(x_{T+1}) \ge \Phi_{T-1}(x_{T}) + \eta h_T^{\top} x_T - \eta^2 \norm{h_T}_{x_T}^2, \]
    which completes the proof of the inequality \eqref{ineq:expectation_bound}.

    Next, to prove inequality \eqref{ineq:ell_infinity_bound_gen} in the ``moreover'' part of the lemma, first note that as $p^{(t)}$ depends on randomness of $h_1, \cdots, h_{t-1}$ but not on $h_t$, we have for every $t = 0, \ldots, T-1$ that
    \[
    \E_t\left[-\eta \langle p^{(t)}, h_t \rangle + \eta^2 \norm{h_t}_{p^{(t)}}^2\right] = 
    -\eta \langle p^{(t)}, g_{x_t} \rangle + \eta^2 \E_t\left[\norm{h_t}_{p^{(t)}}^2\right]  
    \leq v^* . 
    \]
We can therefore apply \Cref{lem:stocastic_MW} with $g^{(t)} = - \eta h_t$ to obtain that for any $i \in V$, with probability at least $1 - \rho/n$,
    \[  \sum_{t = 0}^{m-1} - \eta (h_t)_i \le m v^* + \log \Big(\frac{n^2}{\rho} \Big).\]
    Then a union bound over all coordinates $i \in V$ completes the proof of the lemma.
\end{proof}

The following lemma explains the assumption that $\eta \norm{h_t}_{\infty} < 1/2$ for all $t \in \{0,\ldots, m-1\}$.

\begin{lemma}
\label{lem:MP_proof_tweak}
For any vector $g \in \R^{n}$ with $g_{i} \leq 1.79$ for all $i\in V$ and
$x,u \in \R_{> 0}^V$, then we have
\[
\langle g, u-x \rangle - V_{x}(u) \leq \| g\|_x^2 .
\]
\end{lemma}

\begin{proof}
Define concave function $f(z)$ for all $z\in\R^{n}$ as
\[
f(z)=\langle g, z-x \rangle - V_{x}(z)\,.
\]
Note that $\nabla f(z) = g-\nabla r(z)+\nabla r(x)=g - \log(z/x)$. Consequently, $z_{*} \defeq x\exp(g) \in\R_{>0}^{n}$ entrywise satisfies
that $\nabla f(z_{*})=0$. Therefore
\begin{align*}
\sup_{z\in\R_{>0}^{n}}f(z) & =f(z_{*})=\innerProduct g{z_{*}-x}-V_{x}(z_{*})\\
 & =\innerProduct g{z_{*}-x}-\innerProduct{z_{*}}{\log(z_{*}/x)}-\innerProduct{x-z_{*}}{\vones}\\
 & =\innerProduct{z_{*}-x}{\vones}-\innerProduct gx=\sum_{i\in V}x_{i}\left(\exp(g_{i})-1-g_{i}\right) ,
\end{align*}
where above we used \eqref{eq:breg_formula} in the second line. 
Finally, since $\exp(g_{i})\leq1+g_{i}+g_{i}^{2}$ as $g_{i}\leq1.79$,
the result follows. 
\end{proof}

\subsection{Stochastic Multiplicative Weights}

In this section, we present the Stochastic Multiplicative Weights algorithm and its analysis, which helps with the proof of \Cref{lem:stMD_guarantee}. First, we start with the algorithm pseudocode.

\begin{algorithm}[htp!]
\caption{Stochastic Multiplicative Weights}\label{alg:st_MW}
    \KwData{$\{g^{(0)}, \ldots , g^{(m-1)}\}$, the sequence of stochastic vectors used for the update steps, and the number of iterations $m$}
    \KwResult{$\{w^{(0)}, w^{(1)}, \ldots ,w^{(m)}\}$, the sequence of iterates generated}
    \SetKwFunction{stMW}{StochasticMW}
    \SetKwProg{Fn}{Function}{:}{}
    \Fn{\stMW{$\{g^{(1)}, \ldots, g^{(m)}\}$}}{
    $w^{(0)} = \onesVec \in \R^V$ \tcp*{$w^{(0)}$ is the initial point.}
    \For{$t \in \{0, 1, \ldots, m-1\}$}{
        $w^{(t+1)}_i \gets w^{(t)}_i \exp(g^{(t)}_i)$, for all $i \in V$\;
    }
    \Return $\{w^{(0)}, w^{(1)}, \ldots, w^{(m)}\}$\;
    }
\end{algorithm}

Now, we provide the guarantee on $\{w^{(t)}\}$ and $\{g^{(t)}\}$ that we need. 

\begin{lemma}\label{lem:stocastic_MW}
For each step $t \in \{0,\ldots, m-1\}$ of \Cref{alg:st_MW}, define vector $p^{(t)} \in \R^V$ with entries $p^{(t)}_i \defeq \frac{w^{(t)}_i}{\norm{w^{(t)}}_1}$. Let $v^*$ be a number such that with probability 1, 
\[
\E_t\left[\langle p^{(t)}, g^{(t)} \rangle + \norm{g^{(t)}}_{p^{(t)}}^2\right] \leq v^* , \enspace \forall t \in \{0, \ldots, m-1\} ,  
\]  
where the expectation $\E_t$ is taken over $g^{(t)}$ conditioned on $\{g^{(0)}, \ldots, g^{(t-1)}\}$. 
Suppose $\norm{g^{(t)}}_{\infty}\leq 1/2$ for all $t$ with probability $1$.  
Then for any $\rho > 0$ and $i \in V$, with probability at least $1 - \rho$, 
\[\sum_{t=0}^{m-1} g^{(t)}_i \le m v^* + \log \Big(\frac{n}{\rho} \Big) . 
\]
\end{lemma}

\begin{proof}

Since $\exp(\alpha) \le 1+\alpha + \alpha^2$ whenever $|\alpha| \le \frac{1}{2}$, for any step $t \in \{0, \ldots, m-1\}$ and $i \in V$, 
\[w^{(t+1)}_i \le w^{(t)}_i (1 + g^{(t)}_i + (g^{(t)}_i)^2) = \norm{w^{(t)}}_1 (p^{(t)}_i + p^{(t)}_i g^{(t)}_i + p^{(t)}_i (g^{(t)}_i)^2) ,
\]
where the equality follows from the definition $w^{(t)}_i = p^{(t)}_i \norm{w^{(t)}}_1$. 
Now we sum up the above for all $i \in V$, using that $\sum_{i \in V} p^{(t)}_i = 1$ and $w^{(t+1)} \in \R^V_+$, to obtain 
    \[\norm{w^{(t+1)}}_1 \le \norm{w^{(t)}}_1 \Big(1 + \langle p^{(t)}, g^{(t)} \rangle + \norm{g^{(t)}}_{p^{(t)}}^2 \Big) . 
    \]
Note that $w^{(t)}$ does not depend on the randomness of $g^{(t)}$. Therefore, taking expectation with respect to the randomness of $g^{(t)}$, conditioned on $g^{(1)}, \ldots, g^{(t-1)}$, we obtain 
\begin{align*} 
        \E_t\left[\norm{w^{(t+1)}}_1\right] \le \norm{w^{(t)}}_1 \E_{t}\left[1 + \langle p^{(t)}, g^{(t)} \rangle) + \norm{g^{(t)}}_{p^{(t)}}^2\right] \le \norm{w^{(t)}}_1 \exp(v^*) ,
\end{align*} 
where in the last inequality, we used $\E_{t}\left[\langle p^{(t)}, g^{(t)} \rangle) + \norm{g^{(t)}}_{p^{(t)}}^2\right] \leq v^*$ with probability 1, and that $1 + \alpha \leq \exp(\alpha)$ for any $\alpha \in \R$. 
Further taking the expectation and unravel the above yields
\begin{align} \label{ineq:MW_expectation}
\E\left[\norm{w^{(m)}}_1\right] \leq \E\left[\norm{w^{(m-1)}}_1 \right] \cdot \exp(v^*) \leq \norm{w^{(0)}}_1 \cdot \exp(m v^*) = n \exp(m v^*) ,
\end{align}
where the last equality is because $w^{(0)} = \onesVec$. 
Next, note that for any coordinate $i \in V$, we have 
\[
\E\left[\norm{w^{(m)}}_1\right] \geq \E\left[w^{(m)}_i \right] = \E\left[w^{(0)}_i \exp \left(\sum_{t=0}^{m-1} g_i^{(t)} \right) \right] = \E\left[\exp \left(\sum_{t=0}^{m-1} g_i^{(t)} \right) \right] , 
\]
where the last equality uses $w_i^{(0)} = 1$. 
It follows from Markov's inequality and \eqref{ineq:MW_expectation} that
\[
\exp \left(\sum_{t=0}^{m-1} g_i^{(t)} \right) \leq \frac{n}{\rho} \exp(m v^*) , \, \enspace \text{with probability } \geq 1 - \rho . 
\]
Taking the logarithm on both sides completes the proof of the lemma. 
\end{proof}

\section{Proximal Step Over $k$-simplex}
\label{sec:proximal}

In this subsection, we show how to take a proximal step efficiently. The rest of this section is devoted to proving \Cref{lem:prox_step}. 
Throughout this section, we identify $V$ with $[n]$ as the set of coordinates to ease the comparison between elements in $V$. 

\begin{lemma}[Proximal step] \label{lem:prox_step}
For any vectors $x_0 \in \mathbb{R}_{> 0}^V \cap S_k^V$, $h \in \mathbb{R}^V$, and positive integer $k$, the proximal step 
\[
z \defeq \argmin_{x \in S_k^V} h^\top x + V_{x_0}(x)
\]
can be computed in $O(1)$-depth and $O(k n)$ time. Moreover, for any $i, j \in V$, $z_i \ge z_j$ if and only if either $\log x_{0,i}- h_i  > \log x_{0,j} - h_j$ or $\log x_{0,i}- h_i  \geq \log x_{0,j} - h_j$ and $i < j$.  
\end{lemma}

\begin{proofof}{\Cref{lem:prox_step}}
The procedure for computing a proximal step in \Cref{lem:prox_step} is given in \Cref{alg:prox_step}. 
From the algorithmic description, the algorithm can be implemented in $O(1)$-depth and $O(kn)$ time by implementing the for loop in parallel. 
The rest of this proof is devoted to the correctness of \Cref{alg:prox_step}.

\begin{algorithm}[htp!]
\caption{Proximal step over $k$-simplex}\label{alg:prox_step}
    \KwData{$x_0 \in \mathbb{R}_{> 0}^V \cap S_k^V$, $h \in \mathbb{R}^V$, $k \in \mathbb{Z}_{>0}$}
    \KwResult{$z = \argmin_{x \in S_k^V} h^\top x + V_{x_0}(x)$, where $S_k^V = \{x \in [0, 1]^V: \|x\|_1 \le k\}$}
    \SetKwFunction{proxstep}{ProximalStep}
    \SetKwProg{Fn}{Function}{:}{}
    \Fn{\proxstep{$x_0,h,k$}}{
    $y_j \leftarrow x_{0,j} \cdot e^{-h_j}$ for all $j \in V$\; 
    \label{line:y_def_prox}
    Let $\pi: [n] \rightarrow V$ be permutation sorting the coordinates in decreasing order of $y$, breaking ties alphabetically\;
    \For(\tcp*[f]{$i$ is the number of coordinates of $z$ that are $1$}){$i = 1, \ldots, k-1$ \label{line:for_loop_prox}}{
        Set $\lambda^* \leftarrow \max\{0, \log(\sum_{j=i+1}^n y_{\pi(j)} / (k-i))\}$ \;
        \If{$y_{\pi(i+1)} e^{-\lambda^*} \leq 1$ \label{line:testing_i}}{
        $z_{\pi(j)} \leftarrow 1$ for all $j \leq i$\;
        $z_{\pi(j)} \leftarrow y_{\pi(j)} e^{-\lambda^*}$ for all $j > i$\;
        \Return $z$\; 
        }
        }
    }
\end{algorithm}

The following equivalent view of the proximal step is standard 
\begin{align*}
z &= \argmin_{x \in S_k^V} h^\top x + V_{x_0}(x) = \argmin_{x \in S_k^V} h^\top x + r(x) - \nabla r(x_0)^\top x \\
& = \argmin_{x \in S_k^V} r(x) - (\nabla r(x_0) - h)^\top x \defeq \argmin_{x \in S_k^V} V_y(x) ,
\end{align*}
where $y \in \mathbb{R}_{> 0}^V$ in \Cref{line:y_def_prox} is defined to satisfy $\nabla r(y) = \nabla r(x_0) - h$. This definition is further equivalent to $y_i \defeq x_{0,i} \cdot e^{- h_i}$ for all $i \in V$. 

To reason about $x_* = \argmin_{x \in S_k^V} V_y(x)$, we use the KKT conditions. 
Using \eqref{eq:breg_formula} (with the parameters $x$ and $y$ swapped), we have
\begin{align*}
V_y(x)  = \sum_{i \in V} x_i \log(x_i / y_i) - \sum_{i \in V} x_i + \sum_{i \in V} y_i .
\end{align*}
Note that the last term $\sum_{i \in V} y_i$ is independent of the variables $x$ that we are optimizing over.
Therefore, the point $z = \argmin_{x \in S_k^V} V_y(x)$ is the solution to the following convex program
\begin{align*}
    \min_x \quad &\sum_{i \in V} x_i \log(x_i / y_i) - \sum_{i \in V} x_i \\
    \text{s.t.} \quad & x_i \in [0,1]^V, \quad \sum_{i \in V} x_i \leq k . 
\end{align*}
Let us ignore the constraint $x_i \geq 0$ for now (as we shall see, $x_i > 0$ will be satisfied automatically). We introduce Lagrangian multipliers $\mu_i \geq 0$ for each constraint $x_i \leq 1$ and $\lambda$ for the constraint $\sum_i x_i \leq k$.
Consider the Lagrangian dual
\begin{align*}
L(x, \mu, \lambda) = \sum_i x_i \log(x_i / y_i) - \sum_i x_i + \sum_i \mu_i (x_i - 1) + \lambda (\sum_i x_i - k) .
\end{align*}
Strong duality holds for this program, so KKT condition $\nabla_x L(z,\mu^*,\lambda^*) = 0$ gives
\[
\log(z_i / y_i) + \mu_i^* + \lambda^* = 0 ,
\]
where $\mu^*, \lambda^*$ are the optimal Lagrangian multipliers. 
This implies that
\begin{align} \label{eq:kkt_stability}
z_i = y_i \cdot \exp( - \mu_i^* - \lambda^*) > 0. 
\end{align}
Note the other KKT conditions are (1) primal and dual feasibility:  $z_i \leq 1$, $\sum_i z_i \leq k$, $\mu_i^*, \lambda^* \geq 0$ and (2) complementary slackness: $\mu_i^* (z_i - 1) = 0$ and $\lambda^* (\sum_i z_i - k) = 0$. 

Now, note that each iteration of the for loop in \Cref{line:for_loop_prox} tries one of these $i$ values. Once $i$ is fixed, we compute the project of the remaining coordinates of $y$ onto $\{z \geq 0: \sum_{j > i} z_j \leq n-k\}$. If the projected point $z$ has any coordinate $z_j > 1$, this means our guess is wrong, and we proceed to increase $i$ by $1$; otherwise, our guess for $i$ is correct (the projected point $z \in S_k^V$) and the point $z$ is also the correct point after the proximal step. 
To show correctness of \Cref{alg:prox_step}, it suffices to show that $z^* = \argmin_{x \in S_k^V} V_y(x)$  can indeed be constructed by this procedure. To do this, we use the following key structural properties of the proximal step $z$.

\begin{claim} \label{claim:proj_property}
Let $z^* = \argmin_{x \in S_k^V} V_y(x)$. For any coordinates $j,\ell \in V$ such that $y_j > y_\ell$, we have $z_j \geq z_\ell$. Moreover, if $j < \ell$ and $y_j \ge y_\ell$, we have $z_j \geq z_\ell$.
\end{claim}

Before presenting the proof of \Cref{claim:proj_property}, we first use it to finish the proof of the correctness of the algorithm. 
Note that \Cref{claim:proj_property} implies that the set of $z_j = 1$ must correspond to a set of largest $i$ coordinates of $y$, where $i \leq k$ since $\|z\|_1 \leq k$. Next, that once we set the values $\mu_i$ for $i \in V$, the choice of $\lambda$ for respecting the KKT conditions is unique (if it exists). Hence, our procedure correctly attempts to increase the $\mu_j$ corresponding to the largest $i$ coordinates first, and then find a value of $\lambda$ that satisfies the KKT conditions. 
Thus, to finish the proof, it suffices to show that \Cref{line:testing_i} in \Cref{alg:prox_step} is activated at least once, as the vector returned would satisfy the KKT conditions, or, equivalently, that trying to setting only the first $i$ values of $\mu_{\pi(j)}$ to be positive (for some $i \le k$) is enough. To see why this is true, suppose that it is never activated. In particular, this implies $y_{\pi(k)} > 1$, which in turn implies that $\sum_{j \ge k} y_{\pi(j)} > 1$, which yields \[y_{\pi(k)} e^{\log (\sum_{j \ge k} y_{\pi(j)})} = \frac{y_{\pi(k)}}{\sum_{j \ge k} y_{\pi(j)}} < 1,\]
which yields a contradiction. 

Lastly, the second part of the statement clearly follows from the description of our procedure. 
\end{proofof}

\begin{proofof}{\Cref{claim:proj_property}}
Let $j, l$ be coordinates with $y_j > y_\ell$ and assume for the purpose of contradiction that $z_j < z_\ell$.
Note that this implies, by the complementary slackness, that $z_j < 1$ and thus $\mu_j^* = 0 \leq \mu_\ell^*$. 
Now by \eqref{eq:kkt_stability}, we have 
\[
z_j = y_j \cdot \exp( - \mu_j^* - \lambda^*) > y_\ell \cdot \exp( - \mu_\ell^* - \lambda^*) = z_\ell ,
\]
which leads to a contradiction. This proves the claim. 
\end{proofof}

\subsection*{Update Step for \Cref{st_FTRL}}

In this section, we present the pseudocode and runtime guarantee for one update step for \Cref{st_FTRL}. Our main result is the following:

\begin{lemma}[Update step for follow-the-regularized-leader] \label{lem:FTRL_update}
For any vectors $h_1, h_2, \ldots, h_l \in \mathbb{R}^V$, and positive integer $k$, the proximal step 
\[
z = \argmin_{x \in S_k^V} \sum_{t \in [l]} h_t^\top x + r(x)
\]
can be computed in $O(1)$-depth and $O(k n)$ time. Moreover, for any $i, j \in V$, $z_i \ge z_j$ if and only if either $- \sum_{t \in [l]} [h_t]_i > - \sum_{t \in [l]} [h_t]_j$ or $- \sum_{t \in [l]} [h_t]_i \ge - \sum_{t \in [l]} [h_t]_j$ and $i < j$. 
\end{lemma}

The procedure for computing a proximal step in \Cref{lem:FTRL_update} is given in \Cref{alg:proj_step_FTRL}.

\begin{algorithm}[H]
\caption{Update step over $k$-simplex $S_k^n$ for FTRL}
\label{alg:proj_step_FTRL}
\KwData{$h_1, \ldots, h_l \in \mathbb{R}^V$, $k \in \mathbb{Z}_{>0}$}
\KwResult{$z = \argmin_{x \in S_k^V} \sum_{t \in [l]} h_t^\top x + r(x)$, where $S_k^V = \{x \in [0, 1]^V: \|x\|_1 \le k\}$}
\SetKwFunction{ftrlupdate}{FTRLUpdate}
\SetKwFunction{proxstep}{ProximalStep}
\SetKwProg{Fn}{Function}{:}{}
\Fn{\ftrlupdate{$\{h_1, \ldots, h_l\}, k$}}
{
$x_0 \gets \frac{k}{n} \Vec{1}$\;
$h \gets \sum_{t \in [l]} h_t + \nabla r(x_0)$\;
\Return \proxstep{$x_0, h, k$}
}
\end{algorithm}

\begin{proof}
Note that \Cref{alg:proj_step_FTRL} is an instantiation of \Cref{alg:prox_step} where $h = \sum_{t \in [l]} h_t + \nabla r(x_0)$, where $x_0 \defeq \frac{k}{n} \Vec{1}$ and use \Cref{lem:prox_step}. 
\end{proof}